\documentclass[letterpaper,11pt]{article}
\usepackage{amsmath}
\usepackage{amssymb}
\usepackage{amsthm}
\usepackage{amsfonts}
\usepackage{setspace}
\usepackage[bottom]{footmisc}
\usepackage{bbm}
\usepackage{color}
\usepackage{enumerate}
\usepackage{mathrsfs}
\usepackage{graphicx}
\usepackage{subcaption}
\usepackage{hyperref}
\usepackage{enumitem}
\usepackage[usenames,dvipsnames,svgnames,table]{xcolor}
\usepackage{multirow}
\usepackage[margin=1.10in]{geometry}
\usepackage{tikz}
\usetikzlibrary{fit,positioning,arrows,automata}
\usepackage{enumitem}
\setlist{nosep}

\newtheorem{theorem}{Theorem}[section]
\newtheorem{lemma}{Lemma}[section]
\newtheorem{proposition}{Proposition}[section]

\newtheorem{remark}{Remark}[section]
\makeatletter
\def\thm@space@setup{%
  \thm@preskip=16pt
  \thm@postskip=5pt 
}
\makeatother

\newtheoremstyle{theoremd}
  {\topsep}
  {\topsep}
  {\itshape}
  {0pt}
  {\bfseries}
  {}
  { }
  {\thmname{#1}\thmnumber{ }\thmnote{ (#3)}}

\theoremstyle{theoremd}

\newcommand{\mb}[1]{\mathbb{#1}}

\newcommand{\mf}[1]{\mathbf{#1}}
\newcommand{\mr}[1]{\mathrm{#1}}
\newcommand{\mcr}[1]{\mathscr{#1}}
\newcommand{\mc}[1]{\mathcal{#1}}

\newcommand{\ind}{1\!\mathrm{l}}
\newcommand{\ol}[1]{\overline{#1}}
\newcommand{\ul}[1]{\underline{#1}}
\newcommand{\bs}[1]{\boldsymbol{#1}}

\hypersetup{
     colorlinks   = true,
     citecolor    = Blue,
     urlcolor     = Black,
     linkcolor    = Blue
}

\makeatletter
\renewcommand\paragraph{\@startsection{paragraph}{4}{\z@}%
                                    {0pt \@plus1ex \@minus.2ex}%
                                    {-1em}%
                                    {\normalfont\normalsize\bfseries}}
\makeatother

\usepackage{bibunits}
\usepackage[longnamesfirst]{natbib}

\begin{document}

\defaultbibliography{tp}
\defaultbibliographystyle{chicago}
\begin{bibunit}

\author{%
{Timothy M. Christensen\thanks{%
Department of Economics, New York University, 19 W. 4th Street, 6th floor, New York, NY 10012, USA. E-mail address: \texttt{timothy.christensen@nyu.edu}}
}
} 

\title{%
Dynamic Models with Robust Decision Makers:\\ Identification and Estimation\footnote{
I thank several seminar and conference audiences for comments and suggestions. I am grateful to A. Bhandari, T. Bollerslev, J. Borovi\v{c}ka, G. Chamberlain, T. Cogley, C. Gourieroux (discussant), L.P.  Hansen, C. Ilut, E. Mammen, J. Nesbit, T. Sargent, and B. Sz\H{o}ke for helpful comments.}
}

\date{\today. }

\maketitle

\begin{abstract}  
\singlespacing
 \noindent
This paper studies identification and estimation of a class of dynamic models in which the decision maker (DM) is uncertain about the data-generating process. 
The DM surrounds a benchmark model that he or she fears is misspecified by a set of models. Decisions are evaluated under a worst-case model delivering the lowest utility among all models in this set. 
The DM's benchmark model and preference parameters are jointly underidentified. 
With the  benchmark model held fixed, primitive conditions are established for identification of the DM's worst-case model and preference parameters.  
The key step in the identification analysis is to establish existence and uniqueness of the DM's continuation value function allowing for unbounded statespace and unbounded utilities.
To do so, fixed-point results are derived for monotone, convex operators that act on a Banach space of thin-tailed functions arising naturally from the structure of the continuation value recursion. 
The fixed-point results are quite general; applications to models with learning and Rust-type dynamic discrete choice models are also discussed.
For estimation, a perturbation result is derived which provides a necessary and sufficient condition for consistent estimation of continuation values and the worst-case model. The result also allows convergence rates of estimators to be characterized.
An empirical application studies an endowment economy where the DM's benchmark model may be interpreted as an aggregate of experts' forecasting models. The application reveals time-variation in the way the DM pessimistically distorts benchmark probabilities. Consequences for asset pricing are explored and connections are drawn with the literature on macroeconomic uncertainty.

\medskip 


\medskip


\end{abstract}

\pagenumbering{arabic}

\newpage

\section{Introduction}

A large and active literature explores the implications for individual decision making and policy design under model uncertainty or ambiguity, 
building on the decision-theoretic foundations of \cite{GS}, \cite{HS2001-ack,HS2001}, \cite{EpsteinSchneider}, \cite{KMM2005,KMM2009}, \cite{MMR}, and \cite{Strzalecki}. 
Various applications include monetary and fiscal policy design  \citep{Giannoni,OnatskiStock,CCHS,Woodford,Karantounias}, portfolio choice and asset allocation \citep{HST,BHS,EpsteinSchneiderAR,HS2017sets}, dynamic contracting \citep{MiaoRivera}, sovereign default \citep{PouzoPresno}, climate policy \citep{Xepapadeas,BrockHansen}, and understanding household and professional forecast survey data \citep{BBH,Szoke}.

This paper explores some issues regarding the econometrics of these models. In particular, we study identification and estimation of a class of dynamic models with a single decision maker (DM) who is uncertain about the data-generating process. We will deal mostly with environments in which the DM has multiplier or constraint preferences as in the ``robustness'' literature pioneered by Hansen and Sargent (see \cite{HS2008} and references therein), though extensions to some other classes of preferences will also be discussed. In this setting, the DM's decision problem may be summarized as follows. The DM has a benchmark model of the economy that he or she fears may be misspecified. The DM surrounds the benchmark model by a set consisting of all models whose discounted Kullback--Leibler discrepancy relative to the benchmark model does not exceed some threshold. Decisions are evaluated under a worst-case model that delivers lowest utility among all models within this set, as in the multiple prior framework of \cite{GS} and \cite{EpsteinSchneider}.\footnote{For a DM with multiplier preferences, a relative entropy penalty is instead appended to the DM's continuation value recursion. Nevertheless, the DM's behavior retains an ex-post interpretation that decisions are optimal under a worst-case model in a Kullback--Leibler neighborhood of the benchmark model.} The DM's fear of misspecification induces a wedge between the probability measure under which decisions are evaluated and the data-generating probability measure. In contrast, the two probability measures agree in conventional rational expectations models. This wedge must be accounted for when attempting to identify agents' preference parameters.

Given the dynamic nature of the DM's problem, the worst-case model is that which lowers the DM's continuation value the most. The worst-case model and the DM's continuation value are pinned down jointly, by a particular nonlinear fixed point equation. This adds a further layer of complexity that must be dealt with when identifying model primitives and developing estimation and inference procedures. 

Robust decision rules and robust policies depend implicitly on the benchmark model. 
To date the literature has, with few exceptions, specified tightly-parameterized linear-Gaussian benchmark models. This is largely for the sake of analytic tractability, as it is one of the few instances where the DM's continuation value and worst-case model can be solved for in closed form, at least for certain specifications of the DM's period utility function. 
While analytically tractable, simple linear-Gaussian specifications often induce worst-case models that are time-invariant in the sense that they do not respond to fluctuations in state variables \citep{BHS,BS,BBH}, and consequently do not deliver time-variation in prices of risk/uncertainty \citep{HS2017sets}.\footnote{\cite{Sims} also raised concerns as to whether the focus on simple linear models overlooks other, potentially more important aspects of model uncertainty. Moreover, the literature on nonlinear dynamic stochastic general equilibrium models has also emphasized the quantitative importance of differences between nonlinear models and log-linear approximations. See, e.g., \cite{FVRR} and references therein.} More recently, increasing emphasis has been placed on more elaborate nonlinear benchmark specifications or richer preference structures incorporating preference shocks, learning, or multiple layers of uncertainty. With more complicated benchmark models and preferences, analytic tractability may be lost and issues of existence and uniqueness, model identification, and empirical implementation become more opaque.

Prompted by these issues, this paper attempts to make progress on several questions, namely:  Are there general conditions for existence and uniqueness of continuation values that do not rely on overly restrictive benchmark specifications? What features of model primitives might be identified in general nonlinear Markovian settings? What is required to estimate these models in such settings? We address these questions as follows.

First, we study the identification of model primitives within a class of dynamic models featuring a single agent who solves an infinite-horizon robust decision problem. We allow for general nonlinear Markovian environments. The DM's benchmark model and preference parameters are jointly underidentified, even when the worst-case model is fully known. With the benchmark model held fixed, nonparametric identification of the worst-case model and local identification of the DM's preference parameters are established. The key regularity condition is a very mild condition on the distribution of utility growth, which can be easily verified. No further function-analytic conditions, such as compactness, are required.

A key step in the identification analysis is to establish existence and uniqueness of the DM's continuation value function. The DM's preference for robustness induces a nonlinear adjustment to the continuation value recursion. As a consequence, the recursion is not a contraction mapping when the value function is allowed to be unbounded, which it is in almost all settings.\footnote{For instance, in linear-Gaussian settings the statespace is unbounded and the  value function is affine in the state variable.}  
Existence and uniqueness of the value function is established by repurposing some tools from the modern statistics literature. The analysis is conducted within a Banach space of unbounded but ``thin-tailed'' functions that arises naturally from the structure of the recursion, specifically an exponential Orlicz class used in empirical process theory \citep{vdVW} and modern high-dimensional probability theory \citep{Vershynin}.\footnote{A special case among the class also has connections with information geometry \citep{PistoneSempi} and exponential tilting \citep{Csiszar1995,KomunjerRagusa}.} Monotonicity and convexity properties of the recursion are leveraged to establish existence. Establishing uniqueness requires ensuring that the conditional expectation operator associated with the DM's worst-case model does not move probability mass too far relative to the effect of discounting. Tail inequalities bounding the probabilities of large deviations of thin-tailed random variables are used for this step. Under restrictions on utilities, the value function recursion is isomorphic to that under Epstein--Zin--Weil (EZW) recursive utility and unit intertemporal elasticity of substitution (IES).  The existence and uniqueness results therefore apply equally to such models.  The existence and uniqueness results are leveraged to establish identification of the agent's worst-case model and preference parameters.

As a byproduct, a general existence and uniqueness result is derived for fixed points of monotone, convex operators on classes of unbounded functions.\footnote{See \cite{BorovickaStachurski} for related results for classes of bounded functions with an emphasis on models with EZW recursive preferences.} Its proof is constructive, using only a few basic results from the theory of integration. The result appears well suited to study existence and uniqueness of value functions in models with forward-looking agents more generally. To illustrate its usefulness, existence and uniqueness of value functions is established in two further applications. The first is models featuring a robust DM who learns about hidden states as in \cite{HS2007,HS2010}, which nests models with EZW recursive utility and learning as well as other models of ambiguity studied by \cite{KMM2009} and \cite{JuMiao}. The second application is dynamic discrete choice models \citep{Rust1987} allowing for unbounded utilities and continuous unbounded statespace. 

Second, we derive a set of perturbation results characterizing how the continuation value and worst-case model change as the benchmark model changes. The results provide a \emph{necessary and sufficient} condition for the value function in the perturbed model to converge to the value function in the original model as the  perturbation shrinks to zero. 
These results have several uses. Consider estimating the value function and worst-case model by first estimating the benchmark model (say, from time-series data on state variables or survey data) then solving the continuation value recursion under the estimated model. The perturbation results may be applied to establish consistency and convergence rates of estimators of the value function and worst-case model based on this ``plug-in'' procedure, treating the estimated model as a perturbation of the truth. The results also permit computation of approximate value functions in models where no closed-form solution exists by perturbing models where closed-form solutions do exist. As an example, it is shown how to compute approximate value functions in nonlinear environments with stochastic volatility by perturbing linear-Gaussian environments. The result may also be used to derive influence functions of plug-in estimators of various asset pricing functionals.

Third, we consider an empirical application similar to \cite{BHS} (see also \cite{HHL} and \cite{BS}). In contrast with earlier works, we specify the benchmark model as a covariate-dependent mixture of Gaussian vector autoregressions. This approach has several appealing features. In particular, it has a very natural and intuitive interpretation as a ``mixture of experts'' where each ``expert'' is summarized by a vector autoregression. The weights the DM assigns to each expert's forecast vary in a natural way with the state variables. Variation in the mixing weights generates nonlinearities in the conditional mean and conditional variance, which will be seen to generate important asset-pricing implications. In addition, this specification nests conventional linear-Gaussian models as a special case, making it well-suited to conduct a sensitivity analysis of departures from linearity and Gaussianity.  

The empirical findings are summarized briefly as follows. The time series of the realized change of measure between the benchmark and worst-case model is extracted and is seen to be volatile and counter-cyclical. The worst-case model pessimistically shifts mass towards regions of low consumption growth. Whereas the worst-case model in linear-Gaussian settings is time-invariant, here there is time-variation in the way the DM distorts his or her benchmark model to obtain the worst case. In particular, the DM's worst-case model features a much fatter left tail for consumption growth in ``bad'' economic states than in ``good'' states.  Time-variation in the wedge between the benchmark and worst-case models generates time-variation in term structures of prices of risk/uncertainty which we explore. Further connections with the literature on macroeconomic uncertainty are also drawn.

The remainder of the paper is as follows. Section \ref{s:fw} describes the class of models under consideration. Section \ref{s:identification} presents the identification results for continuation values, preference parameters, and the underidentification result. Section \ref{s:learn} extends the existence and uniqueness results for value functions to models in which the DM is learning about hidden states. Section \ref{s:estimation} presents perturbation results and applies these to estimation.  Finally, Section \ref{s:app} presents the empirical application. Appendix \ref{ax:orlicz} contains background material on Orlicz classes, Appendices \ref{ax:id:gen} and \ref{ax:id} contain additional results for identification, and  Appendix \ref{ax:ddc} presents results for Rust-type dynamic discrete choice models.

\newpage

\section{Framework}\label{s:fw}

This section describes the setup in a single-agent setting. Much of this section is a highly stylized summary of material in \cite{HS2008} to fix ideas and notation.  Extensions to  models with learning and other forms of ambiguity aversion are discussed in Section \ref{s:learn}.

\subsection{Environment}

Consider a discrete-time, infinite-horizon environment. Let $T$ denote the set of non-negative integers. At each date $t \in T$, the DM chooses a vector of controls $C_t \in \mc C_t$ (a constraint set). The source of risk is a time homogeneous, controlled Markov process $X = \{X_t : t \in T\}$ taking values in $\mc X \subseteq \mb R^d$.  There is a conditionally deterministic state process $Z = \{Z_t : t \in T\}$ where $Z_t$ characterizes evolution of variables used to describe the constraint set (e.g. wealth or capital). The process $Z$ has law of motion $Z_{t+1} = z(C_t,Z_t,X_t,X_{t+1})$ known to the DM. 

\subsection{Preferences}

First consider a DM with {\it multiplier preferences}, as introduced by \cite{HS2001} and axiomatized by \cite{Strzalecki}. The DM's preference parameters are $(Q,U,\beta,\theta)$, where $U$ is the DM's period utility function, $\beta \in (0,1)$ is a time preference parameter, and $\theta>0$ a risk-sensitivity parameter. It is assumed throughout that $\theta$ is fixed, but the following analysis may  be extended to accommodate environments in which $\theta$ is state dependent or, more generally, a stationary stochastic process as in \cite{BBH}. The DM's benchmark model for evolution of $X$ is described by a Markov kernel $Q(\cdot|x,c)$ representing the conditional distribution of $X_{t+1}$ given $X_t = x$ and $C_t = c$.  Let $\mb E^Q$ denote conditional expectation under the DM's benchmark model. Let $\mc F_t$ denote the DM's information set at date $t$ and let $\mc M_{t+1}$ denote the set of all $\mc F_{t+1}$-measurable random variables $m_{t+1}$ with $m_{t+1} \geq 0$ (almost surely) and $\mb E^Q[m_{t+1}|X_t,C_t] = 1$. Each $m_{t+1} \in \mc M_{t+1}$ is a Radon--Nikodym derivative that induces a (conditional) probability measure that is absolutely continuous with respect to the benchmark model.  The DM's date-$t$ continuation value $V_t$ is defined by the recursion
\begin{align} 
 V_t &  = \max_{C_t \in \mc C_t} \min_{m_{t+1} \in \mc M_{t+1}} \bigg( U(C_t,X_t) + \beta \mb E^Q \Big[ m_{t+1} \Big( V_{t+1} + \theta \log m_{t+1}\Big) \Big| X_t,C_t \Big] \bigg) \label{e-V-0} \\
 \mbox{s.t. } Z_{t+1} & = z(C_t,Z_t,X_t,X_{t+1}) \,. \notag 
\end{align}
 The term $\beta \theta \mb E^Q [ m_{t+1} \log m_{t+1} | X_t,C_t ]$ penalizes the Kullback--Leibler (KL) divergence between the alternate model induced by $m_{t+1}$ and the benchmark model $Q$. As $\theta$ increases, distortions away from $Q$ become increasingly costly. In the limit as $\theta \to \infty$, multiplier preferences approach expected utility preferences.

Let $C_t^*$ denote the DM's optimal control at date $t$. The DM's worst-case model is induced by the change of measure  
\begin{align} \label{e-m-V}
 m_{t+1}^* = \frac{ e^{-\theta^{-1}V_{t+1}}}{ \mb E^Q[e^{-\theta^{-1} V_{t+1}}|X_t,C_t^*]} \,,
\end{align}
(see, e.g., \cite{HS2008}) which will be referred to as the \emph{worst-case belief distortion}. The worst-case model assigns relatively more weight to events that reduce the DM's continuation value and relatively less weight to events that increase the DM's continuation value. 
Substituting the worst-case distortion into (\ref{e-V-0}) yields the recursion
\begin{align} \label{e-V-recur-1}
 V_t = U(C_t^*,X_t) - \beta \theta \log \mb E^Q[e^{-\theta^{-1} V_{t+1}}|X_t,C_t^*]\,.
\end{align}

Closely related to multiplier preferences are \emph{constraint preferences}. 
Let $Q$, $U$, and $\beta$ be as above. Also let $\mc C$ denote the constraint set for the sequence $C_0,C_1,\ldots$ and $\mc M$ be the set of all sequences $m_0,m_1,\ldots$ of Radon-Nikodym derivatives as defined above. Finally, let $M_{t+1} = M_t m_{t+1}$ for each $t \geq 0$ with $M_0 = 1$. The DM's date-$0$ problem is:
\begin{align}
 V_0 & = \max_{\{C_t\}_{t \in T} \in \mc C} \min_{\{m_{t+1}\}_{t \in T} \in \mc M} \mb E \bigg[ \sum_{t=0}^\infty M_t \beta^t U(C_t,X_t)  \bigg| X_0 ,C_0 \bigg]  \label{e-entropy-constraint} \\
 \mbox{s.t. } &  \phantom{==} \sum_{t=0}^\infty \beta^{t+1} \mb E \Big[ M_t  \mb E[m_{t+1} \log m_{t+1} | \mc F_t] \Big| X_0 \Big] \leq \gamma   \notag
 \,.
\end{align}
The constraint in (\ref{e-entropy-constraint}) makes clear the sense in which the DM is maximizing  worst-case utility over a set of models: these are all models absolutely continuous with respect to $Q$ and whose discounted Kullback--Leibler discrepancy relative to $Q$ are no larger than $\gamma$. The DM's preference parameters consist of $(Q,U,\beta,\gamma)$. 

By introducing an additional control referred to as {continuation entropy},  \cite{HSTW} show that constraint preferences may be studied in the same way as multiplier preferences with $\theta$ reinterpreted as a Lagrange multiplier on the model set in (\ref{e-entropy-constraint}). The continuation entropy at date $t$, denoted $\Gamma_t$, is defined recursively by
\begin{equation*} 
 \Gamma_t = \beta \mb E^Q \Big[  m_{t+1}^* (\Gamma_{t+1} +  \log m_{t+1}^*) \Big| X_t \Big] 
\end{equation*}
with $\Gamma_0 = \gamma$. Thus, $\Gamma_t$ is effectively the size of the neighborhood in the DM's date-$t$ problem. Continuation entropy provides a link between $\theta$ and $\gamma$.  Although multiplier and constraint preferences are observationally equivalent, they induce different orderings over sequences $\{C_t\}_{t \in T}$. \cite{Strzalecki} discusses the connection between multiplier and constraint preferences and other classes of preferences in decision theory.

\subsection{Accommodating non-stationarity state variables}

Growth in the conditionally deterministic state variables may be accommodated under the following mild condition.

\medskip

\paragraph{Condition S} \emph{(i) There exist $v : \mc X \to \mb R$ and $u : \mc X^2 \to \mb R$ such that
\begin{align*}
 v(X_t) & = -\frac{1}{\theta} \left( V_t - \frac{1}{1-\beta} U(C_t^*,X_t)  \right) \,, &
 u(X_t,X_{t+1}) & =  U(C_{t+1}^*,X_{t+1}) - U(C_t^*,X_t)  \,;
\end{align*}
(ii) $X$ is a (strictly) stationary and ergodic, first-order Markov process under $Q(\cdot|X_t,C_t^*)$.
}

\medskip

Condition S is maintained throughout the paper. In stationary environments where the DM's optimal choice is a Markov policy $C_t^* = C^*(X_t)$ then Condition S is without loss of generality. 
In nonstationary environments, Condition S(i) is a homotheticity condition which allows the continuation value to be reformulated in terms of a scaled continuation value $v$ depending on the stationary process $X$ alone. Condition S(ii) is a stationarity condition. Section \ref{s:example} verifies Assumption S in a workhorse model featuring stochastic growth. 

In what follows, with some abuse of notation we write $Q(\cdot|X_t) = Q(\cdot|X_t,C_t^*)$. We also let $Q_0$ denote the stationary distribution of $X_t$ and $Q_0 \otimes Q$ denote the stationary distribution of $(X_t,X_{t+1})$.

Under Condition S, it follows from equations (\ref{e-m-V}) and (\ref{e-V-recur-1}) that $v$ solves the recursion
\begin{equation} \label{e-recur-v}
 v(X_t) = \beta \log  \mb E^Q \left[ \left. e^{ v(X_{t+1}) + \alpha u (X_t,X_{t+1}) } \right|X_t \right] \,,
\end{equation}
where 
\[
 \alpha = -\frac{1}{\theta(1-\beta)}\,.
\]
Note that there is a one-to-one correspondence between $(\theta,\beta)$ and $(\alpha,\beta)$. The worst-case distortion may be expressed in terms of $v$ and $u$ as
\begin{equation} \label{e-m-v}
 m_{t+1}^* = m_v(X_t,X_{t+1}) = \frac{e^{ v(X_{t+1}) + \alpha u (X_t,X_{t+1} )}}{\mb E^Q[e^{ v(X_{t+1}) + \alpha u (X_t,X_{t+1} )}|X_t]} \,,
\end{equation}
and the continuation entropy is given by $\Gamma_t = \Gamma(X_t)$ where $\Gamma$ solves 
\begin{equation} \label{e-recur-R}
 \Gamma(X_t) = \beta \mb E^Q \Big[  m_{t+1}^* (\Gamma(X_{t+1}) +  \log m_{t+1}^*) \Big| X_t \Big] 
\end{equation}
with $\Gamma(X_0) = \gamma$. Equations (\ref{e-recur-v}), (\ref{e-m-v}) and (\ref{e-recur-R}) will be used extensively in what follows.

\subsection{Running example}\label{s:example}

We use a simple example to illustrate the setup and foreshadow some identification issues that arise. 
Consider an exchange economy similar to that studied by \cite{HHL}, \cite{BHS} and \cite{BS}. Let $U(C_t,X_t) = \log (c_t e^{\lambda'X_t})$ where $c_t$ is date-$t$ consumption and $\lambda \in \mb R^d$. Each period, the DM chooses how much to consume, $c_t$, and portfolio weights, $\pi_{t+1}$, hence $C_t = (c_t,\pi_{t+1})$. There is an exogenous Markov state process $X$. It is also assumed that aggregate consumption and dividends are both functions of $(X_t,X_{t+1})$, which is trivially the case in typical partial-equilibrium settings where date-$t$ consumption growth and dividends are themselves components of $X_t$.  At date $t$, the DM solves
\begin{align*}
 V_t & = \max_{C_t \in \mc C_t} \min_{m_{t+1} \in \mc M_{t+1}} \; \log (c_t e^{\lambda'X_t}) + \beta \mb E^Q \Big[m_{t+1} \Big( V_{t+1} + \theta \log m_{t+1} \Big) \Big|  X_t,C_t \Big]
\end{align*}
subject to a budget constraint. 

In equilibrium, $c_t = c_t^*$ and $\pi_t = \pi^*$ where $\pi^*$ denotes the market-clearing vector of portfolio weights. The recursion in equation (\ref{e-V-recur-1}) becomes
\[
 V_t  =  \log (c_t^*) + \lambda'X_t - \beta \theta \log \mb E^Q \Big[e^{-\theta^{-1} V_{t+1}} \Big|  X_t,c_t^* \Big] \,.
\]
By homotheticity, $V_t - \frac{1}{1-\beta}(\log (c_t^*) + \lambda'X_t) = \zeta(X_t)$ for some $\zeta : \mc X \to \mb R$. Setting $v= -\theta^{-1} \zeta$ yields the recursion in equation (\ref{e-recur-v}) with $u(X_t,X_{t+1}) = \log( c_{t+1}^*/ c_t^*)  + \lambda'(X_{t+1}-X_t)$. The DM's first-order conditions deliver the Euler equation
\begin{equation} \label{e-euler}
 \mb E^Q \left[ \left. m_{t+1}^* \beta \left( \frac{c_{t+1}^*}{c_t^*} \right)^{-1} R_{t+1}  \right| X_t \right] = 1
\end{equation}
where $R_{t+1}$ denotes the return on a traded asset from $t$ to $t+1$ and $m_{t+1}^*$ is from equation (\ref{e-m-v}).

\paragraph{Linear-Gaussian (LG) example:}
We use a parametric example  from \cite{BHS} to illustrate some ideas in a transparent way throughout the paper.
Suppose the DM's benchmark model is
\[
 X_{t+1} =  \mu + A X_t + \sigma \varepsilon_{t+1}\,,
\]
where the $\varepsilon_t$ are i.i.d. $N(0,I)$ and all eigenvalues of $A$ are inside the unit circle. Also let $ \log (c_{t+1}^*/c_t^*) = \lambda_0'X_t + \lambda_1' X_{t+1}$ for some fixed $\lambda_0,\lambda_1 \in \mb R^d$. 
A solution to (\ref{e-recur-v}) is $v(x) = a + b x$ where
\begin{align*}
 a & = \frac{\beta}{1-\beta} \Big( (\alpha \lambda_1 + b)'\mu + \frac{1}{2}(\alpha \lambda_1 + b)' \sigma \sigma'(\alpha \lambda_1 + b) \Big) \,, &
 b & = \alpha \beta (I - \beta A')^{-1}(\lambda_0 + A'\lambda_1) \,.
\end{align*}
This is the unique solution among affine functions. It remains to be seen whether solutions to recursion (\ref{e-recur-v}) exist and are unique under departures from linearity and Gaussianity. The identification results in the next section provide affirmative answers to this questions.

The worst-case belief distortion induced by $v$ is 
\[
 m_{t+1}^* = e^{(\sigma'(\alpha \lambda_1 + b))'\varepsilon_{t+1} - \frac{1}{2}(\alpha \lambda_1 + b)'\sigma \sigma' (\alpha \lambda_1 + b)} \,.
\]
This belief distortion corresponds to shifting the mean of $\varepsilon_{t+1}$ from zero to $\sigma'(\alpha \lambda_1 + b)$. Thus, under the DM's worst-case model:
\begin{equation*} 
 X_{t+1} =\mu^* + A X_t + \sigma \varepsilon_{t+1}
\end{equation*}
where $\mu^* =  \mu + \sigma \sigma'(\alpha \lambda_1 + b)$. The worst-case model corresponds to shifting the mean $\mu$ to $\mu^*$ irrespective of the current value of the state. 
There exists a continuum of $(\theta,\mu)$ that yield identical $\mu^*$. Thus, the DM's benchmark model and preference parameters are underidentified from data on the state $X_t$ and asset returns. As shown in the next section, joint underidentification of the benchmark model and preference parameters is generic.

\section{Identification}\label{s:identification}

This section presents three results about identification. First, primitive, directly verifiable conditions are derived for nonparametric identification of the DM's continuation value function, worst-case belief distortion, and continuation entropy given $(Q,U,\beta,\theta)$. 
 The key step is to establish primitive conditions for existence and uniqueness of the DM's value function, which is of independent interest. 
Second, local identification conditions are derived for the preference parameters $(\beta,\theta)$ given $(Q,U)$. As a special case of the model is isomorphic to models with EZW recursive utility with unit IES, these results have direct implications for identification in EZW models also.\footnote{In those settings, agents are typically assumed to have rational expectations, in which case $Q$ can be identified with the data-generating probability measure.} Third, an underidentification result is stated which shows $(Q,\theta)$ are not identified even if ($U$,$\beta$) and the worst-case model are known. 

To locally identify preference parameters, it will be presumed that there exists an auxiliary vector of moment conditions holds under the worst-case model, namely:
\begin{equation} \label{e-g-mom}
 \mb E^Q \Big[  m_{t+1}^* \beta \bs g(X_t,X_{t+1},Y_{t+1}) - \bs 1 \Big| X_t \Big] = \bs 0
\end{equation} 
where $Y_t$ is a vector of variables with support $\mb Y \subset \mb R^{d_y}$ such that the conditional distribution of $(X_{t+1},Y_{t+1})$ given $(X_t,Y_t)$ depends only on $X_t$, and the function $\bs g : \mc X^2 \times \mb Y \to \mb R^{d_g}$ is known. An example of this setting is the Euler equation (\ref{e-euler}), where 
\begin{align}\label{e:g:euler:2}
  \bs g(X_t,X_{t+1},Y_{t+1}) = \left( \frac{c_{t+1}^*}{c_t^*} \right)^{-1} \bs R_{t+1}
\end{align}
where $\bs R_{t+1}$ is a vector of asset returns from date $t$ to $t+1$. The tuples $(Q,U,\beta,\theta)$ and $(Q',U',\beta',\theta')$ are said to be \emph{observationally equivalent} if the conditional moment restriction (\ref{e-g-mom}) holds under both $(Q,U,\beta,\theta)$ and $(Q',U',\beta',\theta')$, where the worst-case belief distortion is constructed as in equations (\ref{e-recur-v}) and (\ref{e-m-v}) under $(Q,U,\beta,\theta)$ and $(Q',U',\beta',\theta')$, respectively. 

Throughout this section, $U$ is assumed to be known by the econometrician. The parameter space for $Q$ is the set $\mc Q$ of all Markov transition kernels on $(\mc X,\mcr X)$ and the parameter space for $(\beta,\theta)$ is $B \times \Theta$ where $B = (0,1)$ and $\Theta = (0,\infty)$.

\subsection{Existence and uniqueness of continuation values}\label{s:id}

The recursion for $v$ in equation (\ref{e-recur-v}) may be written as the fixed-point equation $v = \mb T v$ with 
\[
 \mb T f (x) = \beta \log \mb E^Q \Big[ e^{ f(X_{t+1}) + \alpha u(X_t,X_{t+1})} \Big| X_t=x \Big] \,.
\]
It should be understood that $\mb T$ and $v$ depend implicitly on $(\beta,\theta)$. This implicit dependence will be used to derive local identification conditions for $(\beta,\theta)$ in the next subsection. This subsection presents primitive, directly verifiable conditions on $(Q,U)$ under which the operator $\mb T$ has a unique fixed point within an appropriate class of functions for any $(\beta,\theta) \in B \times \Theta$.

First, a word on two function classes that are not appropriate: (i) bounded functions and (ii) $L^p$ spaces. Most work to date has established existence and uniqueness within the class $B(\mc X)$ of bounded functions on $\mc X$ equipped with the sup norm (see the discussion after Proposition \ref{prop-nonunique}).
Yet $B(\mc X)$ is an inappropriate class for workhorse parametric models where $v$ is unbounded, such as the LG model discussed above. A possible solution might be to truncate the support of $X$ at some arbitrarily large value. However, artificial restriction of the support of an unbounded state process to bounded sets can lead to uniqueness in the restricted problem even when the unrestricted problem does not have a unique solution (see Appendix \ref{s:nonunique} for an example). Intuitively, this is because the nonlinear adjustment in the recursion means that all moments matter, and artificial truncation eventually has a material effect on sufficiently high moments. We therefore seek a class that accommodates unbounded functions. A natural class of unbounded functions is $L^p(Q_0)$, which consists of all functions with finite $p$th moment under $Q_0$. However, the operator $\mb T$ may not be defined on all of $L^p(Q_0)$ for any $1 \leq p < \infty$. For instance, in the LG example above with scalar $X_t$, for any $k \geq 2$ and any $1 \leq p < \infty$ the function $f(x) = x^{2k}$ belongs to $L^p(Q_0)$ but $\mb T f$ is not defined. Intuitively, the tails of $x^{2k}$ under $Q$ are too thick to be compatible with the nonlinear adjustment.

To get around these issues, we embed the analysis within a Banach space of unbounded but ``thin-tailed'' functions. Let $\phi_r(x) = \exp(x^r) -1$ for $r \geq 1$. The \emph{Orlicz space} $L^{\phi_r}(Q_0)$ is (the equivalence class of) all measurable $f : \mc X \to \mb R$ for which 
\begin{equation*} 
 \|f\|_{L^{\phi_r}(Q_0)} := \inf\left\{ c > 0 : \mb E^{Q_0}[\phi_r(|f(X_t)|/c)] \leq 1\right\} < \infty\,.
\end{equation*}
The \emph{Orlicz heart} $E^{\phi_r}(Q_0) \subset L^{\phi_r}(Q_0)$ consists of all $f \in L^{\phi_r}(Q_0)$ with $\mb E^{Q_0}[\phi_r(|f(X_t)|/c)] < \infty$ for each $c > 0$. To simplify notation we drop dependence of the spaces and norm on $Q_0$ and simply write $L^{\phi_r}$, $E^{\phi_r}$ and $\|\cdot\|_{\phi_r}$. Suppose $X_t$ (a scalar) is normally distributed under $Q_0$. The function $a_0 + a_1 x + a_2 x^2$ belongs to $L^{\phi_1}$ because $\mb E^{Q_0}[e^{(x/c)^2}] < \infty$ for all $c$ sufficiently large. As this expectation is infinite for $c$ sufficiently small, the function does not belong to $E^{\phi_1}$ if $a_2 \neq 0$. Similarly, the function $a_0 + a_1 x$ belongs to $L^{\phi_2}$ but not to $E^{\phi_2}$ if $a_1 \neq 0$, and  belongs to $E^{\phi_r}$ for all $1 \leq r < 2$. Moreover, functions that grow no faster than $|x|^{2/r}$ for some $r > 1$ belong to $E^{\phi_s}$ for all $1 \leq s < r$.
The spaces $E^{\phi_r}$ and $L^{\phi_r}$ are (separable and nonseparable, respectively) Banach spaces when equipped with the norm $\|\cdot\|_{\phi_r}$.  Further properties of these spaces are described in Appendix \ref{ax:orlicz}, for now we simply observe that $E^{\phi_r} \subseteq E^{\phi_s}$ and $L^{\phi_r} \subseteq L^{\phi_s}$ for each $1 \leq s \leq r$. Define the spaces $E_2^{\phi_r}$ and $L_2^{\phi_r}$  analogously to $E^{\phi_r}$ and $L^{\phi_r}$ for functions of $(X_t,X_{t+1})$ using the stationary distribution $Q_0 \otimes Q$ of $(X_t,X_{t+1})$.

The only condition required for identification is that utility growth is thin-tailed. We verify this condition for some models at the end of this subsection.

\medskip

\paragraph{Assumption U} \emph{$u \in E^{\phi_r}_2$ for some $r > 1$.}

\medskip

Assumption U ensures $\mb T$ is a well-defined mapping from $E^{\phi_s}$ to $E^{\phi_s}$ for each $1 \leq s \leq r$. 
However, $\mb T$ may not be a contraction on $E^{\phi_s}$ (see Appendix \ref{s:noncontract}). We therefore make use of certain monotonicity and convexity properties of $\mb T$ to establish existence and uniqueness of $v$. 

For the intuition, consider an increasing, convex function $T : \mb R \to \mb R$ (see Figure \ref{f:eu}). The function $T$ can have zero, one, two, or a continuum of fixed points. If there is a point $\ol v$ such that $T(\ol v)$ lies on or below the 45 degree line and the sequence $T^n(\ol v)$ is bounded from below, then $T$ must have at least one fixed point. This is true of the blue, orange, and purple functions plotted in Figure \ref{f:eu}, but not the grey functions. On the other hand, if at every fixed point the function $T$ has a subgradient that is strictly less than 1 then $T$ must have at most one fixed point. The subgradient of the orange function exceeds 1 at its upper fixed point; similarly, the subgradients of the purple function are 1 along the continuum of fixed points on the 45 degree line.

\begin{figure}[hptb]

\vskip 8pt

\begin{center}
\begin{tikzpicture}[scale=0.8, every node/.style={scale=0.8}]
\draw[step=1cm,lightgray,very thin] (-4.5,-1.5) grid (4.5,4.5);
\draw[thick,<->] (-1,-1.5) -- (-1,4.5) node[anchor=west] {$T(v)$};
\draw[thick,<->] (-4.5,0) -- (4.5,0) node[anchor=west] {$v$};
\draw[thick,dashed] (-2.5,-1.5) -- (3.5,4.5) ;
\draw[color=lightgray,domain=-1.5:4,thick,<->]   plot ({\x},{\x })  ; 
\draw[color=lightgray,domain=-4:1.5,thick,<->]   plot ({\x},{0.02*exp(\x+3) + 2.5})  ;
\draw[color=purple,domain=-1.5:-0.5,thick,-]   plot ({\x},{\x+1})  ;
\draw[color=purple,domain=-3.5:-1.5,thick,<-]   plot ({\x},{0.2*\x-.2})  ;
\draw[color=purple,domain=-0.5:0.5,thick,->]   plot ({\x},{3.2*\x+2.1})  ; 
\draw[color=orange,domain=-3:2,thick,<->]   plot ({\x},{0.1*exp(\x+2) -1.25})  ; 
\draw[color=blue,domain=-4:3,thick,<->]   plot ({\x},{0.0125*exp(\x+1) +0.1*\x +1})  ; 
\end{tikzpicture}

\vskip 4pt

\parbox{12cm}{\caption{\small\label{f:eu} Intuition in one dimension. 
}}
\end{center}

\vskip -10pt

\end{figure}
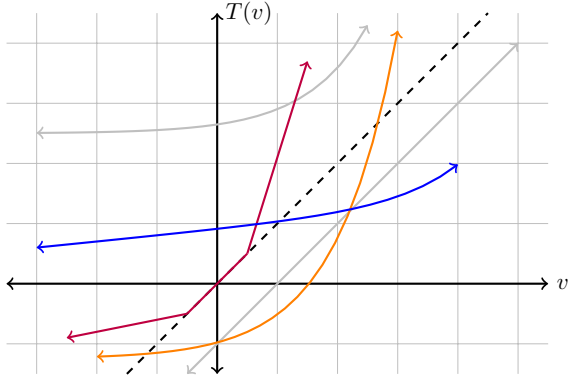

Proposition \ref{p:exun}  in Appendix \ref{ax:id:gen} presents a reasonably general existence and uniqueness result which extends this reasoning from a one-dimensional setting to an infinite-dimensional setting. 
Here we give an heuristic description of the result and introduce relevant definitions. 

Given $f,g \in E^{\phi_s}$, write $f \leq g$ if $f(x) \leq g(x)$ holds $Q_0$-a.e.. Say that $\mb T $ is \emph{monotone} (or \emph{isotone}) if $f \leq g$ implies $\mb T f \leq \mb T g$ and {\it convex} (or \emph{order-convex}) if for each pair of functions $f$ and $g$ and each $\tau \in [0,1]$ we have $\mb T(\tau f + (1-\tau)g) \leq \tau \mb T f  + (1-\tau) \mb T g$.

\begin{lemma}\label{lem-T-prop}
Let Assumption U hold. Then  $\mb T$ is a continuous, monotone and convex operator on $E^{\phi_s}$ for each $1 \leq s \leq r$.
\end{lemma}

The properties of $\mb T$ established in Lemma \ref{lem-T-prop} are used to establish existence of a fixed point $v \in E^{\phi_r}$. Uniqueness requires an appropriate notion of a subgradient of $\mb T$. Let $\mb E_v$ denote expectation under the conditional distribution induced by  $m_v$ from equation (\ref{e-m-v}) and define
\begin{align*}
 \mb D_v f(x) & = \beta \mb E_v[ f(X_{t+1}) | X_t = x] = \beta \mb E^Q[ m_v(X_t,X_{t+1}) f(X_{t+1}) | X_t = x] \,.   
\end{align*}
\cite{HST} showed that the operator $\mb T$ satisfies a subgradient inequality, which they used for pricing assets. In our notation, the subgradient inequality is:
\begin{align}\label{e:subgrad}
 \mb T(v+f) - \mb Tv \geq \mb D_v f \,.
\end{align}
Given a linear operator $\mb K : E^{\phi_s} \to E^{\phi_s}$, let $\|\mb K\|_{E^{\phi_s}} = \sup\{ \| \mb K f\|_{\phi_s} : f \in E^{\phi_s}, \|f\|_{\phi_s} \leq 1\}$ denote its operator norm and $\rho(\mb K; E^{\phi_s}) = \lim_{n \to \infty} \|\mb K^n\|_{E^{\phi_s}}^{1/n}$ denote its spectral radius, where $\mb K^n$ denotes $\mb K$ applied $n$ times in succession. Define $\|\mb K\|_{L^{\phi_s}}$ and $\rho(\mb K; L^{\phi_s})$ analogously.

\begin{lemma}\label{lem-D-bdd}
Let Assumption U hold and fix any $v \in E^{\phi_{r'}}$ with $r' > 1$. Then: for all $1 \leq s < \infty$, $\mb D_v$ and $\mb E_v$ are continuous linear operators on $E^{\phi_s}$ and $L^{\phi_s}$ with $\rho(\mb D_v;E^{\phi_s}) \leq \rho(\mb D_v;L^{\phi_s}) < 1$.
\end{lemma}

The property $\rho(\mb D_v;E^{\phi_s}) < 1$ is analogous to the function $T$ having a subgradient less than 1 at its fixed points. This property, together with convexity, delivers uniqueness. It is always the case that $\rho(\beta \mb E^Q;L) = \beta < 1$ because $\mb E^{Q}$ is a weak contraction when $L$ is any $L^p$ or Orlicz class defined relative to $Q_0$.\footnote{The weak contraction property follows by Jensen's inequality, iterated expectations, and stationarity.} However, the stationary distribution under the worst-case model may be different from $Q_0$ in which case $\mb D_v$ is not, in general, a contraction (see Appendix \ref{s:noncontract}). 
Nevertheless, functions in $E^{\phi_r}$ have sufficiently thin tails that, under repeated application of $\mb D_v$, probability mass only moves ``so far'' and the effect of the discounting by $\beta$ eventually dominates.

The next theorem, which is the main result of this subsection, establishes nonparametric identification of $v$ given $(Q,U,\beta,\theta)$ within a class of ``thin-tailed'' functions.

\begin{theorem}\label{t-id-W}
Let Assumption U hold. Then: $\mb T$ has a fixed point $v \in E^{\phi_r}$. Moreover, $v$ is the unique fixed point of $\mb T$ in $E^{\phi_s}$ for each $1 < s \leq r$.
\end{theorem}

\begin{remark}\label{rmk-order-int} \normalfont
The proof of Theorem \ref{t-id-W} also shows: (i) that $\ul v \leq  v \leq \ol v$ with 
\begin{align*}
 \ul v(x) & = (\mb I - \beta \mb E^Q)^{-1} \beta \mb E^{Q} \big[ \alpha u(X_t,X_{t+1}) \big|X_t = x \big] \\
 \ol v(x) & = (1-\beta) \sum_{i=0}^\infty \beta^{i+1} \log \mb E^{Q} \Big[ e^{\frac{\alpha}{1-\beta} u(X_{t+i},X_{t+i+1})} \Big|X_t = x \Big] \,,
\end{align*}
where $\mb I$ denotes the identity operator; and (ii) that fixed-point iteration on $\ol v$ will converge to $v$.
\end{remark}

It is worth noting an implication of the proof of Theorem \ref{t-id-W} for the class $E^{\phi_1}$ containing thicker-tailed functions not in $E^{\phi_s}$. Let $\mc V \subset E^{\phi_1}$ denote all fixed points of $\mb T : E^{\phi_1} \to E^{\phi_1}$. Note $\mc V$ always contains $v$ from Theorem \ref{t-id-W}. Say $v$ is the \emph{smallest} fixed point of $\mb T : E^{\phi_1} \to E^{\phi_1}$ if $v' \geq v$ for each $v' \in \mc V$. Say $v' \in \mc V$ is \emph{stable} if $\rho(\mb D_{v'};E^{\phi_1}) < 1$ and \emph{unstable} if $\rho(\mb D_{v'};E^{\phi_1}) \geq 1$. Consider the orange function plotted in Figure \ref{f:eu}: its upper fixed point is unstable---iteration on a neighborhood of this fixed point may diverge---whereas its lower fixed point is stable. 

\begin{proposition}\label{prop-nonunique}
Let Assumption $U$ hold. Then: $v$ is both the smallest fixed point and the unique stable fixed point of $\mb T: E^{\phi_1} \to E^{\phi_1}$.
\end{proposition}

\begin{remark} \normalfont
Existence of a fixed point $v \in E^{\phi_1}$ is guaranteed under the weaker condition $u \in E^{\phi_1}_2$. The stronger condition $u \in E^{\phi_r}_2$ for $r > 1$ (which implies $v \in E^{\phi_r}$) is used to establish stability of $v$ (which implies $v$ is both the smallest and unique stable fixed point of $\mb T$). 
\end{remark}

\paragraph{Related results:} There exist several works establish existence and uniqueness of value functions using contraction or local contraction arguments (see, e.g., \cite{RTW,RZRP,RZRP_recursive,MDRV}). However, $\mb T$ and $\mb D_v$ are generally neither contraction mappings nor local contraction mappings on $E^{\phi_s}$, as shown in Appendix \ref{s:noncontract}. 

There exists a recent related literature on existence and uniqueness of value functions using monotonicity and concavity/convexity of various operators (see \cite{MarinacciMontrucchio2010,Balbus,BorovickaStachurski,GuoHe,BloiseVailakis}). Except for \cite{Balbus} and \cite{GuoHe}, these papers impose restrictions that rule out recursions of the form (\ref{e-recur-v}). The results in these papers apply to classes of bounded functions and therefore require either that the state process $X$ has compact support and/or that utilities are bounded. Unfortunately, such restrictions are incompatible with conventional benchmark models, where $X$ is typically a Markov process with full support, and period utility functions, which are often of logarithmic or CRRA form. Artificially truncating the support of $X$ to be bounded in order to apply these results is not necessarily the right approach, as it may result in misleading conclusions about existence and uniqueness (see Appendix \ref{s:nonunique}). Moreover, the above papers generally make use of fixed-point theorems relying on certain topological properties of the space $B(\mc X)$, such as ``solidness'' of positive cones. These properties are not shared by $L^p$ spaces with $p < \infty$ and Orlicz classes.

\cite{HS2012} presented spectral conditions for existence of a fixed point in $L^1$ of a related recursion corresponding to EZW preferences allowing for unbounded state variables but did not study uniqueness. \cite{npsdfd} established  \emph{local} identification in the same EZW recursion allowing for unbounded state variables under spectral radius and Fr\'echet differentiability conditions but did not establish \emph{global} identification or existence. Theorem \ref{t-id-W} establishes both these properties, applies to a broader class of models, does not require a differentiability condition, and the spectral radius condition is verified directly.

We close this subsection with a discussion of Assumption U for the  framework described in Section \ref{s:example}.

\paragraph{LG environments:} Suppose  $u(X_t,X_{t+1}) = \lambda_0'X_t + \lambda_1'X_{t+1}$ is normally distributed under $Q_0 \otimes Q$. Then $u \in L^{\phi_2}_2$ and so $u \in E^{\phi_r}_2$ for each $1 \leq r < 2$. The affine solution $v(x) = a + b'x$ is therefore the unique solution in $E^{\phi_s}$ for all $1 < s < 2$.

\paragraph{Fat tails and rare disasters:} Consider a model featuring time-varying rare disasters from \cite{BS}. Let $\log(c_{t+1}^*/c_t^*) = g_{t+1}$ where
\[
 g_{t+1} = \mu_g + w_{z,t+1} + \sigma w_{g,t+1} \,,
\]
with $w_{g,t+1} \sim N(0,1)$, $w_{z,t+1}|j_{t+1} \sim N(\mu_j j_{t+1}, \sigma_j^2 j_{t+1})$ where $\mu_j < 0$, $j_{t+1}|h_t$ is Poisson distributed with mean $h_t$ which follows an autoregressive gamma (ARG) process (see Appendix \ref{s:nonunique} for details). Consumption growth is subject to occasional ``disasters'' when $j_t > 0$. The rate at which disasters arrive, $h_t$, is time-varying. Define the state as $X_t = (g_t,h_t)$ so that $u(X_t,X_{t+1}) = \lambda_0'X_t + \lambda_1'X_{t+1}$ with $\lambda_0 = (0,0)'$ and $\lambda_1 = (1,0)'$. 
By iterated expectations:
\[
 \mb E^{Q_0 \otimes Q}\left[ e^{c u(X_t,X_{t+1})}\right] = e^{c \mu_g + \frac{c^2 \sigma^2}{2}} \mb E^{Q_0} \left[ \exp \left( h_t \left( e^{c \mu_j + \frac{c^2 \sigma_j^2}{2}} - 1\right) \right) \right]
\]
which is finite only for values of $c$ close to zero because $h_t$ is Gamma distributed under $Q_0$ and the moment generating function of the Gamma distribution is defined only on a neighborhood of the origin. Therefore, $u \in L^{\phi_1}_2$ which violates Assumption U. Indeed, it is known that there may exist zero, one or two fixed points of the form $v = a + b'x$ under this specification. One could modify the above specification so that $w_{z,t+1}|j_{t+1} \sim N(\mu_j j_{t+1}^{1/\varsigma}, \sigma_j^2 )$ for some $\varsigma \in (1,2]$. Given the low frequency of jumps, this modification is likely to be difficult to distinguish empirically from the original specification. Under this modification, one may deduce that  $u \in L^{\phi_\varsigma}_2$ and hence $u \in E^{\phi_r}_2$ for each $1 \leq r < \varsigma$, implying that there is a unique fixed point $v \in E^{\phi_s}$ for all $1 < s \leq r$.

\subsection{Existence and uniqueness of continuation entropy}

The continuation entropy recursion from equation (\ref{e-recur-R}) may be expressed in operator notation as 
\begin{equation} \label{e-R-op}
 (\mb I - \mb D_v) \Gamma = \chi_v \,,
\end{equation}
where $\chi_v(x) = \beta \mb E_v \big[ \log m_v(X_t,X_{t+1}) \big| X_t =x \big]$ 
is the discounted conditional entropy of $m_{t+1}^*$ (cf. equation (\ref{e-m-v})). Equation (\ref{e-R-op}) is a Fredholm equation of the second kind, which have been studied extensively in the applied mathematics literature and used in economics since at least \cite{Lucas1978} and \cite{TH}. It is well known that $\Gamma :=(\mb I - \mb D_v)^{-1} \chi_v $ is the unique solution to (\ref{e-R-op}) in $E^{\phi_s}$ provided $(\mb I - \mb D_v)$ is continuously invertible on $E^{\phi_s}$ and $\chi_v \in E^{\phi_s}$. The spectral radius condition derived in Lemma \ref{lem-D-bdd} is sufficient for invertibility, leading to the following result.

\begin{theorem} \label{t-id-R}
Let Assumption U hold. Then: $\Gamma=(\mb I - \mb D_v)^{-1} \chi_v$ is the unique solution to (\ref{e-R-op}) in $E^{\phi_s}$ for each $1 \leq s \leq r$.
\end{theorem}

\subsection{Local identification of preference parameters}

This section presents sufficient conditions for local identification of preference parameters $(\beta,\theta)$ given $(Q,U)$ based on the moment condition (\ref{e-g-mom}).  \cite{HST} derived an observational equivalence proposition showing that $(\beta,\theta)$ are not separately identified from consumption and investment data alone in linear-quadratic-Gaussian environments. They also showed that data on prices of risky assets could be used to disentangle the two parameters. Intuitively, their positive result arises because varying $(\beta,\theta)$ generates variation in continuation values, and continuation values are reflected in prices of risky assets. The local identification results presented in this section may be viewed partly as a formalization of this intuition. Characterizing the precise source of variation in continuation values required for identification is a nontrivial task, however, as continuation values vary only implicitly as preference parameters vary. 

Though they did not study models with nonlinear fixed point constraints, our approach is similar in spirit to the general approach of \cite{CCLN} for nonlinear semiparametric models. Local identification is linked to the rank of a particular matrix. By Theorem \ref{t-id-W} we know $v$ is \emph{globally} identified for given preference parameters. Therefore, here we derive local identification conditions for  $(\beta,\theta)$ directly. As a consequence, the rank condition we require is weaker than that which would be required for local identification of $(\beta,\theta,v)$ jointly using the general framework for nonlinear models in \cite{CCLN}. Our approach to establishing local identification can also be generalized to other models with recursive preferences.

Throughout this subsection, let $(\beta_0,\theta_0)$ denote the true preference parameters. Say that $(\beta_0,\theta_0)$ is \emph{locally identified} given $(Q,U)$ if there exists a neighborhood $\mc N \subseteq B \times \Theta$ such that $(Q,U,\beta,\theta)$ and $(Q,U,\beta_0,\theta_0)$ are not observationally equivalent for any $(\beta,\theta) \in \mc N$ with $(\beta,\theta) \neq (\beta_0,\theta_0)$. 

There is a one-to-one mapping between $(\beta,\theta)$ and $(\alpha,\beta)$, so local identification of one guarantees local identification of the other. It is slightly cleaner to work with $(\alpha,\beta)$ than $(\beta,\theta)$ in what follows. Let $\alpha_0 = -\frac{1}{\theta_0(1-\beta_0)}$ denote the true value of $\alpha$. Let $v_{(\alpha,\beta)}$ denote the solution to the recursion (\ref{e-recur-v}) for given $(\alpha,\beta)$. By (\ref{e-m-v}), the moment condition (\ref{e-g-mom})  may be written as
\begin{align*}
 \mb E^Q \bigg[ \underbrace{ \frac{e^{ v_{(\alpha_0,\beta_0)}(X_{t+1}) + \alpha_0 u (X_t,X_{t+1} )}}{e^{\beta_0^{-1}v_{(\alpha_0,\beta_0)}(X_t)}} }_{m_{t+1}^*} \beta_0 \bs g(X_t,X_{t+1},Y_{t+1}) - \bs 1 \bigg| X_t \bigg] & = \bs 0\,.
\end{align*}
We can view the conditional expectation on the left-hand side of the above display as a map from $(\alpha,\beta)$ into a $d_g$-vector of functions of $X_t$. Let
\[
 \rho(\alpha,\beta;X_t) = \mb E^Q \bigg[ \frac{e^{ v_{(\alpha,\beta)}(X_{t+1}) + \alpha u (X_t,X_{t+1} )}}{e^{\beta^{-1}v_{(\alpha,\beta)}(X_t)}}  \beta \bs g(X_t,X_{t+1},Y_{t+1}) - \bs 1 \bigg| X_t \bigg] \,.
\]

To introduce the result, let $\bs g_{t+1} =\mb E^Q[\bs g(X_t,X_{t+1},Y_{t+1})|X_t,X_{t+1}]$ and $v_0 = v_{(\alpha_0,\beta_0)}$. Also let $\mb E_{v_0}^n$ denote iterated conditional expectation under the worst-case model at the true parameters. Thus, $\mb E_{v_0}^2 h(x) = \mb E^Q[m_{t+1}^* \mb E^Q[m_{t+2}^* h(X_{t+1},X_{t+2})|X_{t+1}] |X_t = x]$, and so on. The Fr\'echet derivatives of $\rho$ with respect to $\alpha$ and $\beta$ at $(\alpha_0,\beta_0)$ are
\begin{align}
 \partial_\alpha \rho(\alpha_0,\beta_0;X_t) 
 & = \mb E_{v_0} \left[ \left. (\beta_0 \bs g_{t+1} - \bs 1)  \left( u(X_t,X_{t+1}) + \sum_{n=1}^\infty \beta^n \mb E_{v_0}^n u(X_{t+1}) \right)\right| X_t \right] \label{e-rho-deriv-alpha} \\
 \partial_\beta \rho(\alpha_0,\beta_0;X_t) 
 & = \frac{1}{\beta_0}  \left( \mb E_{v_0}  \left[ \left. (\beta_0 \bs g_{t+1} - \bs 1) \left( v_0(X_{t+1}) + \sum_{n=1}^\infty \beta^n \mb E_{v_0}^n v_0(X_{t+1}) \right) \right| X_t \right] - \bs 1 \right) \label{e-rho-deriv-beta} \,.
\end{align}
Define:
\[
 \mf V = \mb E^{Q_0} \left[ \,
 \left(
 \begin{array}{c}
 \partial_\alpha \rho(\alpha_0,\beta_0;X_t)'  \\
 \partial_\beta \rho(\alpha_0,\beta_0;X_t)' 
 \end{array} 
 \right) 
 \left(
 \begin{array}{c}
 \partial_\alpha \rho(\alpha_0,\beta_0;X_t)'  \\
 \partial_\beta \rho(\alpha_0,\beta_0;X_t)' 
 \end{array} 
 \right)'\,
 \right] \,.
\]
Let $\mc A  = (-\infty,0) \times (0,1)$ denote the parameter space for $(\alpha,\beta)$. 
For the following result, we may view $\mb T$ as an operator from $\mc A \times E^{\phi_s}$ into $E^{\phi_s}$ for some $1 < s \leq r$.

\begin{proposition}\label{prop-local-1}
Let Assumption U hold, let $\mb T : \mc A \times E^{\phi_s} \to E^{\phi_s}$ be continuously Fr\'echet differentiable at $(\alpha_0,\beta_0,v_0)$, let each element of $\bs g_{t+1}$ have finite $2+\varepsilon$ moment under $Q_0 \otimes Q$ for some $\varepsilon > 0$, and let $\mf V$ be positive definite. Then: $(\beta_0,\theta_0)$ is locally identified.
\end{proposition}

The key condition for local identification is positive definiteness of $\mf V$. This condition essentially requires sufficient correlation of the residuals $(\beta_0 \bs g_{t+1}-\bs 1)$ with forward-looking expectations of $u$ and $v$ under the worst-case model.

As the value function recursion is isomorphic to models with EZW recursive utility with unit IES, Proposition \ref{prop-local-1} therefore provides sufficient condition for local identification of preference parameters in that setting also. Global identification conditions may be obtained under further structure on $\bs g$ though we defer this to future research.

\subsection{Underidentification of the benchmark model and preference parameters}

The LG example clearly illustrated joint underidentification of $Q$ and $\theta$: there is a continuum of $\theta$ and drift parameters $\mu$ that produce in the same worst-case model. This result is now generalized outside of LG environments. Although perhaps obvious, the result is informative in terms of pinpointing the cause of the underidentification. Specifically, for each $\theta > 0$ we construct an alternative model $Q_\theta$ by distorting $Q$ by an amount that is exactly offset when formulating the worst-case model under $Q_\theta$. Correspondingly, the distinct tuples $(Q_\theta,U,\beta_0,\theta)$ and $(Q,U,\beta_0,\theta_0)$ both induce the same worst-case model and are therefore observationally equivalent.

\begin{proposition}\label{prop-Q-nonid}
Let Assumption U hold. Then: for each $\theta > 0$ there is a $Q_\theta \in \mc Q$ such that $(Q_\theta,U,\beta_0,\theta)$ and $(Q,U,\beta_0,\theta_0)$ are observationally equivalent. 
\end{proposition}

Proposition \ref{prop-Q-nonid}  holds under the conditions that are used to establish existence of continuation values and preference parameters. Thus, underidentification of benchmark models and preference parameters, even when the worst-case model is fully known, is generic. This result is reminiscent of other nonidentification results for Markov decision processes when agents' beliefs and preferences are allowed to vary (see, e.g., \cite{Rust1994}, Section 3.5).

\section{Learning}\label{s:learn}

This section extends the previous existence and uniqueness results to a class of models where the DM learns about a hidden state, e.g. a regime, stochastic volatility, growth process, or time-varying parameter. This setting is relevant for the extension of multiplier preferences by \cite{HS2007,HS2010} to accommodate learning. This extension is also relevant for models with generalized recursive smooth ambiguity preferences of \cite{JuMiao}, recursive smooth ambiguity preferences of \cite{KMM2009}, and EZW recursive preferences with learning about hidden states as used, for example, by \cite{CLL}. 

\subsection{Setting}

Partition $X_t = (\varphi_t',\xi_t')$ where the DM observes only $\varphi_t$. Let $\mc O_t = \sigma(\varphi_t,\varphi_{t-1},\ldots,\varphi_0)$ denote the information set observable to the agent at date $t$. The DM's beliefs about $\xi_t$ are summarized by a posterior distribution $\Pi_t$ conditional on $\mc O_t$. As in the extension of multiplier preferences by \cite{HS2007,HS2010} to accommodate learning, the date-$t$ value function takes the form:
\begin{equation} \label{e-V-recur-learn}
 V_t = U(C_t^*,X_t) - \beta \theta \log \mb E^{\Pi_t}\! \left[ \left. \mb E^Q \left[ \left.  e^{-\vartheta^{-1} V_{t+1}} \right| \mc O_t,\xi_t,C_t^* \right]^\frac{\vartheta}{\theta}  \right| \mc O_t,C_t^* \right] \,,
\end{equation}
where parameters $\vartheta>0$ and $\theta>0$ encode concerns about misspecification of $Q$ and $\Pi_t$. When $U(C_t^*,X_t) = \log c_t$ then this recursion is isomorphic to that obtained under generalized recursive smooth ambiguity preferences of \cite{JuMiao} with unit IES, where $\theta$ and $\vartheta$ are one-to-one transformations of the ambiguity aversion and risk aversion parameters, respectively. When $\vartheta = \theta$, recursion (\ref{e-V-recur-learn}) reduces to
\[
 V_t = U(C_t^*,X_t) - \beta \vartheta \log \mb E^{\Pi_t}\! \left[ \left. \mb E^Q \left[ \left.  e^{-\vartheta^{-1} V_{t+1}} \right| \mc O_t,\xi_t,C_t^* \right]  \right| \mc O_t,C_t^* \right] \,.
\]
With $U(C_t^*,X_t) = \log c_t$, this recursion corresponds to EZW recursive preferences with unit IES when learning about the hidden state. 
A final special case is obtained in the limit as $\vartheta \to \infty$ (thus, the agent is confident in $Q$ but has doubts about the hidden state), in which case:
\begin{equation} \label{e-V-recur-learn-limit}
 V_t = U(C_t^*,X_t) - \beta \theta \log \mb E^{\Pi_t}\! \left[ \left. e^{-\theta^{-1} \mb E^Q \left[ \left.  V_{t+1}  \right| \mc O_t,\xi_t,C_t^* \right] }  \right| \mc O_t,C_t^* \right] \,,
\end{equation}
as is obtained under recursive smooth ambiguity preferences of \cite{KMM2009}. 

Several conditions are imposed to make the analysis tractable. First, the state is assumed to have a conventional hidden Markov structure, in which the conditional distribution factorizes as $Q(X_{t+1}|X_t,C_t^*) = Q_\varphi(\varphi_{t+1}|\xi_t,C_t^*)Q_\xi(\xi_{t+1}|\xi_t,C_t^*)$. This accommodates models with regime-switching studied by \cite{JuMiao} as well as models with learning about a hidden growth term as in \cite{CLL} and \cite{CMST}. Our analysis readily extends to allow for realizations of $\varphi_t$ to influence future realizations of $\varphi$, but we maintain this simpler presentation for convenience. 

The most restrictive condition is a dimension reduction condition assuming $\Pi_t$ is summarized by a finite-dimensional sufficient statistic $\tilde \xi_t$.
This is trivially true  under Bayesian updating when $\xi_t$ is a hidden regime as in \cite{JuMiao} or when the evolution of the full state $X_t$ under $Q$ is described by a Gaussian state-space model as in \cite{HS2007,HS2010},  \cite{CLL}, \cite{CMST}, and several other works. In other settings, $\tilde \xi_t$ could be a sufficient statistic used to update beliefs in a boundedly-rational way. Under this condition, the effective state vector is $\tilde \xi_t$. Let $\tilde X_t = (\varphi_t',\tilde \xi_t')$ and let $\mc X_{\tilde X}$, $\mc X_{\tilde \xi}$, and $\mc X_\varphi$ denote the support of $\tilde X_t$, $\tilde \xi_t$, and $\varphi_t$.
 
It is also assumed that learning is in a ``steady state'' under which the process $\{\tilde \xi_t: t \in T\}$ is stationary. Consider, for instance, LG environments in which learning about hidden states corresponds to updating beliefs via the Kalman filter. If the filter is not initialized in its steady-state then this process will typically be non-stationary. The stationary problem studied here can be viewed as a boundary problem once the filter has converged to its steady state. Solutions could be obtained by backwards iteration from the steady-state boundary solution.\footnote{A similar approach is taken by \cite{CDJL} in models with an EZW agent who learns about parameters of the data-generating process.} Uniqueness of the boundary solution may be used to establish uniqueness of the backward iterates. 

The following assumption is maintained throughout this section (cf. Condition S in Section \ref{s:fw}).

\medskip

\paragraph{Condition S-Learn} \emph{(i)  $X$ is a stationary, first-order Markov process under $Q(\cdot|X_t,C_t^*)$ and the transition distribution factorizes as $Q(X_{t+1}|X_t,C_t^*) = Q_\varphi(\varphi_{t+1}|\xi_t,C_t^*)Q_\xi(\xi_{t+1}|\xi_t,C_t^*)$; \\
(ii) $\Pi_t(\xi_t) = \Pi_\xi(\xi_t|\varphi_t,\tilde \xi_t)$ where $\tilde \xi$ is updated according to a rule $\tilde \xi_{t+1} = \Xi(\tilde \xi_t,\varphi_{t+1})$; \\
(iii) $\{(\xi_t,\tilde X_t) : t \in T\}$ is strictly stationary; \\ 
(iv) There exist $v : \mc X_{\tilde \xi} \to \mb R$ and $u : \mc X_\varphi \to \mb R$ and  such that
\begin{align*}
 v(\tilde \xi_t) & = -\frac{1}{\theta} \left( V_t - \frac{1}{1-\beta} U(C_t^*,X_t) \right)\,, &
 u(\varphi_{t+1}) & =  U(C_{t+1}^*,X_{t+1}) - U(C_t^*,X_t) \,.
\end{align*}
}

\vskip -16pt

 Before proceeding, two examples of settings in which Condition S-Learn holds are given. For both examples, let $U(C_t^*,X_t) = \log (c_t^*)$ and let $\log(c_{t+1}^*/c_t^*)$ be a function of $\varphi_{t+1}$.

\paragraph{Example: regime switching.}
Suppose that $\xi_t$ denotes a hidden regime and evolves as a Markov chain on finite statespace $\{1,\ldots,N\}$ with transition matrix $\bs \Lambda$. Let $\Delta^{N-1}$ denote the simplex in $\mb R^N$. Let the conditional distribution of $\varphi_{t+1}$ given $\xi_t = \xi$ have density $q(\cdot|\xi)$. The posterior $\Pi_t$ is identified with a vector $\tilde \xi_t \in \Delta^{N-1}$ which is updated as:
\[
 \tilde \xi_{t+1} =  \bs \Lambda \frac{\vec q(\varphi_{t+1}) \odot \tilde \xi_t}{\bs 1' (\vec q(\varphi_{t+1}) \odot \tilde \xi_t)} \,,
\]
where $\vec q(\varphi_{t+1})$ is the $N$-vector whose entries are $q(\varphi_{t+1}|\xi)$ for $\xi \in \{1,\ldots,N\}$, $\odot$ denotes element-wise product, and $\bs 1$ is a $N$-vector of ones \cite[Section 4.2]{Hamilton1994}.

\paragraph{Example: Gaussian state-space models.} Suppose $X$ evolves under $Q$ according to:
\begin{align*}
 \varphi_{t+1} & =  A \xi_t + u_{t+1} \,, & 
 \xi_{t+1} & = B \xi_t + w_{t+1} \,,
\end{align*}
where $u_t$ and $w_t$ are i.i.d. $N(0,\Sigma_u)$ and $N(0,\Sigma_w)$, respectively, and where the maximum eigenvalue of $B$ is inside the unit circle. If $\xi_0\sim N(\tilde \mu_{0},\tilde \Sigma_{0})$ under $\Pi_0$ then $\xi_t \sim N(\tilde \mu_{t},\tilde \Sigma_{t})$ under $\Pi_t$. The matrix $\tilde \Sigma_{t}$ will converge to a fixed matrix $\bar \Sigma$ as $t \to \infty$. In this steady state, the sufficient statistic for $\Pi_t$ is $\tilde \xi_t = \tilde \mu_{t}$, which is updated as $\tilde \xi_{t+1} = B \tilde \xi_{t} + B \bar\Sigma A'(A\bar\Sigma A' + \Sigma_u)^{-1}(\varphi_{t+1}  -  A \tilde \xi_{t})$.

\subsection{Existence and uniqueness of continuation values}

 The only existence and uniqueness result for value functions we are aware of in any of these setting is that of \cite{KMM2009}, which applies to a more restrictive model (corresponding to $\vartheta = +\infty$), requires finite support of the state, and applies to the class of bounded functions. The results presented below relax these conditions.
 
In view of Condition S-Learn, we again abuse notation slightly and drop dependence of conditional distributions on $C_t^*$. First consider the case with $\vartheta < \infty$.  The recursion (\ref{e-V-recur-learn}) may be reformulated as the fixed-point equation $v(\tilde \xi) = \tilde{\mb T} v(\tilde \xi)$ where
\[
 \tilde{\mb T} f(\tilde \xi_t) = \beta \log \mb E^{\Pi_\xi}\! \left[ \left. \mb E^{Q_\varphi} \left[ \left.  e^{\frac{\theta}{\vartheta} f(\Xi(\tilde \xi_t,\varphi_{t+1})) + \alpha u(\varphi_{t+1})} \right| \xi_t,\tilde \xi_t \right]^\frac{\vartheta}{\theta}  \right| \tilde \xi_t \right] .
\]
The recursion (\ref{e-V-recur-learn-limit}) in the limiting case with $\vartheta = +\infty$ may be reformulated as the fixed-point equation $v(\tilde \xi) = \tilde{\mb T} v(\tilde \xi)$ where
\[
 \tilde{\mb T} f(\tilde \xi_t) = \beta \log \mb E^{\Pi_\xi}\! \left[ \left.   e^{ \mb E^{Q_\varphi} \left[ \left.  f(\Xi(\tilde \xi_t,\varphi_{t+1})) + \alpha u(\varphi_{t+1})  \right| \xi_t,\tilde \xi_t \right] }   \right| \tilde \xi_t \right] \,.
\]
The existence and unqiueness results presented below apply to either case, though the proofs are presented only for the more complicated case with $\vartheta < \infty$.

The first condition required for identification of $v$ is an appropriate version of Assumption U. Let $\tilde Q_0$ denote the stationary distribution of $\tilde X_t$ and let $E^{\phi_r}_{\tilde X}$ denote the Orlicz heart consisting of all $f : \mc X_{\tilde X} \to \mb R$ for which $\mb E^{\tilde Q_0} [ e^{|u(\tilde X_{t+1})/c|^r}] < \infty$ for each $c > 0$. Similarly, let $E^{\phi_r}_\varphi \subset E^{\phi_r}_{\tilde X}$ and $E^{\phi_r}_{\tilde \xi} \subset E^{\phi_r}_{\tilde X}$ denote functions in $E^{\phi_r}_{\tilde X}$ depending only on $\varphi$ or only on $\tilde \xi$, respectively.

\medskip 

\paragraph{Assumption U-Learn} $u \in E^{\phi_r}_\varphi$ for some $r > 1$.

\medskip

Assumption U-Learn depends only on the marginal distribution of the observed state and is therefore easy to verify. For example, suppose $U(C_{t+1}^*,X_{t+1}) = \log (c_t^*)$. \cite{JuMiao} study an economy in which consumption and dividend growth is modeled as
\begin{align*}
 \log(c_{t+1}^*/c_t^*) & = \kappa_{\xi_t} +  u_{t+1} \,, & 
 \log(d_{t+1}/d_t) & = \zeta \log(c_{t+1}/c_t) + g_d + w_{t+1} \,,
\end{align*}
where $u_t$ and $w_t$ are i.i.d. $N(0,\sigma_u^2)$ and $N(0,\sigma_w^2)$ and $\xi_t$ is a hidden regime. In this example, $\log(c_{t+1}^*/c_t^*) = \varphi_{t+1}$ and the stationary distribution of $u(\varphi_{t+1})$ is a finite mixture of Gaussians. Assumption U-Learn therefore holds for any $r < 2$.
Similarly, Assumption U-Learn holds for any $r < 2$ in Gaussian state-space settings with $\log (c_{t+1}^*/c_t^*) = \lambda_1' \varphi_{t+1}$.

The next theorem establishes nonparametric identification of $v$ given $(Q,U,\beta,\vartheta,\theta)$ within classes of ``thin-tailed'' functions. The result is derived by applying Proposition \ref{p:exun}  in Appendix \ref{ax:id:gen}.
The operator $\tilde{\mb T}$ is a continuous, monotone, convex operator on $E^{\phi_s}$ for each $1 \leq s \leq r$ (see Lemma \ref{lem-T-prop-learn}) and satisfies a subgradient inequality similar to inequality (\ref{e:subgrad}). Here, however, the subgradient is a discounted conditional expectation operator under a distorted posterior-predictive distribution. Analogous continuity and spectral radius conditions for the subgradient also hold (see Lemma \ref{lem-D-bdd-learn}).

\begin{theorem}\label{t-id-W-learn}
Let Assumption U-Learn hold. Then: $\tilde{\mb T}$ has a fixed point $v \in E^{\phi_r}_{\tilde \xi}$. Moreover, $v$ is the unique fixed point of $\tilde{\mb T}$ in $ E^{\phi_s}_{\tilde \xi}$ for each $1 < s \leq r$. 
\end{theorem}

\begin{proposition} \label{prop-nonunique-learn}
Let Assumption U-Learn hold. Then: $v$ is both the smallest fixed point and the unique stable fixed point of $\tilde{\mb T}: E^{\phi_1}_{\tilde \xi} \to E^{\phi_1}_{\tilde \xi}$.
\end{proposition}

It is possible to relax Assumption S-Learn to allow for $u$ to depend on $(\varphi_t,\varphi_{t+1})$. In this case, however, the effective state vector will be $\tilde X_t$ rather than $\tilde \xi_t$.  The above results go through in this case also under an appropriate modification of Assumption U-Learn.

Given Theorem \ref{t-id-W-learn}, one also may derive local identification results for $(\beta,\theta,\vartheta)$ using similar arguments to Proposition \ref{prop-local-1}.

\section{Estimation}\label{s:estimation}

In taking the model to data, the econometrician must either choose a specific benchmark model or adopt a partial identification approach. The previous literature has done the former,\footnote{See, e.g., \cite{HST,HSW,AHS}.} typically taking the benchmark model to be equal to the member of a parametric family that best approximates the data-generating process. 
This section develops estimation results under general conditions based on a plug-in estimator of the benchmark model. The results allow the first-stage estimate to be parametric or nonparametric.

\subsection{Perturbing the benchmark model}\label{s:perturb}

Consider an alternate benchmark model $\hat Q \in \mc Q$. Let $\hat{\mb T}$ be defined by:
\[
 \hat{\mb T} f(x)  = \beta \log \mb E^{\hat Q} \left[ \left. e^{ f(X_{t+1}) + \alpha u(X_t,X_{t+1})} \right| X_t=x \right] \,.
\]
Under some mild regularity conditions below, $\hat{\mb T}$ will be a well-defined operator and it will have a unique fixed point $\hat v \in E^{\phi_s}$ for each $1 < s \leq r$. This section derives conditions under which $\hat v$ converges to $v$ as $\hat Q$ converges to $Q$ in an appropriate sense. 

Let $\ll$ denote absolute continuity of measures. Say $\hat Q$ and $Q$ are \emph{everywhere mutually absolutely continuous} if $Q(\cdot|x) \ll \hat Q(\cdot|x) \ll Q(\cdot|x)$ for each $x$.  We use the notation $Q \lll \hat Q \lll Q$ to denote everywhere mutual absolute continuity. Whenever this condition holds, write:
\begin{align*}
 \hat \ell(x_0,x_1) & = \log( \hat Q(x_1|x_0) / Q(x_1|x_0) ) \,, \\
 \hat \eta(x_0,x_1) & = \hat \ell(x_0,x_1) - \mb E^Q [\hat \ell(X_t,X_{t+1})|X_t = x_0] \,, \mbox{ and }\\
 \kappa_{\hat \eta}(x_0) & = \log \mb E^Q[e^{\hat \eta(X_t,X_{t+1})}|X_t = x_0] \,.
\end{align*}
We will parameterize alternative models by viewing $\hat \ell$ or $\hat \eta$ as elements of Orlicz classes. To do so, let $N^{\phi_s}_2 = \{ f(x_0,x_1) + h(x_1) : f \in L^{\phi_1}_{2}, h \in E^{\phi_s}\}$ equipped with the $L^{\phi_1}_{2}$ norm. We extend $\mb T$ to have domain $N^{\phi_s}_2$ by defining $\mb T : N^{\phi_s}_2 \to L^{\phi_1}$ as:
\[
 \mb T f(x) = \beta \log \mb E^Q \left[ \left. e^{ f(X_t,X_{t+1}) + \alpha u(X_t,X_{t+1})} \right| X_t=x \right] \,.
\]
Note that this extension preserves the fixed points of $\mb T$.
The operator $\hat{\mb T}$ may be related to the extension of $\mb T$ by noting that for any $f \in E^{\phi_s}$:
\begin{align*}
 \hat{\mb T} f(x) & = \beta \log \mb E^{\hat Q} \left[ \left. e^{ f(X_{t+1}) + \alpha u(X_t,X_{t+1})} \right| X_t=x \right] 
  = \mb T(\hat \ell + f)(x) 
  = \mb T(\hat \eta + f)(x) - \beta \kappa_{\hat \eta} \,.
\end{align*}

To study how fixed points of $\hat{\mb T}$ relate to those of $\mb T$, we impose a mild regularity condition on $\hat Q$. If $X$ is stationary under $\hat Q$, let  $\hat Q_0$ denote its stationary distribution and let $\hat \Delta$ and $\hat \Delta_2$ denote the Radon-Nikodym derivatives of $\hat Q_0$ and $\hat Q_0 \otimes \hat Q$ with respect to $Q_0$ and $Q_0 \otimes Q$.

\paragraph{Assumption AM}
Let $Q \lll \hat Q \lll Q$ and let either (a)  or (b) of the following hold:\\
(a) $\hat \ell \in E^{\phi_r}_2$ \\
(b) $\hat \ell \in L^{\phi_1}_2$, $X$ is stationary under $\hat Q$ with $Q_0 \ll \hat Q_0 \ll Q_0$, and there is $p > 1$ such that $\mb E^{Q_0}[\hat \Delta(X_t)^p] < \infty$, $\mb E^{Q_0}[\hat \Delta(X_t)^{1-p}] < \infty$, and $\mb E^{Q_0 \otimes Q}[ \hat \Delta_2(X_t,X_{t+1})^p] < \infty$.

Assumption AM(b) imposes a less restrictive tail condition on $\hat \ell$ than part (a) but carries the added requirement of stationarity. To understand this assumption, consider the LG setup from Section \ref{s:example}. If $X_{t+1} = \hat \mu + A X_t + \sigma \varepsilon_{t+1}$ under $\hat Q$ (i.e. only the mean parameter is perturbed), then $\hat \ell \in E^{\phi_r}_2$ and so Assumption AM(a) holds for each $1 \leq r < 2$. If $X_{t+1} = \hat \mu + \hat A X_t + \hat \sigma \varepsilon_{t+1}$ under $\hat Q$, then $\hat \ell \in L^{\phi_1}_2$ and so Assumption AM(b) holds provided all eigenvalues of $\hat A$ are inside the unit circle.

\begin{lemma}\label{l-perturb}
Let Assumptions U and AM hold. Then: $\hat{\mb T}$ has a fixed point $\hat v \in E^{\phi_r}$ and $\hat v$ is the unique fixed point of $\hat{\mb T}$ in $E^{\phi_s}$ for each $1 < s \leq r$.
\end{lemma}

One may also establish local Lipschitz and linearity results under a uniform version of Assumption AM. Let $M \geq \mb E^{Q_0 \otimes Q}[ \exp(|\frac{\alpha}{1-\beta} u(X_t,X_{t+1})|^r)]$ be a finite positive constant.

\paragraph{Assumption AM2}
Let $Q \lll \hat Q \lll Q$ and let either (a)  or (b) of the following hold: \\
(a) $\hat \ell \in E^{\phi_r}_2$, $\| \hat \ell \|_{\phi_r} \leq M$, and $\mb E^{Q_0 \otimes Q}[ \exp(|\frac{1}{1-\beta} ( \hat \ell(X_t,X_{t+1}) +\alpha u(X_t,X_{t+1})|^r)] \leq M$ \\
(b) $\hat \ell \in L^{\phi_1}_2$, $X$ is stationary under $\hat Q$ with $Q_0 \ll \hat Q_0 \ll Q_0$, and there is $p > 1$ such that $\mb E^{Q_0}[\hat \Delta(X_t)^p]\leq M$, $\mb E^{Q_0}[\hat \Delta(X_t)^{1-p}] \leq M$, and $\mb E^{Q_0 \otimes Q}[ \hat \Delta_2(X_t,X_{t+1})^p] \leq M$.

Let $C$ denote a positive constant depending only on $\alpha$, $\beta$, $\|u\|_{\phi_r}$, $r$, and $M$ under AM2(a) or on $\alpha$, $\beta$, $\|u\|_{\phi_r}$, $r$, $M$, $p$, and $\mb E^{Q_0 \otimes Q}[ \exp(q^2|\frac{\alpha}{1-\beta} u(X_t,X_{t+1})|^r)]$ under AM2(b) where $p^{-1} + q^{-1} = 1$.

\begin{lemma} \label{l-consistent}
Let Assumption U and AM2 hold and let $\| \hat \eta\|_{\phi_1} \leq 1$. Then:
\begin{equation*}
 \| \hat v - v \|_{\phi_1} \leq C \| \hat \eta\|_{\phi_1}
\end{equation*}
and
\begin{equation*}
  C^{-1} \| \hat{\mb T} v - v\|_{\phi_1} \leq \| \hat v - v \|_{\phi_1} \leq C \| \hat{\mb T} v - v\|_{\phi_1} \,.
\end{equation*}
The inequalities also hold in $\|\cdot\|_{\phi_s}$ norm for every $1 \leq s \leq r$ under Assumption AM2(a).
\end{lemma}

The first inequality in Lemma \ref{l-consistent} shows $\hat v - v$ is locally Lipschitz in $\hat \eta$. It follows from the second inequality that the rate at which $\hat v$ converges to $v$ is equivalent to the rate at which $\hat{\mb T}v $ converges to $v$. Thus, it is \emph{necessary} that $\| \hat{\mb T} v - v\|_{\phi_1} \to 0$ in order that $\|\hat v - v\|_{\phi_1} \to 0$.  To interpret the second inequality, note that 
\[
 \hat{\mb T} v(x) - v(x) 
 = \beta \left( \log \mb E_v \left[ \left. e^{\hat \eta(X_t,X_{t+1})} \right| X_t = x \right] -  \log \mb E^Q \left[ \left. e^{\hat \eta(X_t,X_{t+1})} \right| X_t = x \right]  \right) \,,
\]
i.e., the discounted difference between a certainty equivalent adjustment of $\hat \eta$ under the worst-case and benchmark models.

For the following local linearization result, we view $\eta \mapsto \kappa_{\eta}$ as a map from $L^{\phi_1}_2$ to $L^{\phi_1}$ and index the subgradient by $h \in L^{\phi_1}_2$. The operator $\mb D_{h + v}$ is defined formally in Appendix \ref{ax:estimation:proofs}.

\begin{lemma}\label{l-linear}
Let Assumptions U and AM2 hold, let $\eta \mapsto \kappa_{\eta}$ be Fr\'echet differentiable at $\eta = 0$ and let $\|\mb D_{h+v} - \mb D_v\|_{L^{\phi_1}} \to 0$ as $\|h\|_{\phi_1} \to 0$. Then:
\begin{equation*}\label{e:influence:v}
 \hat v - v = \sum_{n = 1}^\infty (\beta \mb E_v)^n \hat \eta + o(\|\hat \eta\|_{\phi_1}) \,.
\end{equation*}
\end{lemma}

Lemma \ref{l-linear} justifies the approximation $\hat v - v \approx \sum_{n = 1}^\infty (\beta \mb E_v)^n \hat \eta$ when $\|\hat \eta\|_{\phi_1}$ is small. This result shows that approximate value functions in models with rich dynamics may be approximated by perturbing simpler models with closed-form solutions. Appendix \ref{s:perturb-example} presents an example showing how to approximate continuation values in models featuring stochastic volatility by perturbing LG environments. The perturbation is in terms of the likelihood ratio relative to a model with a known solution, unlike usual perturbation methods that expand around a deterministic steady state. In that respect, it shares some similarities with the approach of \cite{KoganUppal} used by \cite{HHLR} and \cite{HHL} to compute approximate continuation values by expanding a preference parameter about a value with a known solution. Here the expansion is in the (infinite-dimensional) score of the alternative model rather than a (scalar) preference parameter. Lemma \ref{l-linear} may also be used to compute influence functions of plug-in estimators of asset pricing functionals.

\subsection{Consistency and convergence rates for general estimators}

We first consider plug-in estimators based on frequentist procedures then turn to Bayes procedures. Given a (parametric or nonparametric) first-stage estimator $\hat Q$ of $Q$, the continuation value recursion may be solved under $\hat Q$ to obtain a fixed point $\hat v$.  Lemma \ref{l-perturb} guarantees existence and uniqueness of $\hat v$ provided $\hat Q$ satisfies Assumption AM. Given $\hat v$, the worst-case belief distortion may be estimated using:
\[
 m_{\hat v}(X_t,X_{t+1}) = e^{\hat v(X_{t+1}) + \alpha u(X_t,X_{t+1}) - \beta^{-1} \hat v(X_t)} \,.
\]
Let $\|f\|_p = \mb E^{Q_0 \otimes Q}[f(X_t,X_{t+1})^p]^{1/p}$ denote the $L^p(Q_0 \otimes Q)$ norm. Let $a_n$ be a positive sequence with $a_n \to 0$ as $n \to \infty$.

\begin{proposition}\label{p-freq}
Let Assumption U hold, let $\hat Q$ satisfy assumption AM2 wpa1, let $\|\hat \eta\|_{\phi_1} = o_p(1)$ and let $\| \hat{\mb T} v - v\|_{\phi_1} = O_p(a_n)$. Then: $\| \hat v - v\|_{\phi_1} = O_p(a_n)$, $\|m_{\hat v} - m_v\|_{p} = O_p(a_n)$ and $\|\frac{m_{\hat v}}{m_v}-1\|_{p} = O_p(a_n)$ for each $1 < p < \infty$.
\end{proposition}

For Bayes procedures, let $\Pi_n$ denote a posterior distribution for $Q$. In parametric models $\Pi_n$ can be a posterior over the parameters in $Q$, but we also allow for nonparametric settings in which $\Pi_n$ is a posterior over a nonparametric class of transition kernels. For each draw $\hat Q$ from $\Pi_n$ that satisfies Assumption AM, one can construct $\hat{\mb T} = \mb T(\hat Q)$ then compute its fixed point $\hat v = v(\hat Q)$, the belief distortion $m _{\hat v} = m_v(\hat Q)$, and so on, building up posterior distributions for these quantities across repeated draws. The next result presents conditions under which such a procedure is consistent and characterizes posterior contraction rates.

\begin{proposition} \label{p-bayes}
Let Assumption U hold, let $\Pi_n(\mc A_n) = 1 + o_p(1)$ for a sequence of subsets $\mc A_n$ satisfying AM2 with $\sup_{\hat Q \in \mc A_n} \| \eta(\hat Q) \|_{\phi_1} = o(1)$ and $\sup_{\hat Q \in \mc A_n} \|({\mb T}(\hat Q)) v - v\|_{\phi_s} = O(a_n)$. Then:
\begin{align*}
 \Pi_n ( \{ \hat Q : \| v(\hat Q) - v(Q)\|_{\phi_1} > C_n a_n \} ) & = o_p(1) \,, \\
 \Pi_n ( \{ \hat Q : \| m_v(\hat Q) - m_v(Q)\|_p > C_n a_n \} ) & = o_p(1) \,, \mbox{ and }\\
 \Pi_n ( \{ \hat Q : \| {\textstyle \frac{m_v(\hat Q)}{m_v(Q)}-1 } \|_p > C_n a_n \} ) & = o_p(1) 
\end{align*}
for each $1 < p < \infty$ and each positive sequence $C_n \to \infty$. 
\end{proposition}

\subsection{Mixtures of experts}\label{s:moe}

The empirical approach we take in the next section is to treat the benchmark model $Q$ as a covariate-dependent mixture of Gaussian VARs. This model can be interpreted as  a ``mixture of experts'' where each ``expert'' is represented by a Gaussian VAR(1) and the weight that the DM assigns to each expert's forecast is time-varying. Mixtures of experts have long been popular in statistics, machine learning, and computer science for solving prediction problems, including in various dynamic settings.\footnote{For early applications to time series see \cite{ZMA}. For more recent applications to macroeconomic time series see \cite{VKG} and \cite{KalliGriffin}.} This model is attractive for our purposes for several reasons. First, the model is very flexible yet retains a clear interpretation which is not necessarily the case, say, with estimates of $Q$ based on other ``flexible'' estimation techniques such as kernels. Second, it is easy to compute transition densities and simulate from the model, facilitating easy computation of value functions and equilibrium prices. Third,  the mixtures can approximate smooth conditional densities arbitrarily well as the number of mixing components increases (see, e.g., \cite{Norets2010mixture}). Fourth, the procedure can be embedded in a state-space setting, which may be relevant for dealing with measurement error and/or mixed frequencies at which macroeconomic data are available. Finally, regularity conditions from Section \ref{s:perturb} guaranteeing existence of value functions and so on are easy to verify under transparent conditions.

We treat the joint distribution for $(X_t,X_{t+1})$ as a $K$-component mixture of normals:
\[
 f(x_t,x_{t+1}) = \sum_{k=1}^K w_k \, \phi( (x_t',x_{t+1}')' ; \mu_k^{(2)} , \Omega_k^{(2)} ) \,,
\]
where $0 \leq w_k \leq 1$ with $\sum_{k=1}^K w_k = 1$, $\phi( x ; \mu, \Omega)$ denotes the normal probability density function with mean $\mu$ and covariance $\Omega$ (with dimensions conformable with $x$), and
\begin{align*}
 \mu^{(2)}_k & = \left[ \begin{array}{c} \mu_k \\ \mu_k \end{array} \right] \,, & 
 \Omega^{(2)}_k & = \left[ \begin{array}{cc} 
 \Omega_k & \Omega_k A_k' \\
 A_k \Omega_k & \Omega_k
 \end{array} \right] \,,
\end{align*} 
where $A_k$ is a square matrix with all eigenvalues inside the unit circle and $\Omega_k$ is positive definite and symmetric. The process $X$ is strictly stationary and ergodic under this specification, with stationary density 
\[
 f_0(x_t) = \sum_{k=1}^K w_k \, \phi( x_t ; \mu_k , \Omega_k )
\]
and conditional density
\[
 f(x_{t+1}|x_t ) = \sum_{k=1}^K w_k(x_t) \, \phi( x_{t+1} ; (I - A_k) \mu_k + A_k x_t , \Sigma_k )
\]
where
\[
 w_k(x_t) = \frac{w_k \, \phi( x_t ; \mu_k , \Omega_k ) }{\sum_{i=1}^K w_i \, \phi( x_t ; \mu_i , \Omega_i )} \,,
\]
and $\Sigma_k = \Omega_k - A_k^{\phantom \prime} \Omega_k A_k'$. The quantity $w_k(x_t)$ is the weight assigned to the $k$th forecasting model having observed $X_t = x_t$, the $k$th forecasting model itself being a Gaussian VAR(1) with mean $(I - A_k) \mu_k$, autoregressive coefficients $A_k$, and conditional variance $\Sigma_k$. State dependence of the weights generates time-variation in the conditional mean and conditional variance of $X$.

Let $\mc Q_K$ denote the set of all such $K$-component mixtures. Also let $\bar {\mc Q}_K$ denote all $Q \in \mc Q_{K}$ whose $\mu_k$ are uniformly bounded and the smallest and largest eigenvalues of $\Omega_k$ and $\Omega_k^{(2)}$ are uniformly bounded away from $0$ and $+\infty$. 
Say that $Q_0$ has \emph{Gaussian-like tails} if it has (Lebesgue) density $q_0$ for which there exist $\ul c, \ol c, \ul s, \ol s \in (0,\infty)$ such that $\ul c \exp(-\frac{1}{2\ul s^2} \|x\|^2 ) \leq q_0(x) \leq \ol c \exp(-\frac{1}{2\ol s^2} \|x\|^2)$. Note that $Q_0$ does not necessarily have to be Gaussian to have Gaussian-like tails: it just must lie between some multiples of Gaussian distributions with possibly different covariance matrices. 

\begin{lemma}\label{lem-moe-reg}
Let $Q$ have strictly positive conditional density $q(\cdot|x_t)$ on $\mb R^d$ for each $x_t$ and let the marginal $Q_0$ and joint $Q_0 \otimes Q$ distributions of $X_t$ and $(X_t,X_{t+1})$ have Gaussian-like tails. Then: any $\hat Q \in \mc Q_K$ satisfies Assumption AM(b). If, moreover, $u(X_t,X_{t+1}) = \lambda_0'X_t + \lambda_1'X_{t+1}$, then: Assumption U holds and Assumption AM2(b) holds for each $\hat Q \in \bar{\mc Q}_K$.
\end{lemma}

Consider an environment where the DM's benchmark model $Q$ is the closest approximation within the class $\mathcal Q_K$ to the true dynamics of $X$.\footnote{Here ``closest'' in the sense of minimizing average Kullback--Leibler divergence between the conditional densities under the data-generating process and $Q$.} As the conditions of the first part of Lemma \ref{lem-moe-reg} hold, this guarantees existence and uniqueness of a fixed point $\hat v$ for each $\hat Q \in \mc Q_K$ under Assumption U. If the uniformity conditions in the second part of Lemma \ref{lem-moe-reg} hold then we may apply the earlier consistency results. All that remains to check is whether the score terms vanish in the manner described by Propositions \ref{p-freq} and \ref{p-bayes}.

\section{Empirical application}\label{s:app}

This section revisits an economy similar to that described in Section \ref{s:example} and studied by \cite{HHL}, \cite{BHS}, and \cite{BS}, amongst others. We depart from the previous literature by modeling the DM's benchmark model as a mixture of experts as described in Section \ref{s:moe}. Our perspective here is to treat this application as a type of sensitivity analysis by examining how various equilibrium quantities differ under slightly more flexible, though still intuitive, nonlinear specifications for the benchmark model. As will be seen, introducing nonlinearities into state dynamics in this fashion generates interesting predictions about equilibrium prices and term structures relative to those obtained under LG specifications. This sections explores these differences and the channels through which they arise.

\subsection{Setup}

Preferences are as described in Section \ref{s:example} with $U(C_t,X_t) = \log (c_t)$. Similar to \cite{HHL}, we use two state variables: aggregate consumption growth and the consumption-earnings ratio (both in logs). The data sourced from the NIPA tables, are at the quarterly frequency, and span 1947Q1 to 2018Q3. The two series are plotted in Figure \ref{f:data}. The series are approximately uncorrelated and may be thought of as representing high- and low-frequency sources of risk.

Two benchmark models for state dynamics are used. The first is a covariate-dependent mixtures of Gaussian vector autoregressions as described in Section \ref{s:moe}. We use $K = 4$ mixtures, though our results were reasonably insensitive to this choice. The second is a LG model where the state is treated as a first-order Gaussian vector autoregression. The first specification with $K>1$ allows time-variation in the conditional variance of $X$ whereas the LG model does not. 
Both models are estimated using Bayes procedures. We use the same priors on parameters common to both models. For the mixture specification, we use an adaptive sequential Monte Carlo algorithm \citep{HerbstSchorfheide} to accommodate potential multi-modality of the posterior.

\begin{figure}[p]
\begin{center}
\includegraphics[trim = 0cm 0cm 0cm 0cm, clip, width = .8\textwidth]{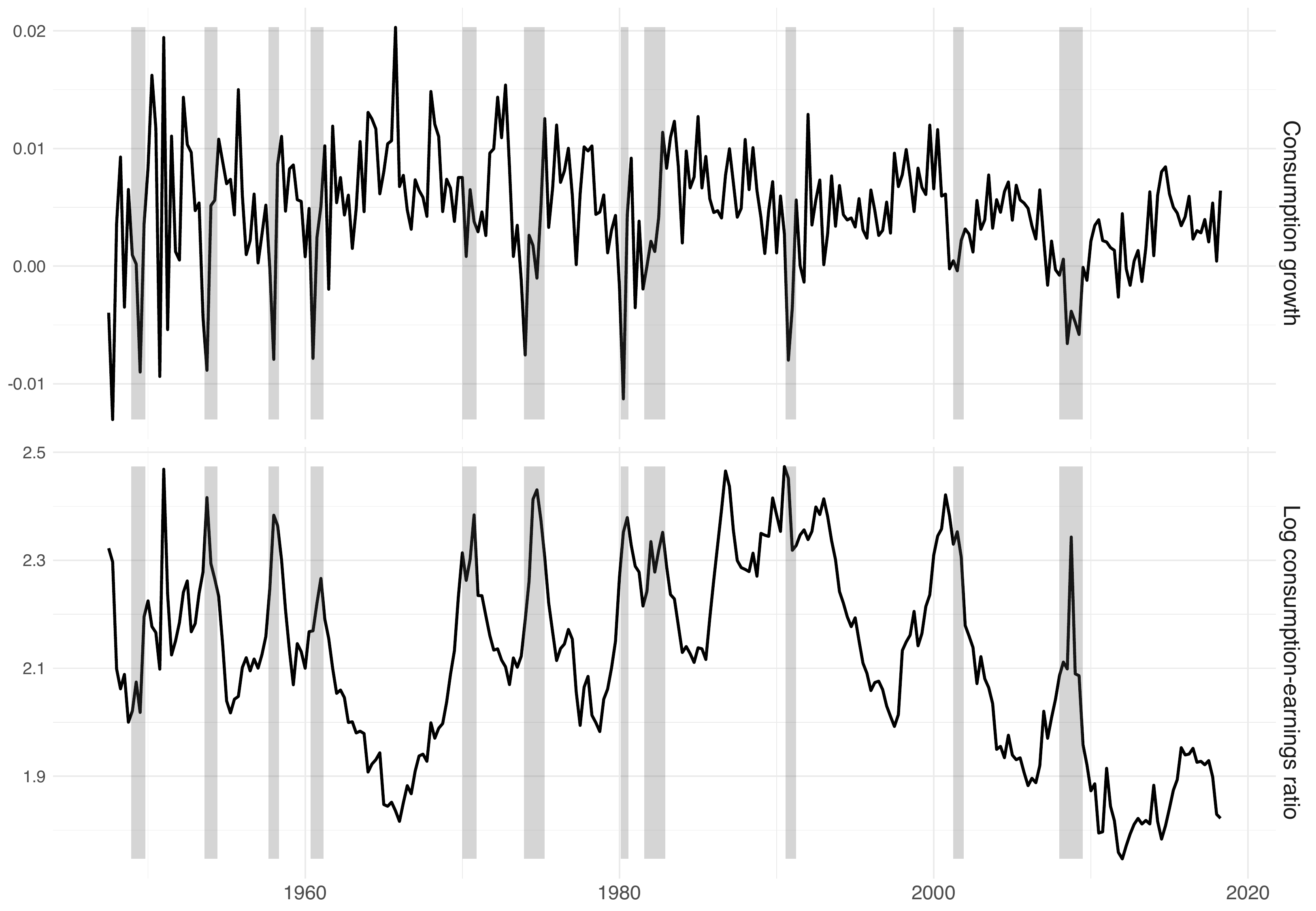}
\parbox{12cm}{\caption{\small\label{f:data}  Time series of log consumption growth and log consumption-earnings ratio. Recession periods are indicated as shaded regions.}}
\end{center}
\end{figure}

\begin{figure}[p]
\begin{center}
\includegraphics[trim = 0cm 0cm 0cm 0cm, clip, width = .8\textwidth]{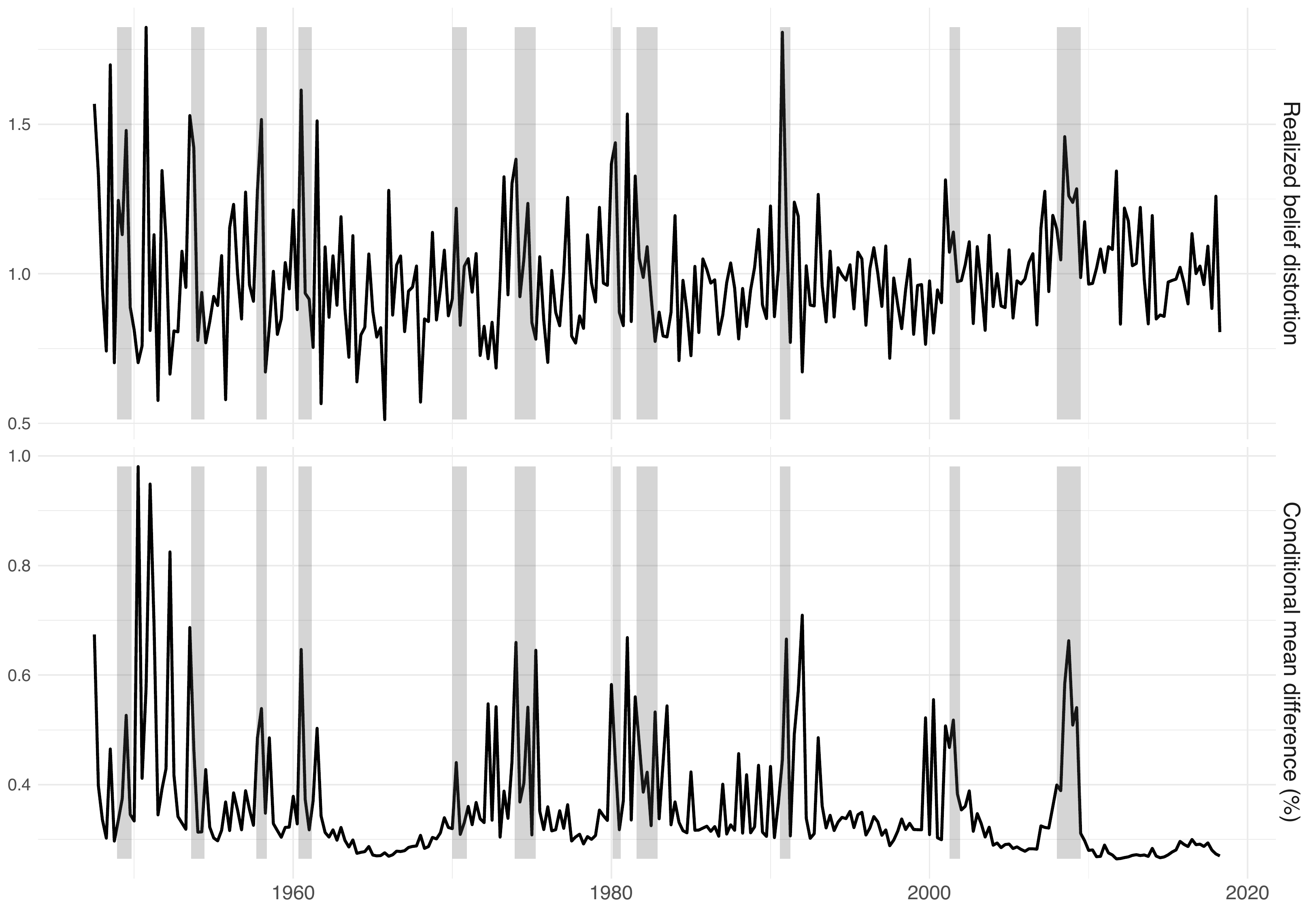}
\parbox{12cm}{\caption{\small\label{f:wc}  Upper panel: Realized belief distortion $m_v(X_t,X_{t+1})$ for the mixture specification. Lower panel: Difference between the conditional means of future consumption growth under the benchmark and worst-case models for the mixture specification. Recession periods are indicated as shaded regions.}}
\end{center}
\end{figure}

For each draw from the posterior, we calculate: (i) the stationary and transition distributions under the benchmark model, (ii) the value function $v$, from which we construct (iii) the worst-case belief distortion $m_v$, (iv) the stationary distribution under the worst-case model and the transition distribution under the worst-case model, (v) the continuation entropy function $\Gamma$, and (vi) term structures of the risk-free rate and excess returns on earnings strips.\footnote{We work with earnings data rather than dividend data to avoid potential seasonality in dividend series.} Computations for the mixture model are performed numerically using interpolation on a large grid.\footnote{As the state-space is compact when using a grid, Proposition \ref{prop-W-unique-bdd} guarantees existence and uniqueness of $v$ in the discretized problem. Our identification results remain relevant in this setting as they ensure that there is a unique solution for the actual un-discretized problem.}

To focus on the role of varying the benchmark model, we calibrate the preference parameters to seemingly reasonable values rather than estimating them directly from data. A more thorough empirical investigation would estimate these parameters from data on asset returns. In particular, we fix the time preference parameter to $\beta = (0.98)^{1/4}$ and the risk-sensitivity parameter to $\theta = 7.367$. The implied return on a 30-year discount bond is around 3\% per annum under this parameterization. Chernoff entropy and detection error probabilities can be used to interpret the scale of $\theta$; see \cite{AHS} and \cite{HS2008}. The posterior mean Chernoff entropy between $Q$ and the worst-case model is $0.0056$ for the mixture specification. The posterior mean half-life of detection-error probabilities is approximately 32 years. Thus, approximately an additional 32 years' worth of data is required in order for error probabilities of likelihood-ratio tests between $Q$ and the worst-case model to halve. The benchmark and worst-case models may therefore reasonably be viewed as statistically difficult to discriminate from one another given the length of data available.

\subsection{The worst-case model: time-varying tails and pessimism}

The upper panel of Figure \ref{f:wc} plots time series of the realized worst-case belief distortion for the mixture specification. This series is constructed by taking the posterior mean of $m_v(X_t,X_{t+1})$ for each date $t$. The series is volatile and pronouncedly counter-cyclical, rising sharply during recessions. Comparing the time series for state variables in Figure \ref{f:data}, the belief distortion also fluctuates at a higher frequency than both of the state variables. 

The lower panel of Figure \ref{f:wc} plots the posterior mean difference between the conditional means of future consumption growth under the benchmark and worst-case models, in percent per year terms.\footnote{I.e., the posterior mean of $(\mb E^Q[\log (C_{t+1}/C_t)|X_t]-\mb E_v[\log (C_{t+1}/C_t)|X_t])\times 400$.} As can be seen, this series is time-varying and counter-cyclical, with the spread rising from below 0.3\% outside of recession periods to around 0.5\%--0.8\% around recession periods. Thus, the worst-case model becomes relatively more pessimistic about consumption growth than the benchmark model during recession periods. The spread is also much more volatile around recession periods. In contrast, the spread is constant for the LG specification. The posterior mean difference is around 0.36\% per annum for the LG model, which agrees with the average posterior mean spread for the mixture specification over the 284 quarters. Thus, the LG model matches the same average spread but misses an important dynamic component.

The time-varying pessimism reported in Figure \ref{f:wc} indicates that the wedge between the benchmark and worst-case models is time-varying. To explore this further and understand differences relative to a LG specification, Figures \ref{f:s1_np} and \ref{f:s3_np} display the conditional distribution for $X_{t+1}$ given $X_t$ under the benchmark and worst-case models in two states. The first is a ``good'' state (Figure \ref{f:s1_np}) when consumption growth is one standard deviation higher than its mean (around 3.77\%) and the consumption earnings ratio is one standard deviation lower than its mean. The second is a ``bad'' state (Figure \ref{f:s3_np}) where consumption growth is one standard deviation lower than its mean (around -0.25\%) and the consumption earnings ratio is one standard deviation high than its mean. Both figures show that the worst-case model assigns more mass to regions of low consumption growth relative to the benchmark model. In the good state, the benchmark and worst-case distributions look similar to those for the LG benchmark specification reported in Figure \ref{f:s1_lg}. In the bad state, however, the conditional distribution in the benchmark model has a longer left tail for consumption growth and the worst-case model assigns relatively more mass far out in the left tail. This variation in the way the benchmark model is distorted to obtain the worst-case model generates the time-varying pessimism reported in Figure \ref{f:wc}. In contrast, for the LG benchmark specification, the worst-case model in the bad state (Figure \ref{f:s3_lg}) looks exactly as it does in the good state, modulo a change in location, with identical contours and marginals. This is entirely as expected: the worst-case model under the LG benchmark is also a Gaussian VAR(1) with a fixed location shift (cf. Section \ref{s:example}).

The asymmetry in the way the left tails of consumption growth behave in the good versus bad states is reminiscent of the work on ``investor fears'' by \cite{BollerslevTodorov} and \cite{BollerslevTodorovXu}. Using S\&P500 options data and model-free continuous-time nonparametric methods, these studies document important time-variation in the wedge between the objective and risk-neutral jump sizes and intensities, and asymmetries between the pricing of left- and right-tail risk, which are ascribed to fluctuations in investor fears. Of course, our frameworks and data sources are very different from these works. Nevertheless, in view of Figures \ref{f:wc}, \ref{f:s1_np}, and \ref{f:s3_np}, it is reasonable to expect that the time-variation in the way the benchmark model is distorted would lead to qualitatively similar pricing of tail events.

\begin{figure}[p]
\begin{center}
\includegraphics[trim = 0cm 0cm 0cm 0cm, clip, width = .8\textwidth]{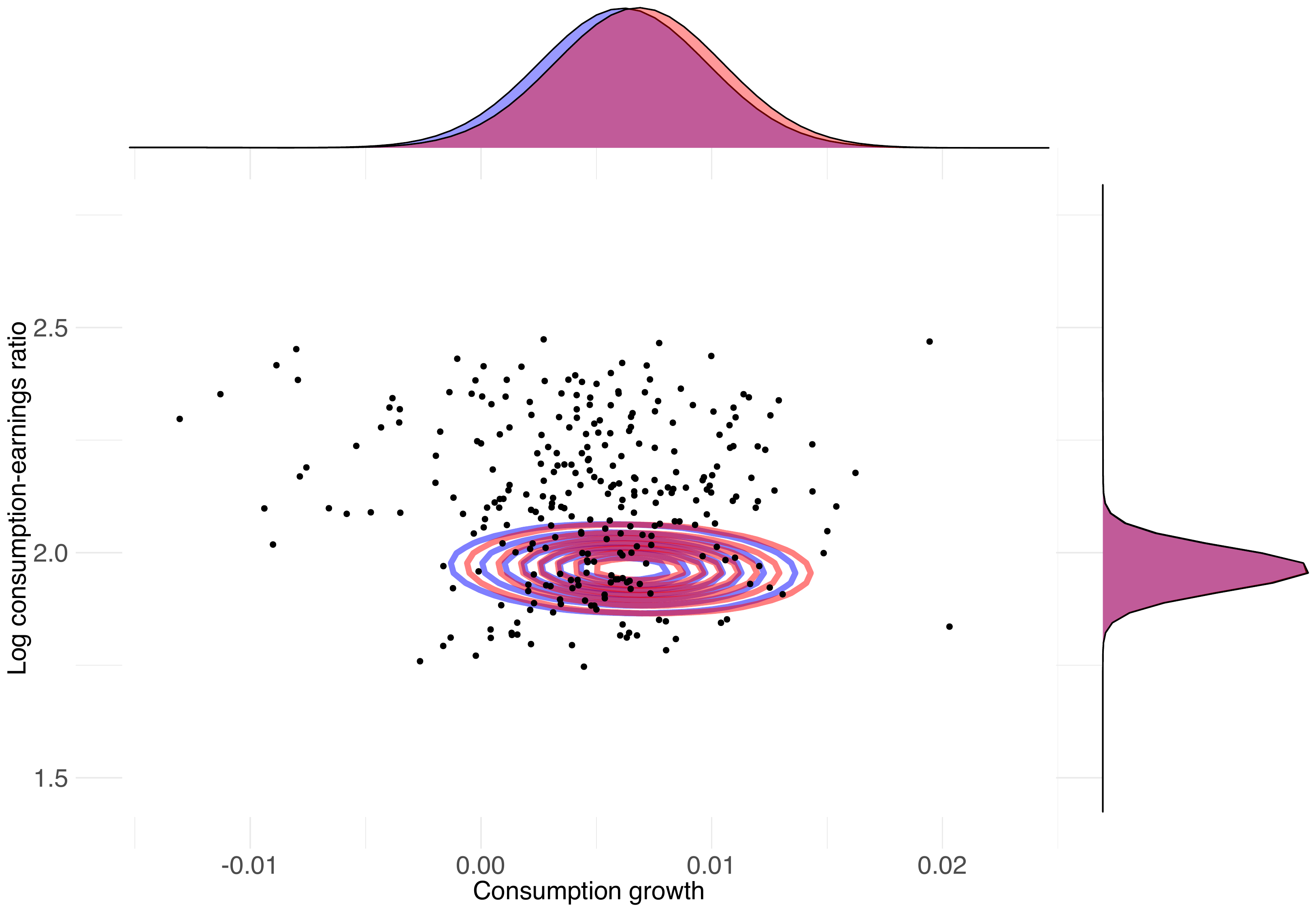}
\parbox{12cm}{\caption{\small\label{f:s1_np}  Mixture specification: Conditional distribution of $X_{t+1}$ given $X_t$ under the benchmark (red) and worst-case (blue) models in the ``good'' state. Data points are plotted in the center.}}
\end{center}
\end{figure}

\begin{figure}[p]
\begin{center}
\includegraphics[trim = 0cm 0cm 0cm 0cm, clip, width = .8\textwidth]{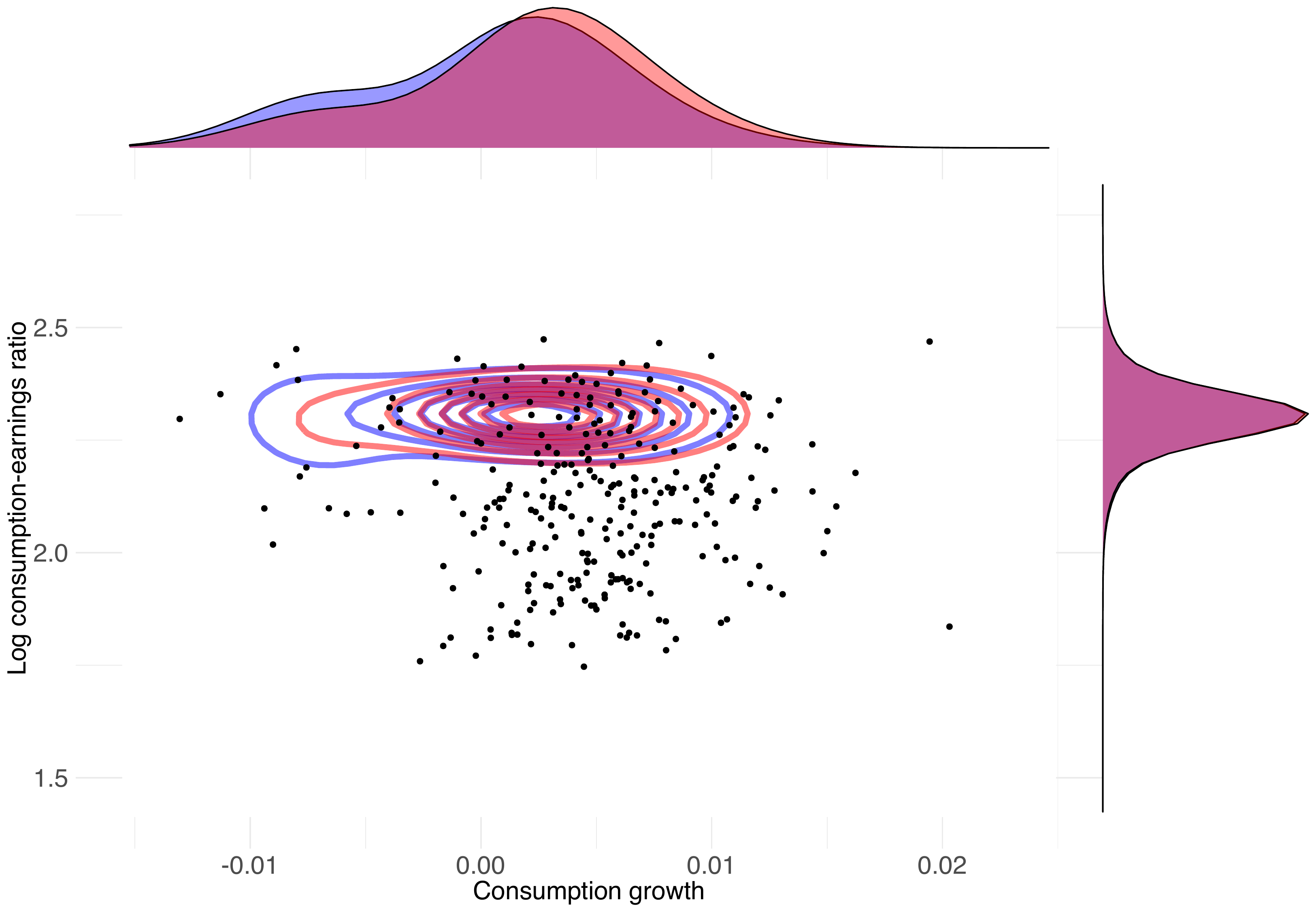}
\parbox{12cm}{\caption{\small\label{f:s3_np}  Mixture specification: Conditional distribution of $X_{t+1}$ given $X_t$ under the benchmark (red) and worst-case (blue) models in the ``bad'' state.}}
\end{center}
\end{figure}

\begin{figure}[p]
\begin{center}
\includegraphics[trim = 0cm 0cm 0cm 0cm, clip, width = .8\textwidth]{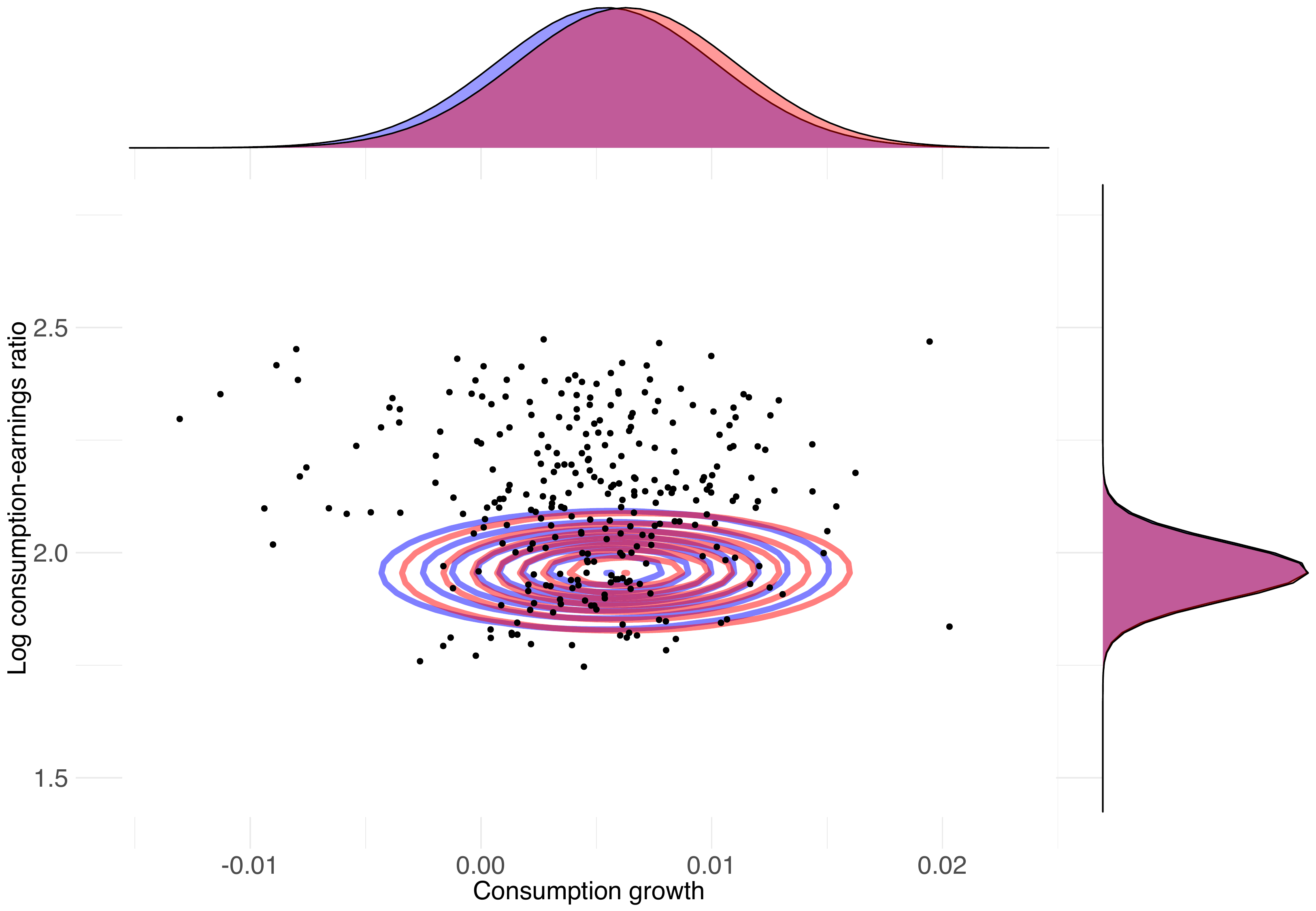}
\parbox{12cm}{\caption{\small\label{f:s1_lg}  LG specification: Conditional distribution of $X_{t+1}$ given $X_t$ under the benchmark (red contours and marginals) and worst-case (blue contours and marginals) models in a ``good'' state.}}
\end{center}
\end{figure}

\begin{figure}[p]
\begin{center}
\includegraphics[trim = 0cm 0cm 0cm 0cm, clip, width = .8\textwidth]{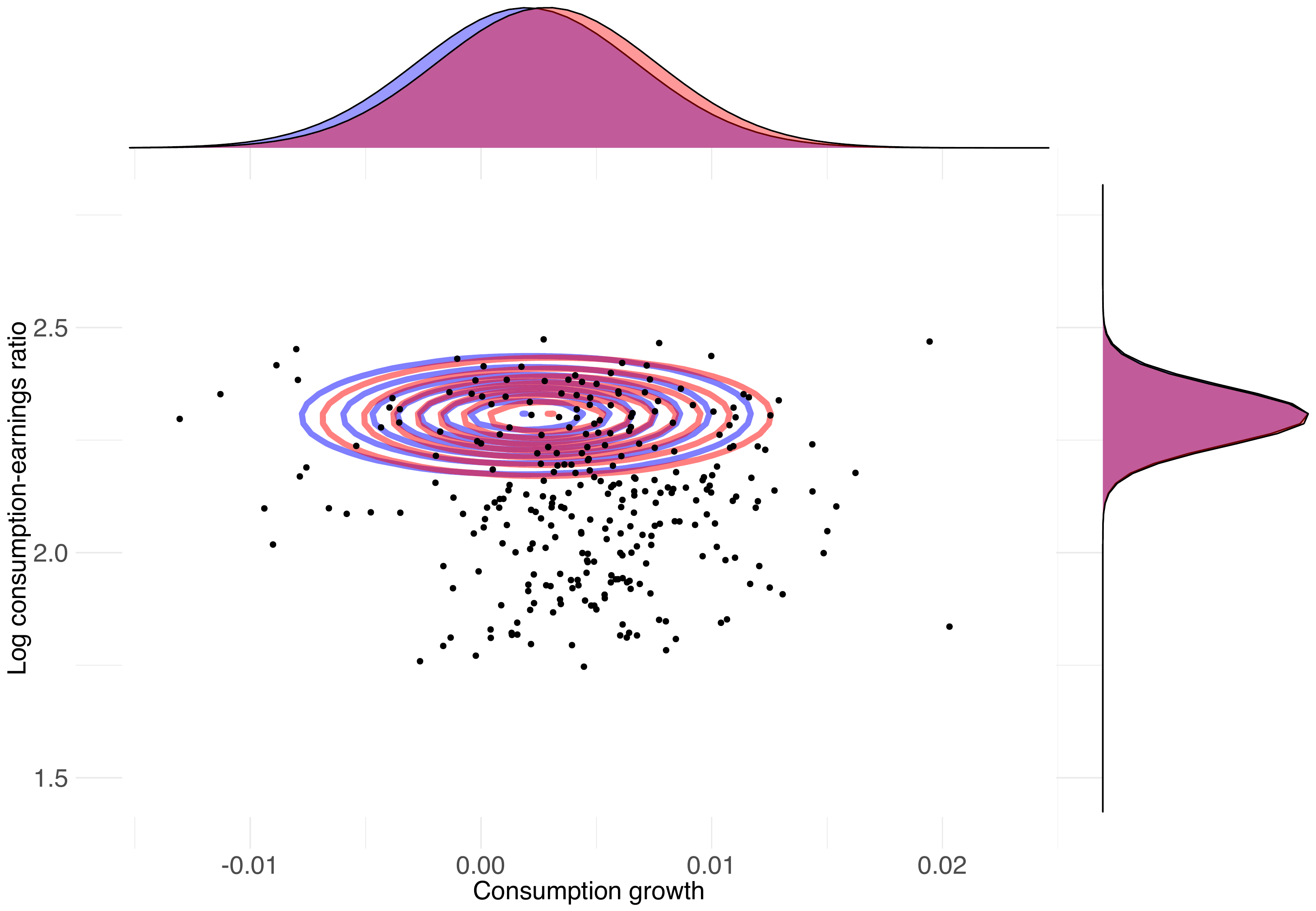}
\parbox{12cm}{\caption{\small\label{f:s3_lg}  LG specification: Conditional distribution of $X_{t+1}$ given $X_t$ under the benchmark (red contours and marginals) and worst-case (blue contours and marginals) models in a ``bad'' state.}}
\end{center}
\end{figure}

Time-variation in the benchmark and worst-case model in the mixture specification also leads to interesting properties of the implied stationary distribution, which is displayed in Figure \ref{f:q_np}. Relative to the benchmark model, the stationary distribution under the worst-case model has a much fatter left tail for consumption growth---a long-run consequence of the distortion exhibited in Figure \ref{f:s3_np}---and a slightly higher mean for the consumption-earnings ratio. 

\begin{figure}[p]
\begin{center}
\includegraphics[trim = 0cm 0cm 0cm 0cm, clip, width = .8\textwidth]{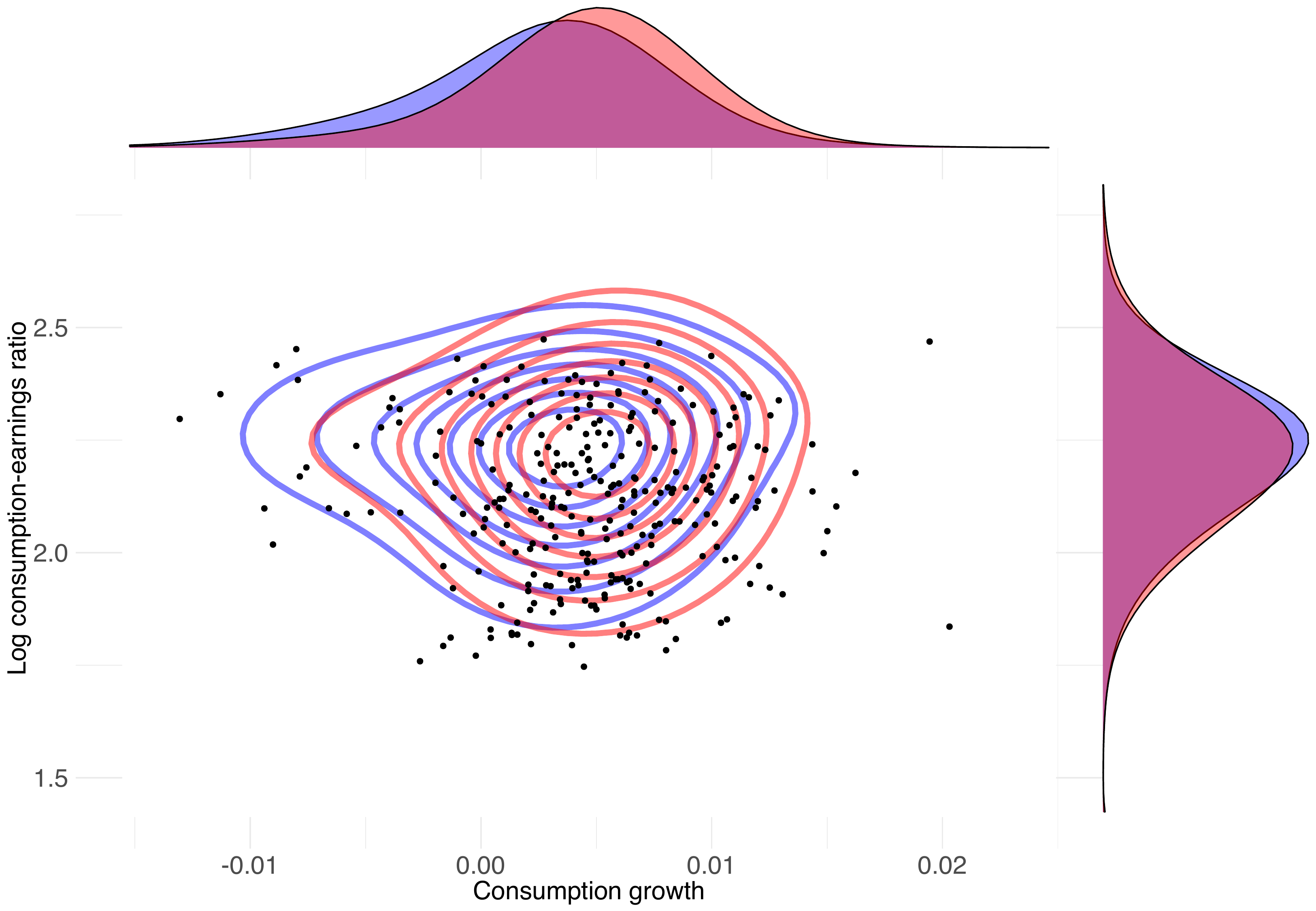}
\parbox{12cm}{\caption{\small\label{f:q_np} Mixture specification: Stationary distribution the benchmark (red contours and marginals) and worst-case (blue contours and marginals) models.}}
\end{center}
\end{figure}

\begin{figure}[p]
\begin{center}
\includegraphics[trim = 0cm 0cm 0cm 0cm, clip, width = .8\textwidth]{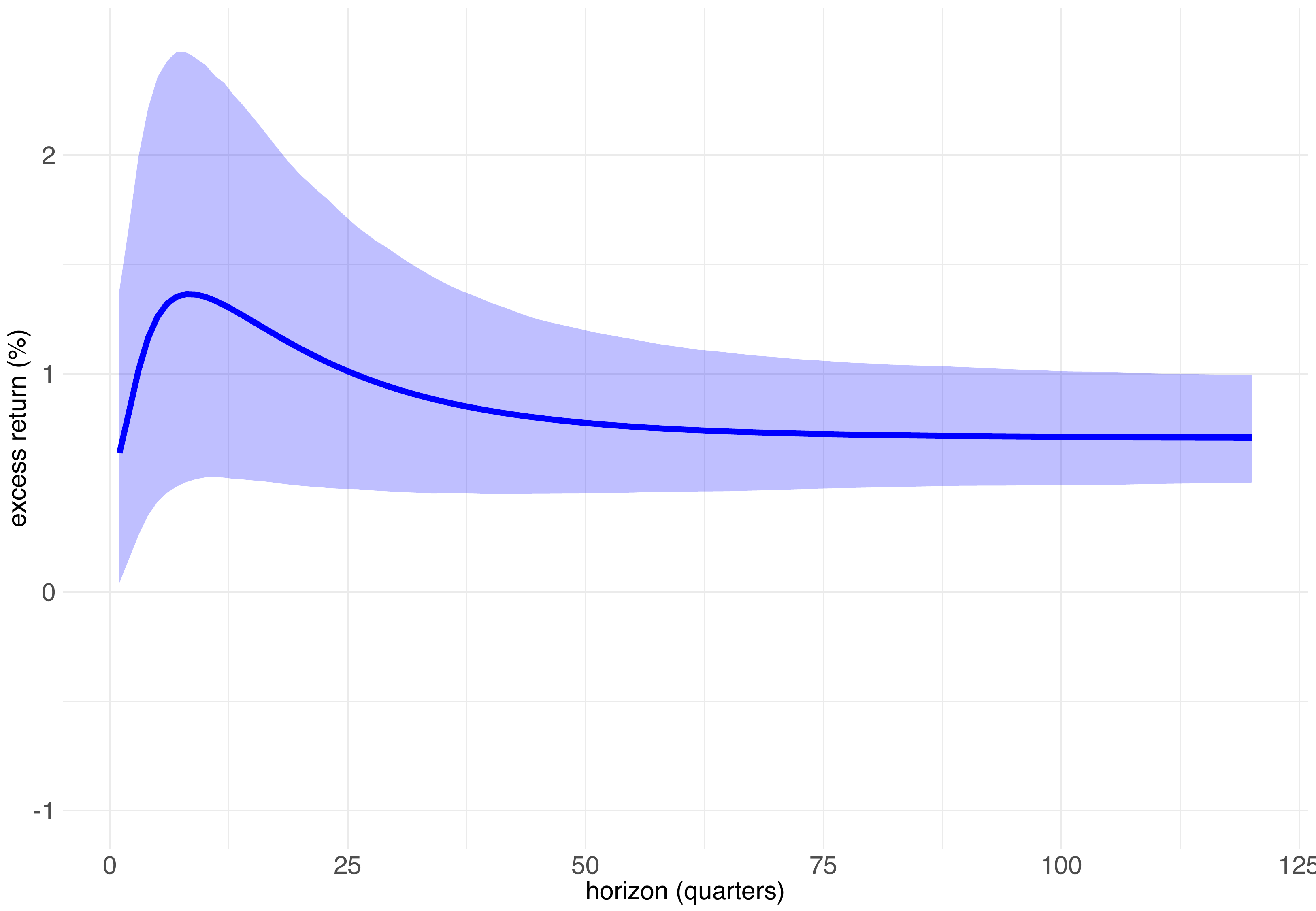}
\parbox{12cm}{\caption{\small\label{f:ts1} Term structures of excess returns on earnings strips in the ``good'' state. Solid lines are posterior means, shaded bands are 90\% pointwise credible sets.}}
\end{center}
\end{figure}

\begin{figure}[p]
\begin{center}
\includegraphics[trim = 0cm 0cm 0cm 0cm, clip, width = .8\textwidth]{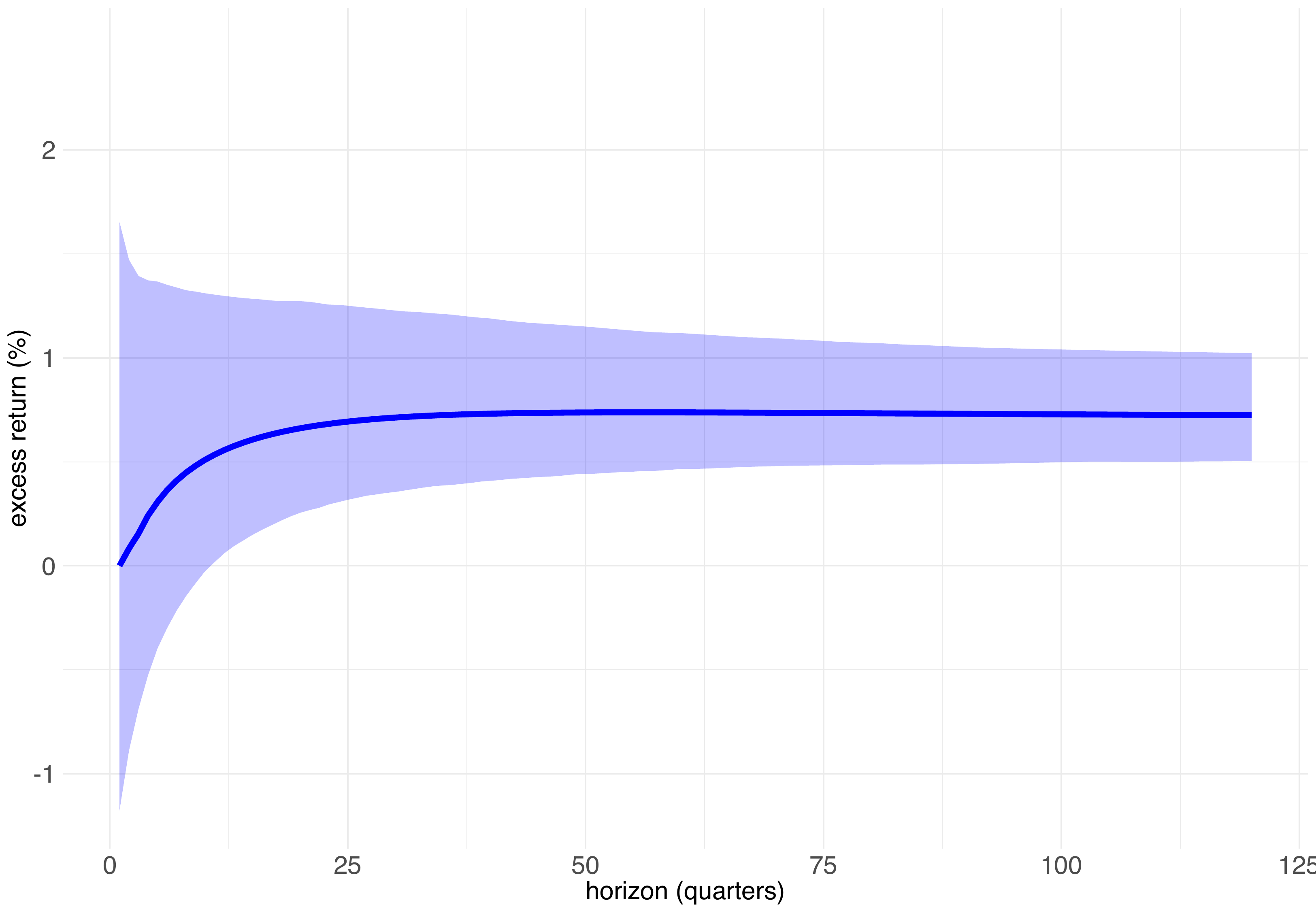}
\parbox{12cm}{\caption{\small\label{f:ts3} Term structures of excess returns on earnings strips in the ``bad'' state. Solid lines are posterior means, shaded bands are 90\% pointwise credible sets.}}
\end{center}
\end{figure}

\begin{figure}[p]
\begin{center}
\includegraphics[trim = 0cm 0cm 0cm 0cm, clip, width = .8\textwidth]{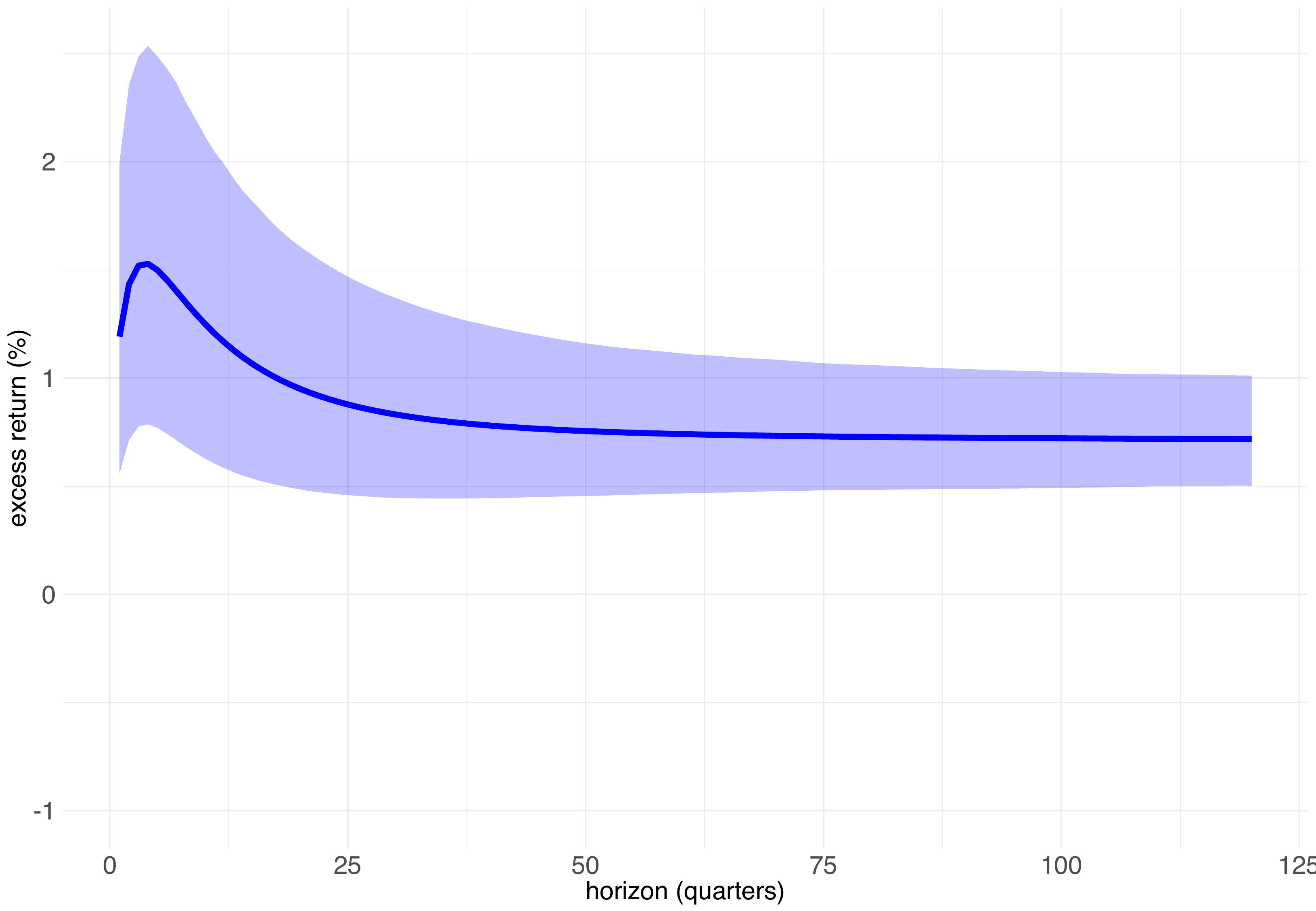}
\parbox{12cm}{\caption{\small\label{f:ts2} Term structures of excess returns on earnings strips in the ``average'' state. Solid lines are posterior means, shaded bands are 90\% pointwise credible sets.}}
\end{center}
\end{figure}

\subsection{Implications for asset prices}

To explore the implications of the model for asset prices, we compute term structures of excess returns on earnings strips in different states.\footnote{I.e. $\log \mb E^Q[  E_{t+\tau}|X_t] - \log  \mb E_v[ \beta^\tau (C_t/C_{t+\tau}) E_{t+\tau}|X_t] + \log  \mb E_v[ \beta^\tau (C_t/C_{t+\tau}) |X_t]$ where $\tau$ is the horizon $E_t$ denotes earnings at date $t$. The term $\beta^\tau (C_t/C_{t+\tau})$ is the DM's stochastic discount factor for pricing claims to date $t+\tau$ payoffs at date $t$. The term $\log  \mb E_v[ \beta^\tau (C_t/C_{t+\tau}) |X_t]$ corrects for the risk-free rate.} The posterior means in three states are plotted in Figure \ref{f:ts1} (good state), \ref{f:ts3} (bad state), and \ref{f:ts2} (an ``average'' state, where both state variables equal their mean). Each plot presents the posterior mean excess return in solid lines together with horizon-wise 90\% credible sets as shaded regions. As can be seen, the term structures are time-varying, with a hump shape in the good state, an upwards-sloping shape in the bad state, and a downwards-sloping shape in the average state.\footnote{As the environment is ergodic, however, the long-end of the term structure remains fixed at around 0.75\%. } This time-variation at the short end cannot be generated in LG benchmark specifications in this setting. \cite{HS2017sets} provide a dynamic extension of max-min preferences in which agents consider both parametric and nonparametric families of models. Their extension of max-min preferences can generate state dependence in worst-case models and uncertainty prices even in LG environments.

\subsection{Macroeconomic uncertainty}

Finally, we compare three time series related to the model with other notions of macroeconomic uncertainty.  The first series is the difference between the conditional mean of consumption growth under the benchmark and worst-case models, as in Figure \ref{f:wc}. The second is the continuation entropy a function of the realized state, i.e. $\Gamma(X_t)$. Both of these series are constant with a LG benchmark model but are time-varying for the mixture specification. The third series is the entropy of the experts' mixture weights, i.e. $-\sum_{k=1}^K w_k(X_t) \log w_k(X_t)$. Each of these series are distinct in nature: the first represents time-varying pessimism. The second represents the size (in terms of discounted relative entropy) of the set of models over which the agent is maximizing worst-case utility. The third series measures the dispersion in the forecast weights in the benchmark model. This third series may be interpreted as uncertainty among the mixture components, and is bounded between zero (where the weight is essentially one for one component and zero for all others) and $\log K$, when all components have equal weight.

\begin{figure}[t]
\begin{center}
\includegraphics[trim = 0cm 0cm 0cm 0cm, clip, width = .8\textwidth]{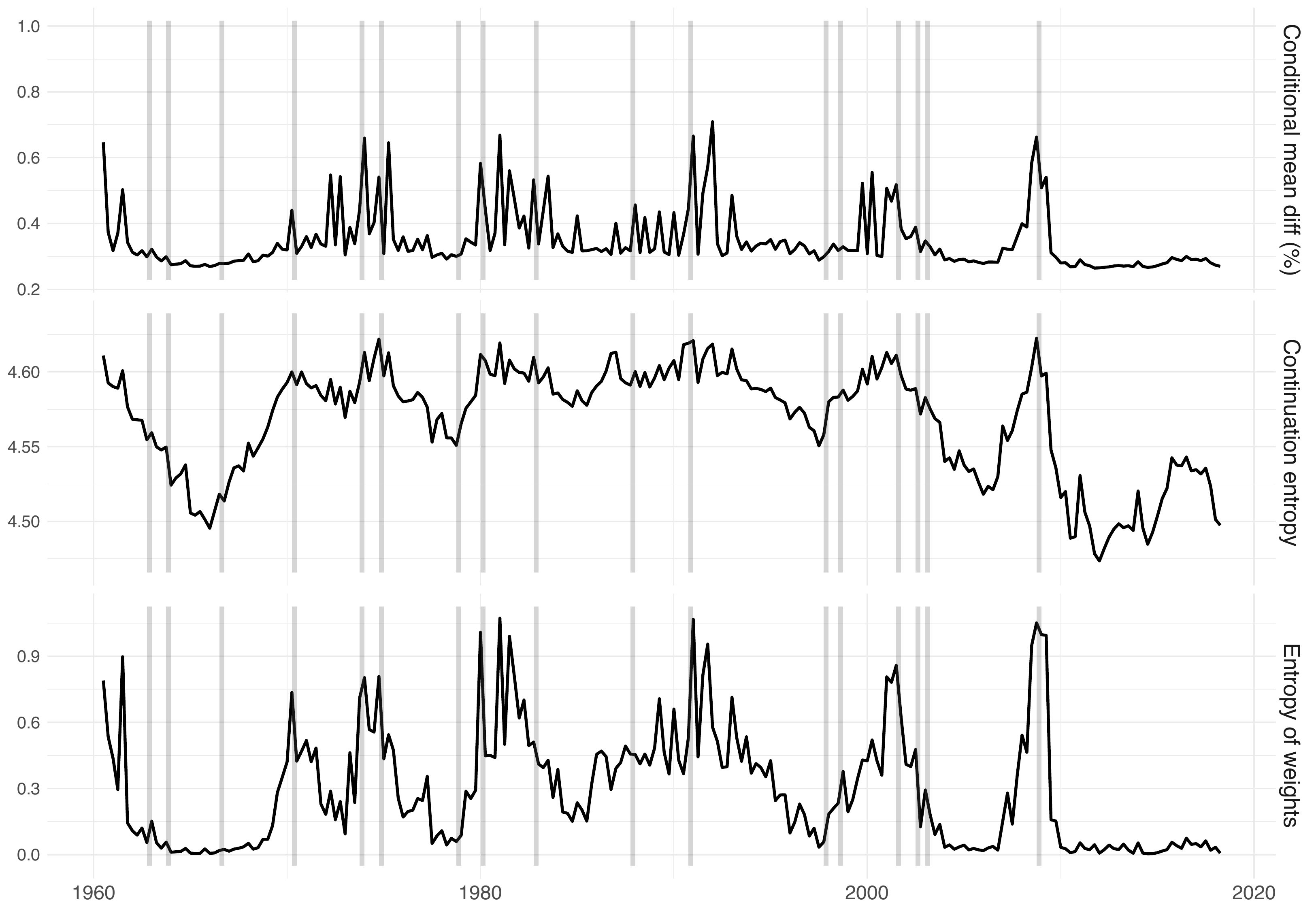}
\parbox{12cm}{\caption{\small\label{f:rr} Time series of the posterior means of the difference between the conditional mean of consumption growth under the worst-case and benchmark models, continuation entropy, and entropy of experts' weights in the benchmark model. \cite{Bloom2009} major stock-market volatility shock dates are indicated as shaded regions.}}
\end{center}
\end{figure}

Figure \ref{f:rr} plots three time series for the mixture specification alongside the (maximum) major stock-market volatility shock dates from  \cite{Bloom2009}. 
Each of the three series peaks around the \cite{Bloom2009} uncertainty dates in the late 1970s, early 80s and 90s, and 2008, but behave differently around the other dates.
In particular, comparing with Figure \ref{f:data}, fluctuations in the continuation entropy appear driven largely by fluctuations in the consumption-earnings ratio whereas the other series appear driven by both low- and high-frequency state variables. The \cite{Bloom2009} dates are essentially dates of stock market volatility shocks. The correlations of the three series with the CBOE S\&P 100 Volatility Index (VXO) over the period 1986Q1 to 2018Q3 is 0.38 for the first two series (pessimism and continuation entropy) and 0.45 for the third (entropy of mixing weights).

Another popular uncertainty measure are the \cite{JLN} indices of macroeconomic uncertainty. The correlations of the indices of uncertainty of horizons 1, 3, and 12 months over the period 1960Q1 to 2018Q3 with our first uncertainty measure (pessimism) are all around 0.48, correlations with the second (continuation entropy) vary between 0.38 and 0.44, and correlations with our third measure (entropy of mixing weights) are all around 0.56. Correlations with the \cite{JLN} indices of financial uncertainty display similar patterns but are weaker.

\section{Conclusion}

This paper studies identification and estimation of a class of dynamic models where the DM is endowed with multiplier or constraint preferences as in the ``robustness'' literature. The DM entertains a set of models surrounding a benchmark model that he or she fears may be misspecified. Decisions are evaluated under a worst-case model delivering lowest utility within this set. This paper derives primitive conditions for identification of the DM's worst-case model and preference parameters. The key step in the identification analysis is to establish existence and uniqueness of the DM's continuation value function allowing for unbounded statespace and unbounded utilities, both of which are important in applications. Extensions to models featuring other types of ambiguity aversion are discussed. For estimation, a perturbation result is derived which provides a necessary and sufficient condition for consistent estimation of continuation values and the worst-case model and allows convergence rates of estimators to be characterized. The result is also useful for computing approximate value functions in models for which no closed form solution exists by perturbing simpler models. An empirical application studies an endowment economy where the DM's benchmark model aggregates experts' forecasting models. Asset pricing consequences are discussed and some  connections are drawn with the literature on macroeconomic uncertainty. Extensions of some results to models with learning have been sketched; we plan to pursue this in more detail going forwards.

\singlespacing

\putbib

\end{bibunit}

\newpage

\begin{bibunit}

\appendix

\section{Background material on Orlicz spaces}\label{ax:orlicz}

Let $\phi_r(x) = e^{x^r}-1$ for $r \geq 1$ and let $Q_0$ denote a probability measure on $(\mc X,\mcr X)$. The Luxemburg norm $\|\cdot\|_{\phi_r} = \|\cdot\|_{L^{\phi_r}(Q_0)}$ of a measurable function $f : \mc X \to \mb R$ is defined as
\begin{align*}
 \|f\|_{\phi_r} = \inf\Big\{ c > 0 : \mb E^{Q_0}[\phi_r(|f(X_0)|/c)] \leq 1 \Big\} \,.
\end{align*}
Let $L^{\phi_r} = L^{\phi_r}(Q_0)$ denote (the equivalence class of) all measurable $f : \mc X \to \mb R$ for which $\|f\|_{\phi_r} < \infty$ and let $E^{\phi_r} = E^{\phi_r}(Q_0) = \{f \in L^{\phi_r} : \mb E^{Q_0}[\phi_r(|f(X_0)|/c)]  < \infty$ for each $c > 0\}$. The spaces $L^{\phi_r}$ and $E^{\phi_r}$ are (nonseparable and separable) Banach spaces when equipped with the norm $\|\cdot\|_{\phi_r}$. The class $L^{\phi_r}$ is an \emph{Orlicz class} and the subset $E^{\phi_r}$ is its \emph{Orlicz heart}, which is the closure of $L^\infty(Q_0)$ in $L^{\phi_r}$. Note that $E^{\phi_r}$ is a proper subset of $L^{\phi_r}$.  We also have the continuous embeddings $L^\infty \hookrightarrow E^{\phi_r} \hookrightarrow L^{\phi_r} \hookrightarrow E^{\phi_s} \hookrightarrow L^{\phi_s} \hookrightarrow L^p$ for each $1 \leq s < r < \infty$ and $1 \leq p < \infty$. We refer the reader to Section 10 of \cite{KrasRuti} for further details. The norms of the embeddings are bounded as follows:
\begin{align*}
 \|f\|_{\phi_{r_1}} & \leq  (\log 2)^{1/r_2-1/r_1}\|f\|_{\phi_{r_2}} \mbox{ if $r_1 \leq r_2$} &
 \|f\|_p & \leq p! \|f\|_{\phi_1} \mbox{ if $1 \leq p < \infty$}
\end{align*}
\cite[p. 95]{vdVW} where $\|\cdot\|_p$ denotes the $L^p(Q_0)$ norm.

\section{A general existence and uniqueness result}\label{ax:id:gen}

Let $(\mc X,\mcr X,\mu)$ be a $\sigma$-finite measure space, and let $\mc L$ denote the (equivalence class of) all measurable $f : \mc X \to \mb R$ for which $\|f\|_\psi < \infty$, where 
\[
 \|f\|_\psi = \inf \left\{ c > 0 : \int \psi(|f(x)|/c) \, \mr d \mu(x) \leq 1 \right\}
\]
for some monotone, strictly convex $\psi : \mb R_+ \to \mb R_+$ with $\psi(0) = 0$ and $\psi(x)/x \to +\infty$ as $x \to +\infty$. The function $\psi(x)$ could be $\phi_r(x)$ as above or $x^p$ with $1 < p < \infty$ to accommodate $L^p$ spaces. Let $\mc E_0 = \{ f \in \mc L : \int \psi(|f(x)|/c) \, \mr d \mu(x) < \infty$ for each $c > 0 \}$ denote the Orlicz heart of $\mc L$. Here $\mc E_0 = L^p(\mu)$ if $x^p$ or $\mc E_0 = E^{\phi_r}(\mu)$ with $\phi_r$ as above. 

Consider a (nonlinear) operator $\mb T : \mc E \to \mc E$ where $\mc E \subseteq \mc E_0$ is a closed linear subspace of $\mc E_0$. Write $f \geq g$ if $f(x) \geq g(x)$ holds for $\mu$-almost every $x$. Say that $\mb T$ is \emph{monotone} (or \emph{isotone}) if $\mb T f \geq \mb T g$ whenever $f \geq g$ and that it is \emph{convex} (or \emph{order-convex}) if $\mb T(\tau f + (1-\tau)g) \leq \tau \mb T f + (1-\tau) \mb T g$ for any $f,g \in \mc E$ and $\tau \in [0,1]$. A bounded linear operator $\mb D_f : \mc E \to \mc E$ is a \emph{subgradient} of $\mb T$ at $f$ if $\mb T g - \mb T f \geq \mb D_f (g-f)$ for each $g \in \mc E$. We say that a decreasing sequence of functions $\{v_n\}_{n \in \mb N}\subset \mc E$ is bounded from below by $\ul v \in \mc E$ if $\liminf_{n \to \infty} v_n \geq \ul v$. Let $\rho(\mb D_f;\mc E)$ denote the spectral radius of $\mb D_f : \mc E \to \mc E$. Let $\mb T^n \ol v$ denote $\mb T$ applied $n$ times in succession to $\ol v$.

\begin{proposition}\label{p:exun}
(i) Existence: Let $\mb T$ be continuous and monotone, let there exist $\ol v \in \mc E$ such that $\mb T \ol v \leq \ol v$, and let the sequence $\mb T^n \ol v$ be bounded from below by some $\ul v$ in $\mc E$. Then: $\mb T$ has a fixed point $v \in \mc E$. \\
(ii) Uniqueness: Let $\mb T$ be convex and at each fixed point $v \in \mc E$ of $\mb T$, let the subgradient $\mb D_v$ be monotone with $\rho(\mb D_v;\mc E) < 1$. Then: $\mb T$ has at most one fixed point in $\mc E$.
\end{proposition}

\begin{remark}
It follows from the proof of Proposition \ref{p:exun}(i) that $\ul v \leq v \leq \ol v$ and that fixed-point iteration on $\ol v$ will converge to $v$.
\end{remark}

\section{Additional results for Section \ref{s:identification}}\label{ax:id}

\subsection{$\mb T$ is a contraction on the space of bounded functions}

Recall that a (linear or nonlinear) operator $\mb K : E^{\phi_s} \to E^{\phi_s}$ is a \emph{contraction mapping} if there exists $\tau \in [0,1)$ such that  
\begin{equation}\label{e:cm}
 \|\mb K f - \mb Kg\|_{\phi_s} \leq \tau \|f - g\|_{\phi_s}
\end{equation}
 for each $f,g \in E^{\phi_s}$, in which case $\tau$ is referred to as the \emph{modulus of contraction}. A nonlinear operator $\mb K : E^{\phi_s} \to E^{\phi_s}$ is a \emph{local contraction mapping} if for each $h \in E^{\phi_s}$ there exists a neighborhood $N_h$ of $h$ and a constant $\tau = \tau_h \in [0,1)$ such that (\ref{e:cm}) holds for all $f,g \in N_h$ .

The following result is a straightforward application of Blackwell's conditions (see Theorem 3.3 in \cite{SLP}).

\begin{proposition}\label{prop-W-unique-bdd}
If $\mb T : B(\mc X) \to B(\mc X)$ then $\mb T$ is a contraction mapping of modulus $\beta$ and therefore has a unique fixed point $v \in B(\mc X)$. 
\end{proposition}

\begin{remark} \normalfont
A sufficient condition for  $\mb T : B(\mc X) \to B(\mc X)$ is that there exists a finite positive constant $C$ such that $C^{-1} \leq \mb E^Q [ e^{ \alpha u(X_t,X_{t+1})} | X_t=x ] \leq C$ holds for all $x \in \mc X$.
\end{remark}

\subsection{$\mb T$ is not a contraction when state variables are unbounded}\label{s:noncontract}

This section provides examples to show that $\mb T$ and $\mb D_v$ are not necessarily contraction mappings when the support of $X$ is unbounded. 
The examples are presented within the context of the LG environment described at the end of Section \ref{s:example}, for which Theorem \ref{t-id-W} implies that $v = a + b x$ is the unique fixed point of $\mb T$ in $E^{\phi_s}$ for each $1 < s < 2$.   To simplify the calculations, let $d=1$, $\mu = 0$, $A = 0$, $\sigma = 1$ and $\mu^* = \alpha(\beta \lambda_0 + \lambda_1) \neq 0$. Then $X_t$ is i.i.d. $N(\mu^*,1)$ under the worst-case model. Let $\Phi$ denote the standard normal c.d.f. 

We first show that $\mb T : E^{\phi_s} \to E^{\phi_s}$ neither a contraction nor a local contraction for any $s \geq 1$. Take $h(x) = \epsilon\ind\{x > \delta\}$ for $\epsilon > 0$ and $\delta \in \mb R$. Then $\mb D_v h = \beta \epsilon \Phi(\mu^* - \delta)$ and:
\[
 \mb E^{Q_0}\left[ e^{|h(X_t)/c|^s}\right] = \Phi(\delta) + e^{(\epsilon/c)^s}(1-\Phi(\delta))
\]
from which it follows that $\|h\|_{\phi_s} = \epsilon (\log(1+\frac{1}{\Phi(-\delta)}))^{-1/s} =: \epsilon g_s(\delta)$. On the other hand: 
\[
 \|\mb T (v+h) - \mb T v \|_{\phi_s} \geq  \|\mb D_v h\|_{\phi_s} = \epsilon \beta \Phi(\mu^* - \delta) (\log 2)^{-1/s} \,.
\]
The function $g_s(\delta)$ is monotone and converges to $(\log 2)^{-1/s}$ as $\delta \to -\infty$ and to zero as $\delta \to +\infty$. We may therefore choose $\delta$ and $\mu^*$ such that $\beta \Phi(\mu^* - \delta) (\log 2)^{-1/s} > g_s(\delta)$. For such values of $\delta$ and $\mu^*$, we have $\|\mb T (v+h) - \mb Tv\|_{\phi_s} > \|h\|_{\phi_s}$. Therefore, $\mb T$ is not a contraction on $E^{\phi_s}$. As this is inequality holds for every $\epsilon > 0$, $\mb T$ is not a local contraction either.

We now show that $\mb D_v : E^{\phi_s} \to E^{\phi_s}$ is not necessarily a contraction for any $s \geq 1$. Take $h(x) = x$. Then $\mb E^{Q_0}[ e^{|h(X_t)/c|^2}] = \sqrt{\frac{c^2}{c^2-2}}$ for $c > \sqrt 2$ and so $\|h\|_{\phi_2} = \sqrt{\frac{ 8}{ 3}}$ and $\|h\|_{\phi_s} \leq  (\log 2)^{1/2-1/s} \sqrt{\frac{ 8}{ 3}}$ for each $1 \leq s < 2$ (see Appendix \ref{ax:orlicz}). On the other hand, $\mb D_v h = \beta \mu^*$ so $\|\mb D_v h\|_{\phi_s} = \beta \mu^* (\log 2)^{-1/s}$ for each $1 \leq s < 2$. We may therefore choose $\mu^*$ sufficiently large that $\| \mb D_v h\|_{\phi_s}  > \|h\|_{\phi_s}$ and hence $\mb D_v$ is not a contraction on $E^{\phi_s}$.

\subsection{Multiple fixed points and truncation of the statespace}\label{s:nonunique}

This section provides an example to show that artificially truncating the support of an unbounded process can yield misleading conclusions as to uniqueness of continuation values.

Suppose $u(X_t,X_{t+1}) = X_t$ where $X$ is an autoregressive gamma (ARG) process with parameters $(c_1,c_2,c_3)$ where $c_1,c_2,c_3 > 0$ and $c_1c_2 < 1$. 
The log conditional moment generating function of the ARG$(c_1,c_2,c_3)$ process is
\[
 \log \mb E^{Q}[e^{sX_{t+1}}|X_t=x] = \frac{c_1 c_2 s x}{1-sc_1} - c_3 \log(1-s c_1)
\]
provided $s < \frac{1}{c_1}$. As $c_1c_2 < 1$, the process $X$ is stationary and ergodic and the stationary distribution of $X_t$ is a Gamma distribution. The function $u$ belongs to $L^{\phi_1}_2$ but does not belong to $E^{\phi_1}_2$, violating Assumption U. 

Conjecture a fixed point of the form $v(x) = a + bx$. Substituting into the above expression for the conditional cumulant generating function yields:
\begin{align*}
 \mb T v(x) & = \beta \log \mb E^Q[ e^{a + b X_{t+1} + \alpha X_t }|X_t = x] \\
 & = \beta a + \alpha \beta x + \frac{\beta b c_1 c_2 x }{1-bc_1} - \beta c_3 \log (1-b c_1)
\end{align*}
therefore:
\begin{align*}
 b & = \frac{1-\beta c_1(c_2-\alpha) \pm \sqrt{(1-\beta c_1(c_2-\alpha))^2-4\alpha \beta c_1}}{2c_1} \\
 a & = -\frac{\beta c_3 \log (1-b c_1)}{1-\beta}
\end{align*}
provided $(1-\beta c_1(c_2-\alpha))^2-4\alpha \beta c_1 \geq 0$ and $1-bc_1 > 0$. For parameterizations of $\alpha, \beta, c_1, c_2, c_3$ such that the discriminant is strictly positive and both solutions for $b$ satisfy the inequality $1-bc_1 > 0$, the operator $\mb T$ has two fixed points of the form $v(x) = a + b x$. Both of these fixed points will belong to the space $L^{\phi_1}$ but not $E^{\phi_1}$. 

Suppose that  the support of $X$ was truncated to $[0,\ol x]$ for some $\ol x < \infty$. A truncated transition kernel $Q_{\ol x}(x'|x)$ may be constructed by setting $Q_{\ol x}(x'|x) \propto Q(x'|x)\ind\{0 \leq x' \leq \ol x\}$ for each $x \in [0,\ol x]$. As $\mb E^{Q_{\ol x}} [ e^{\alpha u(X_t,X_{t+1})} | X_t=x ] = \exp(\alpha x)$ would be bounded between $1$ and $\exp(\alpha \ol x)$, Proposition \ref{prop-W-unique-bdd} implies that the truncated problem has a unique fixed point in $B([0,\ol x])$ yet the actual problem has (at least) two fixed points in $L^{\phi_1}$. Therefore, uniqueness  in the truncated problem, even for an arbitrarily large $\ol x$, does not  imply uniqueness in the actual problem.

\section{Dynamic discrete choice with unbounded utilities} \label{ax:ddc}

Dynamic discrete choice (DDC) models following \cite{Rust1987} are widely used throughout applied microeconomics, industrial organization, marketing and elsewhere. Under a conventional assumption on the distribution of utility shocks, the value function recursion in infinite-horizon DDC models has a similar structure to the recursion of a robust decision maker, involving the composition of logarithms, expectations, and exponentials. In this appendix, Proposition \ref{p:exun} is applied to establish existence and uniqueness of the value function for infinite-horizon DDC models under weaker conditions than typically used in the literature. 

Before introducing the result, the DDC framework following \cite{Rust1987} is first briefly described to fix ideas and notation. At each date $t \in T$, an agent chooses among $D$ discrete alternatives indexed by $d \in \{1,2,\ldots,D\}$ to maximize the expected present discounted value of utility. The flow utility from choosing action $d$ is 
\[
 u(d,X_t,\varepsilon_t;\theta) = u_d(X_t;\theta) + \varepsilon_{dt}
\]
where $X_t$ is a state vector that is observed by the econometrician and agent, $\theta$ is a vector of unknown parameters (we omit dependence on $\theta$ in what follows), and the vector $\varepsilon_t = (\varepsilon_{1t},\ldots,\varepsilon_{Dt})'$ is a vector of utility shocks that are unobserved by the econometrician but observed by the agent. As in much of the literature, assume that $\varepsilon_t$ is i.i.d. over time with each component drawn independently from a type-I extreme value (standard Gumbel) distribution, and that the controlled Markov process $X$ has a conditional distribution which factorizes as:
\[
 F(X_{t+1},\varepsilon_{t+1}|X_t=x,\varepsilon_t=\varepsilon,D_t=d) = M(X_{t+1}|x,d) G(\varepsilon_{t+1})
\]
for every $(x,\varepsilon,d)$, where $M$ is a time-invariant Markov transition kernel and $G$ denotes the assumed distribution of $\varepsilon_{t+1}$. Let $\beta \in (0,1)$ denote the agent's time preference parameter. The agent's problem may be expressed as the Bellman equation:
\[
 v(X_t) = \mb E^G \left[ \left. \max_{d} \left( u_d(X_t) + \varepsilon_{dt} + \beta \mb E^M \left[ \left. v(X_{t+1}) \right| X_t,D_t = d \right] \right) \right| X_t \right] 
\]
where $v$ is the agent's ex ante value function and $\mb E^G[\cdot]$ denotes expectation over $\varepsilon$ under $G$. In view of the parametric assumption on $G$, 
\begin{align} \label{e:ddc-bellman-1}
 v(X_t) = \log\left( \sum_{d=1}^D e^{u_d(X_t) + \beta \mb E^M \left[ \left. v(X_{t+1}) \right| X_t,D_t = d \right]} \right) + \gamma_{\mr{EM}}
\end{align}
where $\gamma_{\mr{EM}} \approx 0.5772$ is the Euler-Mascheroni constant. The recursion (\ref{e:ddc-bellman-1}) may be expressed in operator notation as $v = \mb T v$. 

The existing literature typically assumes that the support of $X$, denoted $\mc X$, is compact (often finite). \cite{Blevins2014} allows for continuous, unbounded state (and continuous choices) but requires the functions $u_d$ to be uniformly bounded. \cite{Norets2010} allows for unbounded utilities under a weighted sup norm  where the weighting function must be chosen to be compatible with utilities and the transition kernel $M$.\footnote{See Assumptions 2--4 in \cite{Norets2010}.} In each of these cases, the operator $\mb T$ is shown to be a contraction mapping on class of functions with finite sup norm or weighted sup norm. Existence and uniqueness within that class then follows by standard arguments.

We instead apply Proposition \ref{p:exun} to derive existence and uniqueness conditions for $v$ in a class of unbounded but ``thin-tailed'' functions by exploiting the monotonicity and convexity of the recursion (\ref{e:ddc-bellman-1}). In Theorems \ref{t-id-W} and \ref{t-id-W-learn} the parameter space was defined relative to the stationary distribution of the state vector (or sufficient statistic in the setting with learning). However, here the agent's optimal date-$t$ decision, say $d_t^*$, depends upon both $X_t$ and $\varepsilon_t$. Therefore, we cannot factor the transition kernel as $M(X_{t+1}|X_t,d_t^*) = M(X_{t+1}|X_t)$ as in the main text. Instead, we define the transition kernel:
\[
 Q(X_{t+1}|X_t) = \frac{1}{D} \sum_{d=1}^D M(X_{t+1}|X_t,d)\,.
\]
Note that $Q$ need not agree with the law of motion of $X$ under the agent's optimal plan. We assume that the process $X$ has a unique stationary distribution $Q_0$ under $Q$. 
This is trivially true when there is a \emph{renewal action}, say $d^*$, for which $Q(X_{t+1}|X_t,d^*)$ does not depend on $X_t$. That is, $Q(\cdot|X_t,d^*) = \nu(\cdot)$ for some distribution $\nu$. In this case, the inequality $Q(\cdot|X_t) \geq D^{-1} \nu(\cdot)$ holds for every $X_t$. This inequality verifies Doeblin's minorization condition and therefore guarantees existence of a unique stationary distribution $Q_0$ \cite[Theorem 16.2.4]{MeynTweedie}. Many models in the DDC literature do indeed have renewal choices, including the bus engine replacement model of \cite{Rust1987}, so our results necessarily encompass, but are not limited to, such models.\footnote{The existence of renewal actions allows the expression for continuation values to be differenced out from the expression for conditional choice probabilities, simplifying estimation (see, e.g., \cite{AM}). Nevertheless, existence and uniqueness of continuation values remains relevant, inter alia, for quantifying the welfare effects of policy interventions.} 

The following result establishes existence and uniqueness of continuation values in models with continuous state variables without restricting the support of such variables or requiring utilities to be bounded. Let $E^{\phi_s} = E^{\phi_s}(Q_0)$. Because of the slightly different nature of the operator here, identification is established for the space $E^{\phi_s}$ for all $1 \leq s \leq r$ allowing $r = 1$ rather than $1 < s \leq r$ with $r > 1$ as in the main text. 

\begin{theorem}\label{t-id-W-ddc}
Let $u_d \in E^{\phi_r}$ hold for some $r \geq 1$ for each $d$. Then: $\mb T$ has a fixed point $v \in E^{\phi_r}$. Moreover, $v$ is the unique fixed point of $\mb T$ in $E^{\phi_s}$ for each $1 \leq s \leq r$.
\end{theorem}

The tail condition $u_d \in E^{\phi_r}$ is trivially satisfied when the $u_d$ functions are bounded. In that case, as the space $E^{\phi_s}$ contains $L^\infty$, we establish existence and uniqueness in a larger class of functions than $B(\mc X)$. As with Theorems \ref{t-id-W} and \ref{t-id-W-learn}, the proof of Theorem \ref{t-id-W-ddc} shows that $\ul v \leq v \leq \ol v$ for known functions $\ul v,\ol v \in E^{\phi_r}$, and that fixed point iteration on $\ol v$ will converge to $v$.

\section{Perturbations towards stochastic volatility}\label{s:perturb-example}

This appendix shows how to apply Lemma \ref{l-linear}  to compute approximate solutions in models with stochastic volatility by viewing these models as perturbations of LG models. 

Consider the LG example from Section \ref{s:example}. We may extend $Q$ to have state $Z_t = (X_t',h_t)'$ where $X_t$ and $h_t$ are independent stationary stochastic processes and consumption and dividend growth remain functions of $(X_t,X_{t+1})$. The sequence of $X_t$ evolves as the LG process described in Section \ref{s:example}. The sequence of $h_t$ will represent a volatility process in the perturbed model. This process evolves as a first-order Markov process with support $\mb R_+$ under $Q$. By independence of $X_t$ and $h_t$ under $Q$ and additivity of KL-divergence for independent distributions, the affine solution $v(z) = v(x,h) = a + b'x$ remains the unique fixed point for this benchmark model. Under the worst-case model, $X_t$ and $h_t$ are independent, $h_t$ evolves as under the benchmark model, and $X_t$ is a Gaussian VAR(1)  with mean parameter shifted from $\mu$ to $\mu^*$ as before.

Consider the perturbed model, say $\hat Q$, under which
\[
 X_{t+1} = \mu + A X_t + e^{\frac{ h_t}{2}} \,\sigma \varepsilon_{t+1}  \,.
\]
So, although $X_t$ and $h_t$ are independent under $Q$ they are no longer independent under $\hat Q$. The score term of $\hat Q$ relative to $Q$ is 
\[
 \hat \eta(Z_t,Z_{t+1}) = -\frac{1}{2} ( e^{- h_t}-1) [(X_{t+1} - \mu - A X_t)'(\sigma \sigma')^{-1}(X_{t+1} - \mu - A X_t) - d]
\]
where $d = \dim(X_t)$. Both $\hat \eta$ and $\hat \ell$ belong to $L^{\phi_1}_2$. The remaining conditions of Assumption AM(b) could be verified given additional structure on the process for $h_t$, in which case it would follow by Lemma \ref{l-perturb} that there is a unique fixed point $\hat v$ in $E^{\phi_s}$ for all $1 < s < 2$. 
To compute the first-order approximation, first observe that 
\[
 \mb E_v \hat \eta (z) = - \frac{1}{2} ( e^{- h_t}-1) [(\mu^* - \mu)' (\sigma \sigma')^{-1} (\mu^* - \mu)]\,.
\]
As the law of motion of $h_t$ under $\mb E_v$ and $\mb E^Q$ are the same, we obtain
\[
 \hat v(z) \approx a + b'x - \frac{\beta}{2} \sum_{i=0}^\infty \beta^i \left( \mb E^Q[e^{- h_{t+i}}|h_t = h] - 1 \right) [(\mu^* - \mu)' (\sigma \sigma')^{-1} (\mu^* - \mu)] \,.
\]
The forward-looking expectations can be calculated in closed form for some processes, or otherwise calculated numerically.

\section{Proofs}

\subsection{Ancillary Lemmas}

\subsubsection{Basic results}

Let $L^p_2 = L^p(Q_0 \otimes Q)$ denote the space all measurable $f : \mc X^2 \to \mb R$ for which $\|f\|_{L^p(Q_0 \otimes Q)} := \mb E^{Q_0 \otimes Q}[|f(X_t,X_{t+1})|^p]^{1/p} < \infty$ To simplify notation, we drop dependence of the norm on the measure and simply write $\|\cdot\|_p$. Recall the definition of $m_v$ from equation (\ref{e-m-v}).

\begin{lemma}\label{lem-lp-m}
Let $v \in E^{\phi_1}$ and $u \in E^{\phi_1}_2$. Then: $m_v \in L^p_2$ for each $1 \leq p < \infty$.
\end{lemma}

\begin{proof}[Proof of Lemma \ref{lem-lp-m}]
By Jensen's inequality (using convexity of $x \mapsto x^{-p}$) and Cauchy-Schwarz, we may deduce:
\begin{align*}
  \mb E^{Q_0 \otimes Q} \Bigg[ \bigg| \frac{e^{v(X_{t+1})+\alpha u(X_t,X_{t+1})}}{\mb E^Q[e^{v(X_{t+1})+\alpha u(X_t,X_{t+1})}|X_t]} \bigg|^p \Bigg] 
 & \leq \mb E^{Q_0} \left[ e^{2p|v(X_{t+1})+\alpha u(X_t,X_{t+1})|} \right] 
\end{align*}
which is finite because  $v \in E^{\phi_1}$ and $u \in E^{\phi_1}_2$. 
\end{proof}

The next Lemma appears in Chapter 2.3 of the manuscript \cite{Pollard}. We include a proof here for convenience. 

\begin{lemma}\label{lem-pollard}
Let $\mb E^{Q_0}[ \phi_r(|f(X_t)|/C) ] \leq C'$ for finite constants $C > 0$ and $C' \geq 1$. Then: $\|f\|_{\phi_r} \leq CC'$.
\end{lemma}

\begin{proof}[Proof of Lemma \ref{lem-pollard}]
Take $\tau \in [0,1]$. By convexity of $\phi_r$:
\[
 \mb E^{Q_0}[ \phi_r(\tau|f(X_t)|/C) ] \leq \tau \mb E^{Q_0}[ \phi_r(|f(X_t)|/C) ] + (1-\tau) \phi_r(0) = \tau \mb E^{Q_0}[ \phi_r(|f(X_t)|/C) ] \,.
\]
The result follows by setting $\tau = 1/C'$.
\end{proof}

\subsubsection{Equivalence of spaces}

Appendix \ref{ax:orlicz} describes relations between Orlicz classes $L^{\phi_r}$ and $E^{\phi_r}$  with different $r$. Here we describe relations between Orlicz classes defined relative to different measures. For the first result, let $\mu$ and $\nu$ be two probability measures on a measurable space $(\mc X,\mcr X)$, let $\Delta = \frac{\mr d \mu}{\mr d \nu}$, and let $\| \Delta\|_{L^p(\nu)}$ denote its $L^p(\nu)$ norm.

\begin{lemma}\label{lem:embed:0}
Let $\mu \ll \nu$ and $\int \Delta^p \, \mr d \nu < \infty$ for some $p > 1$. Then: $E^{\phi_r}(\nu) \hookrightarrow E^{\phi_r}(\mu)$ and $L^{\phi_r}(\nu) \hookrightarrow L^{\phi_r}(\mu)$ for each $r \geq 1$. Then: 
\[
 \|f\|_{L^{\phi_r}(\mu)} \leq \left(\left(2^\frac{1}{q} \| \Delta\|_{L^p(\nu)} - 1 \right) \vee 1 \right) q^{\frac{1}{r}}\|f\|_{L^{\phi_r}(\nu)}
\]
where $q$ is the dual index of $p$. 
\end{lemma}

\begin{proof}[Proof of Lemma \ref{lem:embed:0}]
To see that $E^{\phi_r}(\nu) \subseteq E^{\phi_r}(\mu)$, take any $f \in E^{\phi_r}(\nu)$ and $c > 0$. Then:
\[
 \mb E^{\mu} \left[ e^{|f(X)/c|^r}\right] = \mb E^{\nu} \left[\Delta (X) e^{|f(X)/c|^r} \right] \leq \|\Delta\|_{L^p(\nu)} \mb E^{\nu} \left[ e^{\left|f(X)/(c/q^{1/r})\right|^r} \right]^{\frac{1}{q}} < \infty
\]
because $f \in E^{\phi_r}(\nu)$. For continuity of the embedding, take $f \in L^{\phi_r}(\nu)$ and $c = q^\frac{1}{r}\|f\|_{L^{\phi_r}(\nu)}$. Substituting into the above display yields:
\[
 \mb E^{\mu} [ e^{|f(X)/c|^r}] \leq 2^\frac{1}{q} \|\Delta\|_{L^p(\nu)}  < \infty \,.
\]
Thus $\|f\|_{L^{\phi_r}(\mu)} \leq ((2^\frac{1}{q}  \| \Delta\|_{L^p(\nu)}  -1 ) \vee 1) \|f\|_{L^{\phi_r}(\nu)}$ by Lemma \ref{lem-pollard}, hence $L^{\phi_r}(\mu) \subseteq L^{\phi_r}(\nu)$.
\end{proof}

\begin{lemma}\label{lem:equiv:0}
Let $\nu \ll \mu \ll \nu$, $\int \Delta^p \mr d \nu < \infty$ for some $p > 1$ and $\int \Delta^{-p'} \, \mr d \nu < \infty$ for some $p' > 0$. Then: $E^{\phi_r}(\mu)=E^{\phi_r}(\nu)$ and $L^{\phi_r}(\mu)=L^{\phi_r}(\nu)$ for each $r \geq 1$. Then: 
\begin{align*}
  \|f\|_{L^{\phi_r}(\nu)}& \leq \left( \left( 2^{\frac{1}{q'}}  \mb E^{\nu}[(\Delta(X))^{-p'}]^\frac{1}{1+p'} - 1 \right) \vee 1 \right) (q')^\frac{1}{r} \|f\|_{L^{\phi_r}(\mu)} \\
 \|f\|_{L^{\phi_r}(\mu)} & \leq \left(\left(2^\frac{1}{q} \| \Delta\|_{L^p(\nu)} - 1 \right) \vee 1 \right) q^{\frac{1}{r}}\|f\|_{L^{\phi_r}(\nu)}
\end{align*}
where $q$ is the dual index of $p$ and $q'$ is the dual index of $1+p'$.
\end{lemma}

\begin{proof}[Proof of Lemma \ref{lem:equiv:0}]
Given Lemma \ref{lem:embed:0}, we only need to prove the inclusions $E^{\phi_r}(\mu) \subseteq E^{\phi_r}(\nu)$ and $L^{\phi_r}(\mu) \subseteq L^{\phi_r}(\nu)$ and continuity of the embeddings. For any $f \in E^{\phi_r}(\mu)$:
\begin{align*}
 \mb E^{\nu} \left[ e^{|f(X)/c|^r}\right] 
 & = \mb E^{\mu} \left[(\Delta (X))^{-1} e^{|f(X)/c|^r} \right] \\
 & \leq \mb E^{\mu} \left[((\Delta (X))^{-1})^{1+p'}\right]^{\frac{1}{1+p'}}  \mb E^{\mu} \left[ e^{q'|f(X)/c|^r} \right]^\frac{1}{q'} \\
 & \leq \mb E^{\nu}\left[(\Delta(X))^{-p'} \right]^{\frac{1}{1+p'}}  \mb E^{\mu} \left[ e^{q'|f(X)/c|^r} \right]^\frac{1}{q'} < \infty
\end{align*}
because $f \in E^{\phi_r}(\mu)$. The result follows by the same arguments as the proof of Lemma \ref{lem:embed:0}.
\end{proof}

\subsubsection{Control of the spectral radius via probabilistic arguments}

Let $Q^{\otimes n}$ denote the conditional distribution of $(X_{t+1},\ldots,X_{t+n+1})$ given $X_t$ under the benchmark model. Let $m_g^{\otimes n}(X_t,\ldots,X_{t+n+1}) = \prod_{s=0}^n m_g(X_{t+s},X_{t+s+1})$ denote the change of measure between the conditional distribution of $(X_{t+1},\ldots,X_{t+n+1})$ given $X_t$ under the law of motion induced by $\mb E_g$ relative to $Q^{\otimes n}$. Recall that $\mb D_g = \beta \mb E_g$ with $\beta \in (0,1)$.

\begin{lemma}\label{lem:sr-mp}
If there exist $C \in (0,\infty)$ and $c \in (0,1-\beta)$ and $p \in (1,\infty)$ such that $\| m_g^{\otimes n}\|_p \leq C e^{(\beta+c)^{-n}}$ for each $n \geq 1$, then: $\mb D_g: L^{\phi_r} \to L^{\phi_r}$ is a continuous linear operator for each $r \geq 1$ with $\rho(\mb D_g;L^{\phi_r}) \leq \beta^\frac{r-1}{r} < 1$ for each $r > 1$ and  $\rho(\mb D_g;L^{\phi_1}) \leq \frac{\beta}{\beta + \epsilon c} < 1$ for every $\epsilon \in (0,1)$.
\end{lemma}

\begin{proof}[Proof of Lemma \ref{lem:sr-mp}]
Fix $r > 1$. Let $q$ denote the dual index of $p$. Then for any $0 \neq f \in L^{\phi_r}$:
\begin{align*}
 \mb E^{Q_0}\left[ e^{|\mb D_g^n f(X_t)/(q^\frac{1}{r}(\beta^{\frac{r-1}{r}})^n \|f\|_{\phi_r})|^r} \right] & = \mb E^{Q_0}\left[ e^{q^{-1} \beta^n |\mb E_g^n f(X_t)/\|f\|_{\phi_r}|^r} \right] \\
 & \leq \mb E^{Q_0}\left[ e^{q^{-1} |\mb E_g^n f(X_t)/\|f\|_{\phi_r}|^r} \right]^{ \beta^n} \\
 & \leq \mb E^{Q_0 \otimes Q^{\otimes n}}\left[ m_g^{\otimes n}(X_t,\ldots,X_{t+n+1})  e^{q^{-1} |f(X_{t+n+1})/\|f\|_{\phi_r}|^r} \right]^{ \beta^n} \\
 & \leq (2^\frac{1}{q} \| m_g^{\otimes n} \|_p )^{ \beta^n} \\
 & \leq (2^\frac{1}{q} C e^{(\beta+c)^{-n}} )^{\beta^n}
\end{align*}
by two applications of Jensen's inequality, iterated expectations, H\"older's inequality, and definition of $\|\cdot\|_{\phi_r}$. It follows by Lemma \ref{lem-pollard} that:
\[
 \sup_{0 \neq f \in L^{\phi_r}} \frac{\| \mb D_g^n\|_{\phi_r}}{\|f\|_{\phi_r}} \leq \left( \left( (2^\frac{1}{q} C )^{ \beta^n}  e^{\left(\frac{\beta}{\beta+c}\right)^n} - 1 \right) \vee 1 \right) q^\frac{1}{r}(\beta^\frac{r-1}{r})^n \,.
\]
As $\lim_{n \to \infty} (2^\frac{1}{q} C )^{ \beta^n}  \exp((\frac{\beta}{\beta+c})^n) < \infty$, we obtain $\rho(\mb D_g;L^{\phi_r}) \leq \beta^\frac{r-1}{r} < 1$.

Now fix any $\epsilon \in (0,1)$ and note that $\beta  < \beta + \epsilon c < \beta + c < 1$. For any $0 \neq f \in L^{\phi_1}$:
\begin{align*}
 \mb E^{Q_0}\left[ e^{|\mb D_g^n f(X_t)/(q \beta^n(\beta+\epsilon c)^{-n} \|f\|_{\phi_1})|} \right] & = \mb E^{Q_0}\left[ e^{q^{-1} (\beta+\epsilon c)^n |\mb E_g^n f(X_t)/\|f\|_{\phi_1}|} \right] \\
 & \leq \mb E^{Q_0}\left[ e^{q^{-1} |\mb E_g^n f(X_t)/\|f\|_{\phi_1}|} \right]^{(\beta+\epsilon c)^n} \\
 & \leq \mb E^{Q_0 \otimes Q^{\otimes n}}\left[ m_g^{\otimes n}(X_t,\ldots,X_{t+n+1})  e^{q^{-1} |f(X_{t+n+1})/\|f\|_{\phi_1}|} \right]^{ (\beta+\epsilon c)^n} \\
 & \leq (2^\frac{1}{q} \| m_g^{\otimes n} \|_p )^{ (\beta+\epsilon c)^n} \\
 & \leq (2^\frac{1}{q}  C e^{(\beta+c)^{-n}} )^{ (\beta+\epsilon c)^n} \,.
\end{align*}
It follows by Lemma \ref{lem-pollard} that:
\[
 \sup_{0 \neq f \in L^{\phi_1}} \frac{\| \mb D_g^n\|_{\phi_1}}{\|f\|_{\phi_1}} \leq \left( \left( (2^\frac{1}{q} C )^{ (\beta + \epsilon c)^n}  e^{\left(\frac{\beta+\epsilon c}{\beta+c}\right)^n} - 1 \right) \vee 1 \right) q \left( \frac{\beta}{\beta+\epsilon c} \right)^n \,.
\]
As $\lim_{n \to \infty} (2^\frac{1}{q} C )^{ (\beta + \epsilon c)^n}  \exp((\frac{\beta+\epsilon c}{\beta+c})^n) < \infty$, we obtain $\rho(\mb D_g;L^{\phi_1}) \leq \frac{\beta}{\beta + \epsilon c} < 1$.
\end{proof}

Lemma \ref{lem:sr-mp} implies that $(\mb I - \mb D_g)$ is continuously invertible on $L^{\phi_r}$ for each $r \geq 1$. The following Lemma bounds the operator norm of $(\mb I - \mb D_g)^{-1} : L^{\phi_r} \to L^{\phi_r}$.

\begin{lemma} \label{lem:c-bdd2}
Let the conditions of Lemma \ref{lem:sr-mp} hold and let $q$ denote the dual index of $p$. Then:
\begin{align*}
 \|(\mb I - \mb D_g)^{-1}\|_{\phi_r} & \leq \frac{q^{1/r} }{1-\beta} \left( \left( 2^{1/q}  C e^\frac{(1-\beta)(\beta + c)}{c} - 1 \right) \vee 1 \right) < \infty  \,, \mbox{ and }\\ 
 \|(\mb I - \mb D_g)^{-1}\mb D_g\|_{\phi_r} & \leq \frac{q^{1/r} }{1-\beta} \left( \left( 2^{1/q}  C e^\frac{1-\beta}{c} - 1 \right) \vee 1 \right) < \infty 
\end{align*}
for each $r \geq 1$.
\end{lemma}

\begin{proof}[Proof of Lemma \ref{lem:c-bdd2}]
Fix any $r \geq 1$. As $\rho(\mb D_g; L^{\phi_r}) < 1$, we may write $(\mb I - \mb D_g)^{-1}$ as the Neumann series $(\mb I - \mb D_g)^{-1} = \sum_{n=0}^\infty (\beta \mb E_g)^n$ and $(\mb I - \mb D_g)^{-1}\mb D_g = \sum_{n=1}^\infty (\beta \mb E_g)^n$. Then for any $f \in L^{\phi_r}$: 
\begin{align*}
 \mb E^{Q_0}\left[ e^{|(\mb I-\mb D_g)^{-1} f(X_t)/(q^{1/r}\|f\|_{\phi_r}/(1-\beta))|^r} \right] 
 & = \mb E^{Q_0}\left[ e^{|\sum_{n=0}^\infty (1-\beta)\beta^n \mb E_g^n f(X_t)/(q^{1/r}\|f\|_{\phi_r})|^r} \right] \\
 & \leq \mb E^{Q_0}\left[ e^{\sum_{n=0}^\infty (1-\beta)\beta^n |\mb E_g^n f(X_t)/(q^{1/r}\|f\|_{\phi_r})|^r} \right] \\
 & = \mb E^{Q_0}\left[ \prod_{n=0}^\infty e^{(1-\beta)\beta^n |\mb E_g^n f(X_t)/(q^{1/r}\|f\|_{\phi_r})|^r} \right] 
\end{align*}
by Jensen's inequality (using the fact that $\sum_{n=0}^\infty (1-\beta)\beta^n = 1$ and convexity of $x \mapsto |x|^r$). By a version of H\"older's inequality for infinite products (e.g. \cite{Karakostas}), we obtain:
\begin{align*}
 \mb E^{Q_0}\left[ \prod_{n=0}^\infty e^{(1-\beta)\beta^n |\mb E_g^n f(X_t)/(q^{1/r}\|f\|_{\phi_r})|^r} \right] 
 & \leq \prod_{n=0}^\infty \mb E^{Q_0} \left[ e^{|\mb E_g^n  f(X_t)/(q^{1/r}\|f\|_{\phi_r})|^r} \right]^{ (1-\beta)\beta^n}  \\
 & \leq \prod_{n=0}^\infty \mb E^{Q_0 \otimes Q^{\otimes n}}\left[ m_g^{\otimes n}(X_t,\ldots,X_{t+n}) e^{| f(X_{t+n+1})/(q^{1/r}\|f\|_{\phi_r})|^r} \right]^{ (1-\beta)\beta^n} \\
 & \leq \prod_{n=0}^\infty \left( 2^{1/q} \|m_g^{\otimes n}\|_p \right)^{ (1-\beta)\beta^n} \,. 
\end{align*}
Substituting $\| m_g^{\otimes n}\|_p \leq C e^{(\beta+c)^{-n}}$ yields:
\begin{align*}
 \mb E^{Q_0}\left[ e^{|(\mb I-\mb D_g)^{-1} f(X_t)/(q^{1/r}\|f\|_{\phi_r}/(1-\beta))|^r} \right] 
 & \leq  \prod_{n=0}^\infty \left( (2^{1/q} C)^{\beta^n} e^{\left(\frac{\beta}{\beta+c}\right)^n}   \right)^{ (1-\beta)} \\
 & = 2^{1/q} C \left( e^{\sum_{n=0}^\infty \left(\frac{\beta}{\beta+c}\right)^n}   \right)^{ (1-\beta)} \\
 & = 2^{1/q} C e^\frac{(1-\beta)(\beta + c)}{c}\,.
\end{align*}
The first result now follows by Lemma \ref{lem-pollard}; the proof of the second result is almost identical.
\end{proof}

\begin{lemma}\label{lem:msubexp}
Let $g \in L^{\phi_r}_2$ and $u \in L^{\phi_r}_2$ for some $r > 1$. Then: for  any $p \in (1,\infty)$ there exist $C \in (0,\infty)$ and $c \in (0,1-\beta)$ depending on $\beta$, $p$, $r$ and $\|g + \alpha u\|_{\phi_r}$ such that $\| m_g^{\otimes n}\|_p \leq C e^{(\beta+c)^{-n}}$ for each $n \geq 1$.
\end{lemma}

\begin{proof}[Proof of Lemma \ref{lem:msubexp}]
The result follows by applying Lemma \ref{lem:com} with $a_t = g(X_t,X_{t+1}) + \alpha u (X_t,X_{t+1})$ (where $\|a_t\|_{\phi_r} \leq 2 \| g + \alpha u \|_{\phi_r} < \infty$ by the triangle inequality), which shows that $\log (\|m_g^{\otimes n}\|_p)$ grows algebraically, rather than exponentially, in $n$.
\end{proof}

\begin{lemma} \label{lem:com}
Let $\{a_t : t \in T \}$ be a sequence of strictly stationary random variables for which $\|a\|_{\phi_r} := \inf\{c > 0 : \mb E[ \exp(|a_t/c|^r)] \leq 2\} < \infty$ for some $r > 1$ and let $\{\mc F_t : t \in T\}$ be a sequence of sigma-fields. Let $M_n = \prod_{t=1}^n \frac{e^{a_t}}{\mb E[e^{a_t}|\mc F_t]}$. Then for each $p \in (1,\infty)$ and $b \in (1,\infty)$:
\[
 \mb E[M_n^p]^\frac{1}{p} \leq \left(  e^{(4 p n \|a\|_{\phi_r})^\frac{r}{r-1}b^\frac{1}{r-1}} + \frac{b}{b-1} 2^\frac{1}{b} \right)^\frac{1}{p} 
\]
\end{lemma}

\begin{proof}[Proof of Lemma \ref{lem:com}]
First note that $\mb E[M_n^p] \leq \mb E[e^{2np a_t}]$ by H\"older's inequality and Jensen's inequality. Let $A$ be a positive constant (specified below) and set $a_t = a_t^+ + a_t^-$ with $a_t^+ = a_t \ind\{|a_t| \leq A\}$ and $a_t^- = a_t \ind\{|a_t| > A\}$. Then for any $t \geq 0$ and $z > 0$:
\begin{align} \label{e:tail}
  \Pr( e^{2np a_t} \geq z ) & = \Pr\left( a_t \geq \frac{\log z}{2np} \right) \leq \Pr\left( a_t^+  \geq \frac{\log z}{4pn}  \right)  + \Pr\left( a_t^-  \geq \frac{\log z}{4pn} \right) \,.
\end{align}
 For the second term in (\ref{e:tail}):
\begin{align*}
  \Pr\left( a_t^-  \geq \frac{\log z}{4pn} \right) 
  & \leq \Pr\left( |a_t|^r \ind\{a_t > A\}  \geq \frac{A^{r-1} \log z}{4pn} \right) \\
  & = \Pr\left( \exp \left( \frac{|a_t|^r \ind\{a_t > A\}}{\|a\|_{\phi_r}^r} \right)  \geq \exp \left(\frac{1}{\|a\|_{\phi_r}^r} \frac{A^{r-1} \log z}{4pn} \right) \right) \\
  & \leq  \frac{\mb E\left[ \exp \left( \left| a_t/\|a\|_{\phi_r} \right|^r \right) \right] }{e^{\frac{1}{\|a\|_{\phi_r}^r} \frac{A^{r-1} \log z}{4pn} }} \\
  & \leq  2 e^{-\frac{1}{\|a\|_{\phi_r}^r} \frac{A^{r-1} \log z}{4pn} } 
\end{align*}
by Markov's inequality. Substituting $A = ( \|a\|_{\phi_r}^r 4pnb)^\frac{1}{r-1}$ yields:
\[
 \Pr\left( a_t^-  \geq \frac{\log z}{4pn} \right) 
 \leq 2 z^{- b } \,.
\]
As $2z^{-b} \geq 1$ if $z \leq 2^\frac{1}{b}$, we have:
\begin{align} \label{e:tail:trunc}
 \int_0^\infty \Pr\left( a_t^-  \geq \frac{\log z}{2pn} \right) \, \mr d z \leq 2^\frac{1}{b} + 2 \int_{2^\frac{1}{b}}^\infty z^{-b} \, \mr d z = \frac{b}{b-1} 2^\frac{1}{b} \,.
\end{align}
For the first term on the right-hand side of (\ref{e:tail}), as $|a_t^+| \leq A$ we have $\Pr(  a_t^+  \geq \frac{\log z}{4pn} )  = 0$ if $t > e^{4pAn}$ where $4pnA = (4 p n \|a\|_{\phi_r} )^\frac{r}{r-1}b^\frac{1}{r-1}$.  It now follows from (\ref{e:tail}) and (\ref{e:tail:trunc}) and an alternate expression for the expected value of a non-negative random variable that:
\begin{align*}
 \mb E[M_n^p] \leq \mb E[e^{np a_t}] & = \int_0^\infty \Pr( e^{npa_t} \geq z) \, \mr d z \\
 & \leq  \int_0^\infty  \Pr\left( a_t^+  \geq \frac{\log z}{2pn}  \right) \, \mr d z  +  \int_0^\infty \Pr\left( a_t^- \geq \frac{\log z}{2pn} \right) \, \mr d z \\
 & \leq e^{(4 p n \|a\|_{\phi_r})^\frac{r}{r-1}b^\frac{1}{r-1}} + \frac{b}{b-1} 2^\frac{1}{b} \,.
\end{align*}
\end{proof}

\subsection{Proofs for Section \ref{s:identification}}

\begin{proof}[Proof of Lemma \ref{lem-T-prop}]
Fix any $1 \leq s \leq r$. We first show that $\mb T : E^{\phi_s} \to E^{\phi_s}$. It suffices to show that $\mb E^{Q_0}[\exp(|\mb T f(X_t)/(\beta c)|^s)] < \infty$ holds for each $f \in E^{\phi_s}$ and $c \in (0,1]$. By convexity of $x \mapsto e^{\frac{1}{c^s}|\log x|^s}$ for $c \leq 1$ and Jensen's inequality:
\begin{align*}
 \mb E^{Q_0}[\exp(|\mb T f(X_t)/(\beta c)|^s)]
 & = \mb E^{Q_0} \left[\exp \left(\left|\frac{1}{c}\log \mb E^Q \left[ \left. e^{f(X_{t+1}) + \alpha u(X_t,X_{t+1})}\right|X_t\right]\right|^s\right)\right] \\ 
 & \leq \mb E^{Q_0} \left[ \mb E^Q \left[ \left. \exp \left(\left|\frac{1}{c}\log  e^{f(X_{t+1}) + \alpha u(X_t,X_{t+1})}\right|^s\right) \right|X_t\right] \right] \\ 
 & = \mb E^{Q_0 \otimes Q} \left[  \exp \left(\left|\frac{f(X_{t+1}) + \alpha u(X_t,X_{t+1})}{c}\right|^s\right) \right] < \infty
\end{align*}
which is finite because $f \in E^{\phi_s}$ and $u \in E^{\phi_r}_2 \subseteq E^{\phi_s}_2$. 

Continuity: Fix any $f\in E^{\phi_s}$. Take $g \in E^{\phi_s}$ with $\|g\|_{\phi_s}  \in (0,2^{-1/s}]$ and set $c = 2^{1/s} \|g\|_{\phi_s}$. Then by convexity of $x \mapsto e^{\frac{1}{c}|\log x|^s}$ for $c \leq 1$, Jensen's inequality, and the Cauchy-Schwarz inequality:
\begin{align*}
 \mb E^{Q_0} \left[ \phi_s(| \mb T(f + g)(X_t) - \mb T f(X_t)|/(\beta c)) \right] +1
 & = \mb E^{Q_0} \left[\exp \left(\left|\frac{1}{c}\log \mb E_f \left[ \left. e^{g(X_{t+1}) }\right|X_t\right]\right|^s\right)\right] \\
 & \leq \mb E^{Q_0} \left[\mb E_f \left[ \left. \exp \left(\left|\frac{1}{c}\log  e^{g(X_{t+1}) } \right|^s\right) \right|X_t\right] \right] \\
 & = \mb E^{Q_0 \otimes Q} \left[m_f(X_t,X_{t+1})  \exp \left(\left|\frac{g(X_{t+1})}{c} \right|^s\right)  \right] \\
 & \leq \mb E^{Q_0} \big[ e^{2|g(X_t)/c|^s} \big]^{1/2} \| m_f\|_2  \\
 & =  \sqrt 2 \| m_f\|_2 
\end{align*}
where $\| m_f\|_2 < \infty$ by Lemma \ref{lem-lp-m} and the final line is because $c = 2^{1/s} \|g\|_{\phi_s}$. Therefore:
\[
 \|  \mb T(f + g) - \mb T f\|_{\phi_s} \leq 2 \beta ((\sqrt 2 \| m_f\|_2 - 1) \vee 1) \times \|g\|_{\phi_s} \to 0 \mbox{ as $\|g\|_{\phi_s} \to 0$}
\]
by Lemma \ref{lem-pollard}, proving continuity.

Monotonicity follows from monotonicity of the exponential and logarithm functions and monotonicity of conditional expectations. Convexity is immediate by applying H\"older's inequality to the conditional expectation 
\[
 \mb E^Q \left[ \left. e^{\tau(v_1(X_{t+1}) + \alpha u(X_t,X_{t+1})) + (1-\tau)(v_2(X_{t+1}) + \alpha u(X_t,X_{t+1}))}  \right|X_t=x \right]
\]
with $\frac{1}{p} = \tau$ and $\frac{1}{q} = 1-\tau$.
\end{proof}

\begin{proof}[Proof of Lemma \ref{lem-D-bdd}]
It is clear that $\mb E_v$ is a linear operator. To show that $\mb E_v :L^{\phi_s} \to L^{\phi_s}$ is continuous for all $s \geq 1$, take any $f \in L^{\phi_s}$ with $\|f\|_{\phi_s} = 1$ and fix $c \in (1,\infty)$. By Jensen's inequality:
\begin{align*}
 \mb E^{Q_0} \Big[ \phi_s(| \mb E_v f(X_t)|/c) \Big] 
 & =  \mb E^{Q_0} \Big[  \phi_s(| \mb E_v f(X_t)/c|)   \Big] \\
 & \leq \mb E^{Q_0} \Big[  \mb E_v [ \phi_s \big(|f(X_{t+1})/c| \big) \big|X_t \big] \Big] \\
 & = \mb E^{Q_0 \otimes Q} \Big[  m_v(X_t,X_{t+1}) e^{|f(X_{t+1})/c|^s}  \Big] -1\,.
\end{align*}
By H\"older's inequality with $p^{-1} + q^{-1} = 1$ for $q = c^s$:
\begin{align*}
  \mb E^{Q_0 \otimes Q} \Big[  m_v(X_t,X_{t+1}) e^{|f(X_{t+1})/c|^s}  \Big] & \leq \mb E^{Q_0} \Big[ e^{|f(X_{t+1})|^s}  \Big]^{1/q} \|m_v\|_p \\
  & = 2^{1/q}\|m_v\|_p
\end{align*}
where $\| m_v\|_p < \infty$ by Lemma \ref{lem-lp-m} and the final line is because $\|f\|_{\phi_s} = 1$. Now by Lemma \ref{lem-pollard}:
\[
 \sup_{f \in L^{\phi_s} : \|f\|_{\phi_s} = 1} \|\mb E_v f\|_{\phi_s} \leq q^\frac{1}{s} ( ( 2^\frac{1}{q} \|m_v\|_p - 1 ) \vee 1 )  < \infty 
\]
as required. Continuity of $\mb D_v: L^{\phi_s} \to L^{\phi_s}$ follows by identical arguments. 

To show $\mb E_v : E^{\phi_s} \to E^{\phi_s}$, take any $f \in E^{\phi_s}$ and any $c > 0$. By the same arguments used above:
\begin{align*}
 \mb E^{Q_0} \Big[ \phi_s(| \mb E_v f(X_t)|/c) \Big] 
 & \leq \mb E^{Q_0 \otimes Q} \Big[  m_v(X_t,X_{t+1}) e^{|f(X_{t+1})/c|^s}  \Big] -1 \\
 & \leq \|m_v\|_p \mb E^{Q_0 } \Big[ e^{q|f(X_{t+1})/c|^s}  \Big]^{1/q} -1
\end{align*}
for any $p \in (1,\infty)$, where $q$ is the dual index of $p$. The right-hand side is finite because $\| m_v\|_p < \infty$  (by Lemma \ref{lem-lp-m}) and because $f \in E^{\phi_s}$. Therefore,  $\mb E_vf \in E^{\phi_s}$ which proves $\mb E_v : E^{\phi_s} \to E^{\phi_s}$. The proof that $\mb D_v:  E^{\phi_s} \to E^{\phi_s}$ is identical. As $E^{\phi_s}$ is a closed linear subspace of $L^{\phi_s}$, for each continuous linear operator $\mb K : L^{\phi_s} \to L^{\phi_s}$ we have: 
\begin{align}\label{e:normcomp}
 \|\mb K\|_{E^{\phi_s}} \leq \|\mb K\|_{L^{\phi_s}} \,.
\end{align}
Therefore, $\mb D_v: E^{\phi_s} \to E^{\phi_s}$ and $\mb E_v: E^{\phi_s} \to E^{\phi_s}$ are continuous.

Lemmas \ref{lem:sr-mp} and \ref{lem:msubexp} and the conditions $v \in E^{r'}$ and $u \in E^r_2$ imply that $\rho(\mb D_v;L^{\phi_s}) < 1$ for all $s \geq 1$. It follows from inequality (\ref{e:normcomp}) and formula $\rho(\mb D_v;E^{\phi_s}) = \lim_{n \to \infty} \|( \mb D_v)^n \|_{E^{\phi_s}}^{1/n}$ that:	
\[
 \rho(\mb D_v;E^{\phi_s}) \leq \rho(\mb D_v;L^{\phi_s}) \,,
\]
hence $\rho(\mb D_v;E^{\phi_s}) < 1$.
\end{proof}

\begin{proof}[Proof of Theorem \ref{t-id-W}]
We verify the conditions of Proposition \ref{p:exun}. For existence, continuity and monotonicity have been established in Lemma \ref{lem-T-prop}. Define
\[
 \ol v(x) = (1-\beta) \sum_{n=0}^\infty \beta^{n+1} \log \mb E^{Q} \Big[ e^{\frac{\alpha}{1-\beta} u(X_{t+n},X_{t+n+1})} \Big|X_t = x \Big] \,.
\]
We first show that $\ol v \in E^{\phi_r}$. It suffices to show that $\mb E^{Q_0}[\exp(|\ol v(X_t)/(\beta c)|^r)] < \infty$ holds for each $c \in (0,1]$. By Jensen's inequality (using the fact that $\sum_{n=1}^\infty (1-\beta) \beta^n = 1$ and convexity of $x \mapsto e^{|x|^r}$ and $x \mapsto x^{1/c}$ for $c \in (0,1]$):
\begin{align*}
 \mb E^{Q_0}\left[ e^{\left|\ol v(X_t)/(\beta c)\right|^r} \right] 
 & = \mb E^{Q_0}\left[ \exp \left( \left| (1-\beta) \sum_{n=0}^\infty \frac{\beta^{n}}{c} \log \mb E^{Q} \Big[ e^{\frac{\alpha}{1-\beta} u(X_{t+n},X_{t+n+1})} \Big|X_t  \Big] \right|^r \right) \right] \\
 & \leq \mb E^{Q_0}\left[  (1-\beta) \sum_{n=0}^\infty \beta^{n}  \exp \left( \left| \frac{1}{c} \log \mb E^{Q} \Big[ e^{\frac{\alpha}{1-\beta} u(X_{t+n},X_{t+n+1})} \Big|X_t  \Big] \right|^r \right) \right] \\
 & = (1-\beta) \sum_{n=0}^\infty \beta^{n} \mb E^{Q_0}\left[    \exp \left( \left| \log \left(\mb E^{Q} \Big[ e^{\frac{\alpha}{1-\beta} u(X_{t+n},X_{t+n+1})} \Big|X_t  \Big]^{1/c} \right) \right|^r \right) \right] \\
 & \leq (1-\beta) \sum_{n=0}^\infty \beta^{n} \mb E^{Q_0}\left[    \exp \left( \left| \log \mb E^{Q} \Big[ e^{\frac{\alpha}{(1-\beta)c} u(X_{t+n},X_{t+n+1})} \Big|X_t  \Big] \right|^r \right) \right] \,.
\end{align*}
By another application of Jensen's inequality (using convexity of $x \mapsto \exp(|\log x|^r)$) and iterated expectations:
\begin{align*}
 \mb E^{Q_0}\left[ e^{\left|\ol v(X_t)/\beta\right|^r} \right] 
 & \leq (1-\beta) \sum_{n=0}^\infty \beta^n \mb E^{Q_0 \otimes Q^{\otimes n}}\Big[ e^{|\frac{\alpha}{(1-\beta)c} u(X_{t+n},X_{t+n+1})|^r}\Big] =  \mb E^{Q_0 \otimes Q} \Big[ e^{|\frac{\alpha}{(1-\beta)c} u(X_t,X_{t+1})|^r}\Big] 
\end{align*}
which is finite because $u \in E^{\phi_r}_2$. In particular, it follows by taking $c = 1$ and applying Lemma \ref{lem-pollard} that  
\[
 \|\ol v \|_{\phi_r} \leq \beta \Big(\Big(\mb E^Q\Big[ e^{|\frac{\alpha}{1-\beta} u(X_t,X_{t+1})|^r}\Big] - 1\Big) \vee 1\Big) < \infty \,.
\]
We now show that $\mb T \ol v \leq \ol v$. By Holder's inequality:
\begin{align}
 \mb T \ol v(X_t) & \leq \beta \log \left( \mb E^Q \Big[ e^{\ol v(X_{t+1})/\beta} \Big|X_t \Big]^{\beta} \mb E^Q \Big[ e^{\frac{\alpha}{1-\beta} u(X_t,X_{t+1})} \Big|X_t \Big]^{1-\beta} \right) \notag \\
 & = \beta^2  \log  \mb E^Q [ e^{\ol v(X_{t+1})/\beta}|X_t ] + (1-\beta) \beta \log \mb E^Q\Big[ e^{\frac{\alpha}{1-\beta} u(X_t,X_{t+1})} \Big|X_t \Big] \,. \label{e-exist-1}
\end{align} 
By a version of H\"older's inequality for infinite products (see, e.g., \cite{Karakostas}), we obtain:
\begin{align}
 \log  \mb E^Q \left[ \left. e^{\ol v(X_{t+1})/\beta} \right|X_t \right] 
 & = \log  \mb E^Q \left[ \left. \exp \left( (1-\beta) \sum_{n=0}^\infty \beta^{n} \log \mb E^Q \left[ \left. e^{\frac{\alpha}{1-\beta} u(X_{t+n+1},X_{t+n+2})} \right| X_{t+1} \right] \right) \right| X_t\right] \notag \\
 & = \log  \mb E^Q \left[ \left.  \prod_{n=0}^\infty    \mb E^Q \left[ \left. e^{\frac{\alpha}{1-\beta} u(X_{t+n+1},X_{t+n+2})} \right| X_{t+1} \right]^{(1-\beta)  \beta^{n}}  \right| X_t\right] \notag \\
 & \leq \log  \left( \prod_{n=0}^\infty  \mb E^Q \left[ \left.     \mb E^Q \left[ \left. e^{\frac{\alpha}{1-\beta} u(X_{t+n+1},X_{t+n+2})} \right| X_{t+1} \right]  \right| X_t\right]^{(1-\beta)  \beta^{n}} \right) \notag \\
 & = (1-\beta) \sum_{n=1}^\infty \beta^{n-1}  \log  \mb E^Q \Big[ e^{\frac{\alpha}{1-\beta} u(X_{t+n},X_{t+n+1})} \Big| X_t \Big] \,. \label{e-exist-2}
\end{align}
Substituting (\ref{e-exist-2}) into (\ref{e-exist-1}) yields $\mb T \ol v \leq \ol v$. 

We now show $\{\mb T^n \ol v\}_{n \geq 1}$ is bounded from below by some $\ul v \in E^{\phi_r}$. By Jensen's inequality, for any $f \in E^{\phi_r}$:
\[
 \mb T f (x) = \beta \log \mb E^Q[ e^{f(X_{t+1}) + \alpha u(X_t,X_{t+1})} | X_t = x] \geq \beta \mb E^Q[ f(X_{t+1}) + \alpha u(X_t,X_{t+1}) | X_t = x] \,.
\]
In particular, with $f = \ol v$ we have $\mb T \ol v \geq \beta \mb E^Q( \ol v + \alpha u)$. Iterating and using monotonicity of $\mb T$:
\begin{align*}
 \mb T^2 \ol v & \geq \mb T(  \beta \mb E^Q( \ol v + \alpha u)) \\
 & \geq \beta \mb E^Q ( \beta \mb E^Q( \ol v + \alpha u) + \alpha u ) \\ 
 & = (\beta \mb E^Q)^2 \ol v + (\mb I + \beta \mb E^Q) ( \beta \mb E^Q  (\alpha u))  \,.
\end{align*}
Thus, by induction:
\[
 \mb T^n \ol v \geq (\beta \mb E^Q)^n \ol v + \sum_{s=0}^{n-1}  (\beta \mb E^Q)^s ( \beta \mb E^Q(\alpha u))
\]
for each $n \geq 2$. By Jensen's inequality we may deduce $\rho(\beta \mb E^Q;E^{\phi_r}) = \beta < 1$. Therefore, $\lim_{n \to \infty} (\beta \mb E^Q)^n \ol v = 0$ and $\lim_{n \to \infty} \sum_{s=0}^n (\beta \mb E^Q)^s=(\mb I - \beta \mb E^Q)^{-1}$ (where the limit is in terms of norm topology on the algebra of bounded linear operators on $E^{\phi_r}$). Finally, noting that $u \in E^{\phi_r}_2$ implies $ \beta \mb E^Q(\alpha u) \in  E^{\phi_r}$ we see that $\lim_{n \to \infty} \sum_{s=0}^{n-1}  (\beta \mb E^Q)^s ( \beta \mb E^Q(\alpha u)) = (\mb I - \beta \mb E^Q)^{-1}( \beta \mb E^Q(\alpha u))$ is a well-defined element of $E^{\phi_r}$. We have therefore shown that:
\[
 \liminf_{n \to \infty} \mb T^n \ol v \geq (\mb I - \beta \mb E^Q)^{-1}( \beta \mb E^Q(\alpha u)) \in E^{\phi_r} \,.
\]
Existence of a fixed point $v \in E^{\phi_r}$ now follows by applying Proposition \ref{p:exun}(i).

For uniqueness, $v$ is necessarily a fixed point of $\mb T : E^{\phi_s} \to E^{\phi_s}$ for each $1 \leq s \leq r$. Moreover, $\mb T : E^{\phi_s} \to E^{\phi_s}$ is convex by Lemma \ref{lem-T-prop} and the subgradient $\mb D_v$ at any $v \in E^{\phi_s}$ is a bounded, monotone linear operator with $\rho(\mb D_v;E^{\phi_s}) < 1$ by Lemma \ref{lem-D-bdd}. Uniqueness now follows by Proposition \ref{p:exun}(ii).
\end{proof}

\begin{proof}[Proof of Proposition \ref{prop-nonunique}]
By the subgradient inequality (cf. equation (\ref{e:subgrad})), for $v,v' \in \mc V$:
\[
 v'-v = \mb T v' - \mb Tv \geq \mb D_v(v'-v)
\]
hence
\[
 (\mb I - \mb D_v)(v'-v) \geq 0\,.
\]
Lemma \ref{lem-D-bdd} shows that $\rho(\mb D_v;E^{\phi_1}) < 1$ hence $(\mb I - \mb D_v) : E^{\phi_1} \to E^{\phi_1}$ is invertible and $(\mb I - \mb D_v) ^{-1} = \sum_{n=0}^\infty \mb D_v^n$. As $\mb D_v$ is monotone so too is $(\mb I - \mb D_v)^{-1}$. Applying $(\mb I - \mb D_v)^{-1}$ to both sides of the above display yields $v' - v \geq 0$. If any $v' \in \mc V$ distinct from $v$ were stable, we could apply an identical argument to obtain the reverse inequality $v - v' \geq 0$, a contradiction.
\end{proof}

\begin{proof}[Proof of Theorem \ref{t-id-R}]
The resolvent set of $\mb D_v: E^{\phi_s} \to E^{\phi_s}$ is defined, for a suitable complexification of the space, as $\{z \in \mb C : (z \mb I - \mb D_v)^{-1}$ is a bounded linear operator on $E^{\phi_s}\}$. The spectrum $\sigma( \mb D_v; E^{\phi_s})$ of $\mb D_v: E^{\phi_s} \to E^{\phi_s}$ is the complement of its resolvent set. The spectral radius $\rho(\mb D_v; E^{\phi_s})$ may be equivalently formulated as $\rho(\mb D_v; E^{\phi_s}) = \sup\{ |z| : z \in \sigma( \mb D_v;E^{\phi_s})\}$. Thus, $\rho(\mb D_v;E^{\phi_s}) < 1$ implies $(\mb I - \mb D_v)^{-1}$ is a bounded linear operator on $E^{\phi_s}$.

Lemma \ref{lem-D-bdd} shows that $\rho(\mb D_v;E^{\phi_s}) < 1$ for each $1 \leq s \leq r$, hence $(\mb I - \mb D_v)^{-1}$ is a bounded linear operator on $E^{\phi_s}$ for each $1 \leq s \leq r$. To complete the proof, it remains to show that $\chi_v \in E^{\phi_r}$. To do so, first note that $\log  m_v \in E^{\phi_r}_2$ because $v \in E^{\phi_r}$ and $u \in E^{\phi_r}_2$. Now for any $c > 0$:
\begin{align*}
 \mb E^{Q_0} \Big[ e^{|\chi_v(X_t)/(\beta c)|^r} \Big] 
 & = \mb E^{Q_0} \Big[ e^{|\mb E_v[\log m_v(X_t,X_{t+1})/c|X_t]|^r} \Big] \\
 & \leq \mb E^{Q_0} \Big[ \mb E_v\big[ e^{| \log m_v(X_t,X_{t+1})/c|^r} \big| X_t \big]  \Big] \\
 & = \mb E^{Q_0 \otimes Q} \Big[ m_v(X_t,X_{t+1}) e^{| \log m_v(X_t,X_{t+1})/c|^r}  \Big] \\
 & \leq \|m_v\|_2 \mb E^{Q_0 \otimes Q} \left[  e^{2| \log m_v(X_t,X_{t+1})/c|^r}  \right]^{\frac{1}{2}} < \infty
\end{align*}
by Lemma \ref{lem-lp-m} and the fact that $\log m_v \in E^{\phi_r}_2$. Therefore, $\chi_v \in E^{\phi_r}$.
\end{proof}

\begin{proof}[Proof of Proposition \ref{prop-local-1}]
First step: Let $\partial_i v_0 = \partial_i v_{(\alpha_0,\beta_0)}$ for $i \in \{\alpha,\beta\}$ denote the Fr\'echet derivatives of $v$ with respect to $\alpha$ and $\beta$ at $(\alpha_0,\beta_0)$. To characterize the derivatives we apply Theorems 20.1 and 20.3 of \cite{Kras}. Define the operator $\mb F$ on $\mc A \times E^{\phi_r}$ by $\mb F((\alpha,\beta),v) = v - \mb T v$. Then $v_{(\alpha,\beta)}$ solves $\mb F((\alpha,\beta),v)  = 0$. The Fr\'echet derivative of $\mb F$ with respect to $v$ at $v_0$ is $(\mb I - \mb D_v)$ which is continuously invertible on $E^{\phi_r}$ by Lemma \ref{lem-D-bdd}. Theorem 20.3 of \cite{Kras} ensures the derivatives are well defined and take the form: 
\begin{align*}
 \partial_\alpha v_0(x) & = (\mb I - \mb D_{v_0})^{-1}\beta_0 \mb E_{v_0}[u(X_t,X_{t+1})|X_t=x] \,, & 
 \partial_\beta v_0(x) & = \frac{1}{\beta_0} (\mb I - \mb D_{v_0})^{-1} v_0(x) 
\end{align*}
which may be written as
\begin{align}
 \partial_\alpha v_0(x) & = \beta_0 \mb E_{v_0}[\partial_\alpha v_0(X_{t+1}) + u(X_t,X_{t+1})|X_t=x] \label{e-W-deriv-alpha}   \\
 \partial_\beta v_0(x) & = \beta_0 \mb E_{v_0}[\partial_\beta v_0(X_{t+1}) |X_t=x] + \frac{1}{\beta_0} v_0(x) \,. \label{e-W-deriv-beta} 
\end{align}

Second step: The higher-than-second moment condition on $\bs g$ and the conditions $u \in E^{\phi_r}_2$ and $v \in E^{\phi_r}$ ensure that the functional derivatives
\begin{align*}
 \partial_\alpha \rho(\alpha_0,\beta_0;X_t) 
 & = \mb E^Q\big[ m_{t+1}^* \beta_0 \bs g_{t+1}( u(X_t,X_{t+1}) + \partial_\alpha v_0(X_{t+1}) )  \big| X_t \big]- \frac{1}{\beta_0} \partial_\alpha v_0(X_t) \bs 1  \\
 \partial_\beta \rho(\alpha_0,\beta_0;X_t) 
 & = \mb E^Q\big[ m_{t+1}^*  \beta_0 \bs g_{t+1} \partial_\beta v_0(X_{t+1}) \big| X_t \big] + \frac{1}{\beta_0} \Big(  \frac{1}{\beta_0} v_0(X_t) - \partial_\beta v_0(X_t) + 1 \Big)\bs 1 
\end{align*}
are well defined. Substituting (\ref{e-W-deriv-alpha}) and (\ref{e-W-deriv-beta}) into these expressions yields:
\begin{align*}
 \partial_\alpha \rho(\alpha_0,\beta_0;X_t) 
 & = \mb E^Q\big[ m_{t+1}^*  (\beta_0 \bs g_{t+1} - \bs 1) (\partial_\alpha v_0(X_{t+1}) +u(X_t,X_{t+1})) \big| X_t \big]  \\
 \partial_\beta \rho(\alpha_0,\beta_0;X_t) 
 & = \mb E^Q\big[ m_{t+1}^*  (\beta_0 \bs g_{t+1} - \bs 1) \partial_\beta v_0(X_{t+1}) \big| X_t \big] + \frac{1}{\beta_0} \bs 1  \,.
\end{align*}
The expressions in displays (\ref{e-rho-deriv-alpha}) and (\ref{e-rho-deriv-beta}) now follow.

Third step: Let $\delta = (\alpha,\beta)'$ and $\delta_0 = (\alpha_0,\beta_0)'$. By step 2 and the fact that $\rho(\alpha_0,\beta_0;X_t) = \mf 0$: 
\[
 \| \rho(\alpha,\beta;X_t) -  \partial_{\delta'} \rho(\alpha_0,\beta_0;X_t) (\delta - \delta_0) \|_2^2 = o( \|\delta - \delta_0\|^2 ) 
\] 
where $\|\cdot\|_2$ denotes that the $L^2(Q_0 \times Q)$ norm is applied element-wise. Positive definiteness of $\mf V$ ensures that there is $\epsilon > 0$ such that whenever $\|\delta-\delta_0\| < \epsilon$, the remainder term is smaller than $(\delta-\delta_0)'\mf V  (\delta-\delta_0)$. Therefore, whenever $\|\delta-\delta_0\| < \epsilon$, we have
\begin{align*}
  \| \rho(\alpha,\beta;X_t) -  \partial_{\delta'} \rho(\alpha_0,\beta_0;X_t) (\delta - \delta_0) \|_2^2 < \| \partial_{\delta'} \rho(\alpha_0,\beta_0;X_t) (\delta - \delta_0) \|_2^2
\end{align*}
and hence $\rho(\alpha,\beta;X_0) \neq 0$.
\end{proof}

\begin{proof}[Proof of Proposition \ref{prop-Q-nonid}]
Let $v_0 \in E^{\phi_r}$ denote the unique continuation value function corresponding to $(Q,U,\beta_0,\theta_0)$. Let $\alpha(\theta) = \frac{-1}{\theta(1-\beta_0)}$ and $\alpha_0 = \alpha(\theta_0)$. For each $\theta > 0$, define $Q_\theta$ by:
\[
 \frac{\mr d Q_\theta(x'|x)}{\mr d Q(x'|x)} = e^{v_{(\alpha_0 - \alpha(\theta),\beta_0)}(x') + (\alpha_0 - \alpha(\theta)) u(x,x') - \beta^{-1} v_{(\alpha_0 - \alpha(\theta),\beta_0)}(x)} 
\]
where $v_{(\alpha_0 - \alpha(\theta),\beta_0)} \in E^{\phi_r}$ denotes the unique solution to the continuation value recursion with parameters $(\alpha(\theta),\beta_0)$ under $Q$ (existence and uniqueness of $v_{(\alpha_0 - \alpha(\theta),\beta_0)} $ is ensured by Theorem \ref{t-id-W}). The function $v_0 - v_{(\alpha_0 - \alpha(\theta),\beta_0)}$ solves the fixed-point recursion (\ref{e-recur-v}) under $Q_\theta$ with composite parameter $\alpha(\theta)$, because:
\begin{align*}
 & \beta \log \mb E^{Q_\theta}\left[ \left. e^{v_0(X_{t+1}) - v_{(\alpha_0 - \alpha(\theta),\beta_0)}(X_{t+1}) + \alpha(\theta) u(X_t,X_{t+1})} \right| X_t = x \right] \\
 & = \beta \log \mb E^{Q}\left[ \left. e^{ v_0(X_{t+1}) + \alpha_0 u(X_t,X_{t+1}) - \beta^{-1} v_{(\alpha_0 - \alpha(\theta),\beta_0)}(X_t)} \right| X_t = x \right] \\
 & = v_0(x) -  v_{(\alpha_0 - \alpha(\theta),\beta_0)}(x) \,.
\end{align*}
Moreover, $v_0 - v_{(\alpha_0 - \alpha(\theta),\beta_0)}$ is the unique solution in $E^{\phi_r}$ because $v_0$ is the unique solution in $E^{\phi_r}$ to the recursion under $(Q,U,\beta_0,\theta_0)$. The worst-case belief distortion corresponding to $Q_\theta$ is (cf. equation (\ref{e-m-v})):
\[
 m_{t+1,\theta}^* = e^{v_0(X_{t+1}) - v_{(\alpha_0 - \alpha(\theta),\beta_0)}(X_{t+1}) + \alpha(\theta) u(X_t,X_{t+1}) - \beta^{-1} (v_0(X_t) - v_{(\alpha_0 - \alpha(\theta),\beta_0)}(X_t))}  \,.
\]
The resulting worst-case model relative to $Q_\theta$ has a change of measure relative to $Q$ given by:
\[
 m_{t+1,\theta}^* \times \frac{\mr d Q_\theta(X_{t+1}|X_t)}{\mr d Q(X_{t+1}|X_t)} = e^{v_0(X_{t+1}) + \alpha_0 u(X_t,X_{t+1}) - \beta_0^{-1} v_0(X_t)}
\]
which is precisely the worst-case model corresponding to $(Q,U,\beta_0,\theta_0)$. It follows that $(Q,U,\beta_0,\theta_0)$ and $(Q_\theta,U,\beta_0,\theta)$ are observationally equivalent. 
\end{proof}

\subsection{Proofs for Section \ref{s:learn}}

Throughout this subsection only, $\|\cdot\|_{\phi_s}$ denotes the Luxemburg norm on $L^{\phi_s}_{\tilde X} := \{f : \mc X_{\tilde X} \to \mb R : \mb E^{\tilde Q_0}[ \exp(|f(\tilde X_t)/c|^s)] < \infty$ for some $c > 0$. 

Under conditions S-Learn (i)--(ii), the conditional distribution $\tilde Q$ of $(\xi_t,\tilde X_{t+1})$ given $\tilde X_t$ is represented by the conditional expectation operator:
\[
 \mb E^{\tilde Q} [h(\xi_t,\tilde X_{t+1})|\tilde X_t] = \mb E^{\tilde Q} [h(\xi_t,\tilde X_{t+1})|\tilde \xi_t] = \mb E^{\Pi_\xi \otimes Q_\varphi}[ h(\xi_t,\varphi_{t+1},\Xi(\tilde \xi_t,\varphi_{t+1})) | \tilde \xi_t] \,.
\]
The operator $\tilde{\mb T}$ satisfies a subgradient inequality analogous to equation (\ref{e:subgrad}). For $v \in E^{\phi_1}_{\tilde \xi}$, define
\begin{align*}
 m_v^{\Pi_\xi}(\xi_t,\tilde \xi_t) & = \frac{ \mb E^{Q_\varphi} \left[ \left.  e^{\frac{\theta}{\vartheta} v(\Xi(\tilde \xi_t,\varphi_{t+1})) + \alpha u(\varphi_{t+1})} \right| \xi_t ,\tilde \xi_t \right]^\frac{\vartheta}{\theta} }{\mb E^{\Pi_\xi}\! \left[ \left. \mb E^{Q_\varphi} \left[ \left.  e^{\frac{\theta}{\vartheta} v(\Xi(\tilde \xi_t,\varphi_{t+1})) + \alpha u(\varphi_{t+1})} \right| \xi_t ,\tilde \xi_t \right]^\frac{\vartheta}{\theta}  \right| \tilde \xi_t \right]} \\
 m_v^{Q_\varphi}(\xi_t,\tilde \xi_t,\varphi_{t+1}) & = \frac{e^{\frac{\theta}{\vartheta} v(\Xi(\tilde \xi_t,\varphi_{t+1})) + \alpha u(\varphi_{t+1})}}{\mb E^{Q_\varphi} \left[ \left.  e^{\frac{\theta}{\vartheta} v(\Xi(\tilde \xi_t,\varphi_{t+1})) + \alpha u(\varphi_{t+1})} \right| \xi_t, \tilde \xi_t \right]} \,.
\end{align*}
The quantity $m_v^{\Pi_\xi}$ distorts the posterior distribution for $\xi_t$ given $\tilde X_t$ whereas $m_v^{ Q_\varphi}$ distorts the conditional distribution $Q_\varphi$. Let $\mb E_v^{\Pi_\xi}$ and $\mb E_v^{Q_\varphi}$ denote conditional expectations under the distorted measures. The subgradient of $\tilde{\mb T}$ at $v$, denoted $\tilde{\mb D}_v$, is given by:
\[
 \tilde{\mb D}_v f(\tilde \xi) = \beta \mb E^{\Pi_\xi}_v \left[ \left. \mb E^{ Q_\varphi}_v \left[ \left. f(\Xi(\tilde \xi_t,\varphi_{t+1})) \right| \xi_t, \tilde \xi_t \right] \right| \tilde \xi_t=\tilde \xi \right]  = \beta \mb E^{\tilde Q} \left[ \left. m_v (\xi_t,\tilde \xi_t,\varphi_{t+1}) f(\tilde \xi_{t+1}) \right| \tilde \xi_t=\tilde \xi \right] 
\]
where $m_v (\xi_t,\tilde \xi_t,\varphi_{t+1}) = m_v^{\Pi_\xi}(\xi_t,\tilde \xi_t) m_v^{Q_\varphi}(\xi_t,\tilde \xi_t,\varphi_{t+1})$.

\begin{lemma} \label{lem-T-prop-learn}
Let Assumption U-Learn hold. Then: $\tilde{\mb T}$ is a continuous, monotone and convex operator on $ E^{\phi_s}_{\tilde \xi}$ for each $1 \leq s \leq r$.
\end{lemma}

\begin{proof}[Proof of Lemma \ref{lem-T-prop-learn}]
The proof is similar to the proof of Lemma \ref{lem-T-prop}. Fix any $1 \leq s \leq r$. The show $\tilde{\mb T} : E^{\phi_s}_{\tilde \xi} \to E^{\phi_s}_{\tilde \xi}$, it suffices to show that $\mb E^{\tilde Q_0}[\exp(|\tilde{\mb T} f(\tilde \xi_t)/(\beta c)|^s)] < \infty$ holds for each $f \in E^{\phi_s}_{\tilde \xi}$ and $c \in (0,\frac{\vartheta}{\theta} \wedge 1]$. By convexity of $x \mapsto e^{\frac{1}{c^s}|\log x|^s}$ for $c \leq 1$ and Jensen's inequality:
\begin{align*}
 \mb E^{\tilde Q_0}\left[\exp\left(\left|\frac{\tilde{\mb T} f(\tilde \xi_t)}{\beta c}\right|^s\right)\right] 
 & = \mb E^{\tilde Q_0} \left[ \exp \left( \frac{1}{c^s} \left| \log \mb E^{\Pi_\xi}\! \left[ \left. \mb E^{Q_\varphi} \left[ \left. e^{\frac{\theta}{\vartheta} f(\Xi(\tilde \xi_t,\varphi_{t+1})) + \alpha u (\varphi_{t+1})} \right| \xi_t, \tilde \xi_t \right]^\frac{\vartheta}{\theta}  \right| \tilde \xi_t \right]  \right|^s \right) \right] \\
 & \leq \mb E^{\tilde Q_0} \left[ \mb E^{\Pi_\xi}\! \left[ \left. \exp \left( \frac{1}{c^s}  \left| \log  \mb E^{Q_\varphi} \left[ \left. e^{\frac{\theta}{\vartheta} f(\Xi(\tilde \xi_t,\varphi_{t+1})) + \alpha u (\varphi_{t+1})} \right| \xi_t, \tilde \xi_t \right]^\frac{\vartheta}{\theta}  \right|^s \right) \right| \tilde \xi_t \right]\right] \\ 
 & \leq \mb E^{\tilde Q_0} \left[ \mb E^{\Pi_\xi}\! \left[ \left. \mb E^{Q_\varphi} \left[ \left.  \exp \left( \frac{1}{c^s}  \left| \frac{\vartheta}{\theta} \log  e^{\frac{\theta}{\vartheta} f(\Xi(\tilde \xi_t,\varphi_{t+1})) + \alpha u (\varphi_{t+1})}  \right|^s \right)  \right| \xi_t, \tilde \xi_t \right]  \right| \tilde \xi_t \right]\right] \\ 
 & = \mb E^{\tilde Q_0 \otimes \Pi_\xi \otimes Q_\varphi} \left[ \exp \left( \frac{1}{c^s}  \left| f(\Xi(\tilde \xi_t,\varphi_{t+1}))  +\frac{\vartheta}{c\theta}  \alpha u (\varphi_{t+1})  \right|^s \right)  \right]  
\end{align*}
which is finite because $f \in {E}^{\phi_s}_{\tilde \xi}$ and $u \in E^{\phi_r}_\varphi$. Therefore, $\tilde {\mb T} :E^{\phi_s}_{\tilde \xi} \to E^{\phi_s}_{\tilde \xi}$. 

For continuity, fix $f\in E^{\phi_s}_{\tilde \xi}$. Take $0 \neq g \in E^{\phi_s}_{\tilde \xi}$ with $\|g\|_{\phi_s} \leq 2^{-1/s}(1 \wedge \frac{\vartheta}{\theta})$ and set $c = 2^{1/s} \|g\|_{\phi_s}$. Note that:
\[
 \tilde{\mb T}(f+g)(\tilde \xi) - \tilde{\mb T} f(\tilde \xi) = \beta \log \left( \mb E^{\Pi_{\xi}}_f \left[ \left. \mb E^{Q_\varphi}_f \left[ \left. e^{\frac{\theta}{\vartheta} g(\Xi(\tilde \xi_t,\varphi_{t+1}))} \right| \xi_t, \tilde \xi_t \right]^\frac{\vartheta}{\theta} \right| \tilde \xi_t=\tilde \xi \right] \right)\,. 
\]
By similar arguments to the above, we may deduce:
\begin{align*}
 \mb E^{\tilde Q_0} \left[ \exp \left( \left| \frac{\tilde{\mb T}(f + g)(\tilde \xi_t) - \tilde{\mb T} f(\tilde \xi_t)}{\beta c} \right|^s \right) \right] 
 & \leq \mb E^{\tilde Q_0} \left[ \mb E^{\Pi_\xi}_f\! \left[ \left. \mb E^{Q_\varphi}_f \left[ \left.  \exp \left(  \left| \frac{1}{c} g(\Xi(\tilde \xi_t,\varphi_{t+1}))  \right|^s \right)  \right| \xi_t, \tilde \xi_t \right]  \right| \tilde \xi_t \right]\right] \\ 
 & = \mb E^{\tilde Q_0} \left[ \mb E^{\tilde Q} \left[ \left. m_f(\xi_t,\tilde \xi_t,\varphi_{t+1}) \exp \left(  \left| \frac{1}{c} g(\Xi(\tilde \xi_t,\varphi_{t+1}))  \right|^s \right)   \right| \tilde \xi_t \right]\right] \\
 & \leq \mb E^{\tilde Q_0 \otimes \tilde Q}\! \left[ m_f(\xi_t, \tilde \xi_t, \varphi_{t+1})^2 \right]^{1/2}  \mb E^{\tilde Q_0} \left[ \exp(2 |g(\tilde \xi_{t+1})/c|^s \right]^{1/2}  \\
 & \leq  \left( 2 \mb E^{\tilde Q_0 \otimes \tilde Q}\! \left[ m_f (\xi_t,\tilde \xi_t,\varphi_{t+1})^2 \right] \right)^{1/2}  \,,
\end{align*}
because $c = 2^{1/s} \|g\|_{\phi_s}$. The expectation on the right-hand side is finite  because $f \in E^{\phi_s}_{\tilde \xi}$ and $u \in E^{\phi_r}_\varphi$. It follows by Lemma \ref{lem-pollard} that $\| \tilde{ \mb T}(f + g) - \tilde{\mb T} f\|_{\phi_s} \to 0$ as $\| g \|_{\phi_s} \to 0$.

Finally, monotonicity follows from monotonicity of the exponential and logarithm functions and monotonicity of conditional expectations. Convexity follows by H\"older's inequality.
\end{proof}

\begin{lemma}\label{lem-D-bdd-learn}
Let Assumption U-Learn hold and fix any $v \in E^{\phi_{r'}}_{\tilde \xi}$ with $r' \geq 1$. Then: for all $1 \leq s < \infty$, $\tilde{\mb D}_v$ and $\tilde{\mb E}_v$ are continuous linear operators on $E^{\phi_s}_{\tilde \xi}$ with $\rho(\tilde{\mb D}_v;E^{\phi_s}_{\tilde \xi})  < 1$.
\end{lemma}

\begin{proof}[Proof of Lemma \ref{lem-D-bdd-learn}]
Fix any $s \geq 1$, take $f \in E^{\phi_s}_{\tilde \xi}$ and $c > 0$. By convexity of $\phi_s$:
\begin{align*}
 \mb E^{\tilde Q_0} \left[ \exp(| \tilde{\mb D}_vf (\tilde \xi_t)/(\beta c)|^s) \right] 
 & \leq \mb E^{\tilde Q_0} \left[ \mb E^{\Pi_\xi}_v\! \left[ \left. \mb E^{Q_\varphi}_v \left[ \left.  \exp \left(  \left| \frac{1}{c} f(\Xi(\tilde \xi_t,\varphi_{t+1}))  \right|^s \right)  \right| \xi_t, \tilde \xi_t \right]  \right| \tilde \xi_t \right]\right] \\ 
 & \leq \mb E^{\tilde Q_0 \otimes \tilde Q}\! \left[ m_v(\xi_t, \tilde \xi_t, \varphi_{t+1})^2 \right]^{1/2}  \mb E^{\tilde Q_0} \left[ \exp(2 |f(\tilde \xi_{t+1})/c|^s \right]^{1/2}  \,.
 \end{align*} 
The first term on the right-hand side is finite because $v \in E^{\phi_{r'}}_{\tilde \xi}$ and $u \in E^{\phi_r}_\varphi$ and the second term is finite for any $c > 0$ because $f \in E^{\phi_s}_{\tilde \xi}$. Continuity follows by taking $c = 2^{1/s} \|f\|_{\phi_s}$ and applying Lemma \ref{lem-pollard}. 

The spectral radius may be controlled by similar arguments to Lemma \ref{lem:sr-mp}, replacing the term $\|m_g^{\otimes n}\|_p$ in Lemma \ref{lem:sr-mp} by the term:
\[
 \tilde m_{n,p} := \mb E^{\tilde Q_0 \otimes \tilde Q } \left[ m_{v,t+1}^p \, \mb E^{\tilde Q} \left[ \left. m_{v,t+2}^p  \mb E^{\tilde Q} \left[ \left.  \cdots \mb E^{\tilde Q} \left[ \left. m_{v,t+n+1}^p \right| \tilde X_{t+n} \right] \cdots \right| \tilde X_{t+2} \right] \right| \tilde X_{t+1} \right]  \right]^{1/p} 
\]
where $m_{v,t+1} := m_v(\xi_t, \tilde \xi_t, \varphi_{t+1})$. By H\"older's inequality and stationarity of $(\xi_t,\tilde X_t)$:
\[
 \tilde m_{n,p} \leq \mb E^{\tilde Q_0 \otimes \tilde Q } \left[ m_{v,t+1}^{pn} \right]^{1/p} \leq \mb E^{\tilde Q_0 \otimes \tilde Q } \left[ m_v^{\Pi_\xi}(\xi_t,\tilde \xi_t) ^{2pn} \right]^{1/2p} \mb E^{\tilde Q_0 \otimes \tilde Q } \left[ m_v^{Q_\varphi}(\xi_t,\tilde \xi_t,\varphi_{t+1}) ^{2pn} \right]^{1/2p}  \,.
\]
As $v \in E^{\phi_{r'}}_{\tilde \xi}$ and $u \in E^{\phi_r}_{\varphi}$ with $r,r' > 1$, each of the two expectations may be controlled by similar arguments to Lemmas \ref{lem:msubexp} and \ref{lem:com}, leading to algebraic, rather than exponential, rates of growth of $\log \tilde m_{n,p}$ in $n$. 
\end{proof}

\begin{proof}[Proof of Theorem \ref{t-id-W-learn}]
The proof follows similar arguments to the proof of Theorem \ref{t-id-W}, substituting Lemmas \ref{lem-T-prop-learn} and \ref{lem-D-bdd-learn} for Lemmas  \ref{lem-T-prop} and \ref{lem-D-bdd}. We describe only the necessary modifications. For existence, it suffices to construct a $\bar v \in E^{\phi_r}_{\tilde \xi}$ for which $\tilde{\mb T} \bar v \leq \bar v$. If $\vartheta \geq \theta$, let
\[
 \bar v(\tilde \xi) = (1-\beta) \sum_{n=0}^\infty \beta^{n+1} \log \left( \Big(\mb E^{\tilde Q}\Big)^{n+1} \big( e^{\frac{\alpha\vartheta}{(1-\beta)\theta} u} \big)(\tilde \xi) \right) \,, 
\]
where  $(\mb E^{\tilde Q})^n$ denotes applying $\mb E^{\tilde Q}$ to a function $n$ times in succession. Recall that $\tilde Q_0$ is the stationary distribution of $\tilde X_t$ under $\tilde Q$. For any $c \in (0,1]$, by Jensen's inequality we may deduce:
\begin{align*}
 \mb E^{\tilde Q_0}[ e^{|\bar v(\tilde \xi_t)/(\beta c)|^r}] & \leq (1-\beta) \sum_{n=0}^\infty \beta^n \mb E^{\tilde Q_0}  \left[\left( \Big(\mb E^{\tilde Q}\Big)^{n+1} \big( e^{|\frac{\alpha\vartheta}{(1-\beta)\theta c} u|^r} \big)(\tilde \xi_t) \right) \right] \,.
\end{align*}
The right-hand side is finite because $u \in E^{\phi_r}_\varphi$, hence $\bar v \in \tilde E^{\phi_r}$. To show $\tilde{\mb T} \bar v \leq \bar v$, first note:
\begin{align*}
 \tilde{\mb T} \bar v(\tilde \xi) & = \beta \log \mb E^{\Pi_\xi}\! \left[ \left. \mb E^{Q_\varphi} \left[ \left.  e^{\frac{\theta}{\vartheta} \bar v(\Xi(\tilde \xi_t,\varphi_{t+1})) + \alpha u(\varphi_{t+1})} \right| \xi_t, \tilde \xi_t \right]^{\vartheta/\theta}  \right| \tilde \xi_t = \tilde \xi \right] \\
 & \leq \beta \log \mb E^{\tilde Q} \left[ \left.  e^{ \bar v(\tilde \xi_{t+1}) + \alpha \frac{\vartheta}{\theta} u(\varphi_t,\varphi_{t+1})}  \right| \tilde \xi_t = \tilde \xi \right] \\ 
 & \leq \beta^2 \log \mb E^{\tilde Q} \left[ \left.  e^{ \bar v(\tilde \xi_{t+1})/\beta }  \right| \tilde \xi_t = \tilde \xi \right] +  \beta(1-\beta) \log \mb E^{\tilde Q} \left[ \left.  e^{  \frac{\alpha\vartheta}{(1-\beta)\theta} u(\varphi_{t+1})}  \right| \tilde \xi_t = \tilde \xi \right]  \,.
\end{align*}
By a version of H\"older's inequality for infinite produces (see, e.g., \cite{Karakostas}), we obtain:
\[
 \log \mb E^{\tilde Q} \left[ \left.  e^{ \bar v(\tilde X_{t+1})/\beta }  \right| \tilde X_t = \tilde x \right] \leq (1-\beta) \sum_{n=1}^\infty \beta^{n-1} \log \left( \Big(\mb E^{\tilde Q}\Big)^{n+1} \big( e^{\frac{\alpha\vartheta}{(1-\beta)\theta} u} \big)(\tilde x) \right) 
\]
hence $\tilde{\mb T} \bar v \leq \bar v$. On the other hand, if $\vartheta \leq \theta$, let
\[
 \bar v(\tilde x) = \frac{\vartheta}{\theta} (1-\beta) \sum_{n=0}^\infty \beta^{n+1} \log \left( \Big(\mb E^{\tilde Q}\Big)^{n+1} \big( e^{\frac{\alpha}{1-\beta} u} \big)(\tilde x) \right) \,.
\]
we may deduce $\bar v \in E^{\phi_r}_{\tilde \xi}$ under the condition $u \in E^{\phi_r}_\varphi$ using similar arguments to the above. By Jensen's inequality and H\"older's inequality:
\begin{align*}
 \tilde{\mb T} \bar v(\tilde \xi) & = \beta \log \mb E^{\Pi_\xi}\! \left[ \left. \mb E^{Q_\varphi} \left[ \left.  e^{\frac{\theta}{\vartheta} \bar v(\Xi(\tilde \xi_t,\varphi_{t+1})) + \alpha u(\varphi_{t+1})} \right| \tilde \xi_t, \xi_t \right]^\frac{\vartheta}{\theta}  \right| \tilde \xi_t = \tilde \xi \right] \\
 & \leq \frac{\vartheta}{\theta} \beta \log \mb E^{\tilde Q} \left[ \left. e^{\frac{\theta}{\vartheta} \bar v(\tilde \xi_{t+1}) + \alpha u(\varphi_{t+1})}  \right| \tilde \xi_t = \tilde \xi \right] \\ 
 & \leq \frac{\vartheta}{\theta}  \beta^2 \log \mb E^{\tilde Q} \left[ \left.  e^{ \frac{\theta}{\vartheta} \bar v(\tilde \xi_{t+1})/\beta }  \right| \tilde \xi_t = \tilde \xi \right] + \frac{\vartheta}{\theta}  \beta(1-\beta) \log \mb E^{\tilde Q} \left[ \left.  e^{  \frac{\alpha}{1-\beta} u(\varphi_{t+1})}  \right| \tilde \xi_t = \tilde \xi \right]  \,.
\end{align*}
and the result $\tilde{\mb T}\bar v \leq \bar v$ now follows by similar arguments to the previous case. The remainder of the proof of existence follows the same arguments to the proof of Theorem \ref{t-id-W}. 

For uniqueness, $v$ is necessarily a fixed point of $\tilde{\mb T} : \tilde E^{\phi_s} \to \tilde E^{\phi_s}$ for each $1 \leq s \leq r$. Moreover, $\tilde{\mb T} : \tilde E^{\phi_s} \to \tilde E^{\phi_s}$ is convex by Lemma \ref{lem-T-prop-learn} and the subgradient $\tilde{\mb D}_v$ at any $v \in \tilde E^{\phi_s}$ is a bounded, monotone linear operator with $\rho(\tilde{\mb D}_v;\tilde E^{\phi_s}) < 1$ by Lemma \ref{lem-D-bdd-learn}.  Uniqueness now follows by Proposition \ref{p:exun}(ii).
\end{proof}

\subsection{Proofs for Section \ref{s:estimation}}\label{ax:estimation:proofs}

\begin{proof}[Proof of Lemma \ref{l-perturb}]
Proof under AM(a). To establish existence and uniqueness under condition AM(a), note that $\hat{\mb T} : E^{\phi_s} \to E^{\phi_s}$ given by
\[
 \hat{\mb T}f(x) = \beta \log \mb E^Q\left[ \left. e^{ f(X_{t+1}) + \hat \ell(X_t,X_{t+1}) + \alpha u(X_t,X_{t+1}))} \right|X_t =x \right]
\]
is of the same form as $\mb T : E^{\phi_s} \to E^{\phi_s}$ but with $\hat \ell + \alpha u$ in place of $\alpha u$. By analogous arguments to the proof of Theorem \ref{t-id-W}, we obtain existence of a fixed point $\hat v \in E^{\phi_r}$.  Uniqueness follows by identical arguments to the proof of Theorem \ref{t-id-W}.

Proof under AM(b). First note $u \in E^{\phi_r}(\hat Q_0 \otimes \hat Q)$ by Lemma \ref{lem:embed:0} under Assumption U and the condition $\mb E^{Q_0 \otimes Q}[ \hat \Delta_2(X_t,X_{t+1})^p < \infty]$. If $u(X_t,X_{t+1})$ is a function of either $X_t$ or $X_{t+1}$ then Assumption U implies $u \in E^{\phi_r}(Q_0)$ and we may use Lemma \ref{lem:embed:0} and the condition $\mb E^{Q_0}[ \hat \Delta_{0}(X_t)^p ] < \infty$ to deduce $u \in E^{\phi_r}(\hat Q_0)$ and hence that $u \in E^{\phi_r}_2(\hat Q_0 \otimes \hat Q)$. Therefore, $\hat{\mb T}$ is a well-defined operator on $E^{\phi_s}(\hat Q_0)$ for each $1 \leq s \leq r$ and it inherits the properties of $\mb T$ established in Lemmas \ref{lem-T-prop} and \ref{lem-D-bdd}. We may therefore deduce from Theorem \ref{t-id-W} that $\hat{\mb T}$ has a fixed point $\hat v \in E^{\phi_r}(\hat Q_0)$ and that this is the unique fixed point in $E^{\phi_s}(\hat Q_0)$ for each $1 < s \leq r$. Lemma \ref{lem:equiv:0} implies that  $E^{\phi_s}(Q_0)= E^{\phi_s}(\hat Q_0)$ and the norms $\|\cdot\|_{L^{\phi_s}(Q_0)}$ and $\|\cdot\|_{L^{\phi_s}(\hat Q_0)}$ are equivalent. Therefore $\hat{\mb T}$ is a well-defined operator on $E^{\phi_s}(Q_0)$ for each $1 \leq s \leq r$, $\hat v \in E^{\phi_r}(Q_0)$ is a fixed point of  $\hat{\mb T}$, and $\hat v$ is the unique fixed point in $E^{\phi_s}(Q_0)$ for each $1 < s \leq r$.
\end{proof}

\begin{lemma}\label{lem-vhat-exist-bdd}
Let Assumptions U and AM hold. Then, under condition AM(a):
\[
 \|\hat v \|_{\phi_s} \leq \frac{\beta}{1-\beta} \| \hat \ell + \alpha u\|_{\phi_s} + \beta \Big(\Big(\mb E^{Q_0 \otimes Q} \Big[ e^{|\frac{1}{1-\beta} (\hat \ell(X_t,X_{t+1}) + \alpha u(X_t,X_{t+1}))|^s}\Big] - 1\Big) \vee 1\Big) 
\]
and under condition AM(b):
\[
 \|\hat v\|_{L^{\phi_s}(\hat Q_0)} \leq \frac{\alpha \beta}{1-\beta} \| u\|_{L^{\phi_s}_2(\hat Q_0)} + \beta \Big(\Big(\mb E^{\hat Q_0 \otimes \hat Q}\Big[ e^{|\frac{\alpha}{1-\beta} u(X_t,X_{t+1})|^s}\Big] - 1\Big) \vee 1\Big)\,. 
\]
\end{lemma}

\begin{proof}[Proof of Lemma \ref{lem-vhat-exist-bdd}]
Proof under AM(a). By analogous arguments to the proof of Theorem \ref{t-id-W}, we obtain existence of a fixed point $\hat v \in E^{\phi_r}$ with $\ul{\hat v} \leq \hat v \leq \ol{\hat v}$ where
\[
 \ol {\hat v}(x) = (1-\beta) \sum_{n=0}^\infty \beta^{n+1} \log \mb E^{Q} \Big[ e^{\frac{1}{1-\beta} \hat \ell(X_{t+n},X_{t+n+1}) + \frac{\alpha}{1-\beta} u(X_{t+n},X_{t+n+1})} \Big|X_t = x \Big] \
\]
and
\[
 \ul{\hat v}(x) = (\mb I - \beta \mb E^Q)^{-1}(\beta \mb E^Q(\hat \ell + \alpha u)) \,.
\]
Therefore, by similar arguments to the proof of Theorem \ref{t-id-W}:
\[
 \|\ol{\hat v} \|_{\phi_s} \leq \beta \Big(\Big(\mb E^{Q_0 \otimes Q} \Big[ e^{|\frac{1}{1-\beta} \hat \ell(X_t,X_{t+1}) + \frac{\alpha}{1-\beta} u(X_t,X_{t+1})|^s}\Big] - 1\Big) \vee 1\Big) 
\]
and also $\|\ul{\hat v}\|_{\phi_s} \leq \frac{\beta}{1-\beta} \| \hat \ell + \alpha u\|_{\phi_s}$ for each $1 \leq s \leq r$. The bound for $\|\hat v\|_{\phi_s}$ follows because $\|\cdot\|_{\phi_s}$ is a lattice norm.

Proof under AM(b). By following identical steps to the proof of Theorem \ref{t-id-W} substituting $\hat{ \mb T}$ in place of $\mb T$, $\hat Q$ in place of $Q$, and $\hat Q_0$ in place of $Q_0$, we see that we have $\ul {\hat v} \leq \hat v \leq \ol{\hat v}$ where
\[
 \ol{\hat v}(x) = (1-\beta) \sum_{n=0}^\infty \beta^{n+1} \log \mb E^{\hat Q} \Big[ e^{\frac{\alpha}{1-\beta} u(X_{t+n},X_{t+n+1})} \Big|X_t = x \Big] 
\]
with
\[
 \|\ol{\hat v} \|_{L^{\phi_s}(\hat Q_0)} \leq \beta \Big(\Big(\mb E^{\hat Q_0 \otimes \hat Q} \Big[ e^{|\frac{\alpha}{1-\beta} u(X_t,X_{t+1})|^s}\Big] - 1\Big) \vee 1\Big) 
\]
and
\[
 \ul{\hat v}(x) = (\mb I - \beta \mb E^{\hat Q})^{-1}( \beta \mb E^{\hat Q}(\alpha u)) 
\]
with $\|\ul{\hat v} \|_{L^{\phi_s}(\hat Q_0)} \leq \frac{\alpha \beta}{1-\beta} \| u\|_{L^{\phi_s}_2(\hat Q_0 \otimes \hat Q)}$. The result follows because $\|\cdot\|_{\phi_s}$ is a lattice norm.
\end{proof}

The extended subgradient of $\mb T$ at $g \in N^{\phi_s}_2$ is defined as:
\[
 \mb D_g f(x) = \beta \mb E_g f(x) = \beta \mb E[m_g(X_t,X_{t+1}) f(X_t,X_{t+1})|X_t = x]
\]
with 
\[
 m_g(X_t,X_{t+1}) = \frac{e^{g(X_t,X_{t+1}) + \alpha u (X_t,X_{t+1})}}{\mb E^Q[e^{g(X_t,X_{t+1}) + \alpha u (X_t,X_{t+1})}|X_t]} \,.
\]
If $g(x_0,x_1) = v(x_1)$ for $v \in E^{\phi_s}$ and $\mb T$ and $\mb D_v$ are restricted to $\{f : f(x_0,x_1) = f_1(x_1)$ with $f_1 \in E^{\phi_s}\}$ then the extended operators reduce to the operators analyzed in Section \ref{s:id}.

\begin{lemma}\label{l-perturb-full}
Let Assumptions U and AM hold. Then:
\begin{equation} \label{e:ineq-WW-44}
 (\mb I - \mb D_{\hat \eta + \hat v})^{-1} (\hat{\mb T} v - v) 
 \geq \hat v - v 
 \geq (\mb I-\mb D_{\hat \eta + v})^{-1} (\hat{\mb T} v - v)
\end{equation}
and:
\begin{equation} \label{e:ineq-WW-55}
 \| \hat v - v\|_{\phi_1} \leq \|(\mb I-\mb D_{\hat \eta + v})^{-1} (\hat{\mb T} v - v)\|_{\phi_1} + \|(\mb I - \mb D_{\hat \eta + \hat v})^{-1} (\hat{\mb T} v - v) \|_{\phi_1} \,.
\end{equation}
Inequality (\ref{e:ineq-WW-55}) also holds in $\|\cdot\|_{\phi_s}$ norm for every $1 \leq s \leq r$ under Assumption AM(a).
\end{lemma}

\begin{proof}[Proof of Lemma \ref{l-perturb-full}]
We first show that $(\mb I - \mb D_{\hat \eta + v})$ and $(\mb I - \mb D_{\hat \eta + \hat v})$ are continuously invertible on $L^{\phi_1}$.

To prove this claim under AM(a), by Assumptions U and AM and Lemmas \ref{lem:sr-mp} and \ref{lem:msubexp}, we have $\rho(\mb D_{\hat \eta + v};L^{\phi_s}) < 1$ and $\rho(\mb D_{\hat \eta + \hat v};L^{\phi_s}) <1 $ for all $s \geq 1$, so $(\mb I - \mb D_{\hat \eta + v})$ and $(\mb I - \mb D_{\hat \eta + \hat v})$ are continuously invertible on $L^{\phi_s}$ for all $s \geq 1$. 

To prove this claim under condition AM(b), the proof of Lemma \ref{lem-vhat-exist-bdd} shows that $E^{\phi_r}(Q_0)= E^{\phi_r}(\hat Q_0)$, $\hat v \in  E^{\phi_r}(\hat Q_0)$, and $u \in E^{\phi_r}(\hat Q_0 \otimes \hat Q)$. Also note that $\mb D_{\hat \eta + v} = \mb D_{\hat \ell + v}$ and $\mb D_{\hat \eta + \hat v} = \mb D_{\hat \ell + \hat v}$. By Lemmas \ref{lem:sr-mp} and \ref{lem:msubexp}, we may therefore deduce that $\mb D_{\hat \eta + v}$ and $\mb D_{\hat \eta + \hat v}$ are continuous linear operators on $L^{\phi_s}(\hat Q_0)$ for all $s \geq 1$ with $\rho(\mb D_{\hat \eta + v} ; L^{\phi_s}(\hat Q_0)) < 1$ and $\rho(\mb D_{\hat \eta + \hat v} ; L^{\phi_s}(\hat Q_0)) < 1$. But $L^{\phi_s}(Q_0) = L^{\phi_s}(\hat Q_0)$ and their norms are equivalent by Lemma \ref{lem:equiv:0}. Therefore, $\mb D_{\hat \eta + v}$ and $\mb D_{\hat \eta + \hat v}$ are continuous linear operators on $L^{\phi_s}(Q_0)$ and their spectral radius does not depend on the norm, i.e., $\rho(\mb D_{\hat \eta + v} ; L^{\phi_s}(Q_0))  = \rho(\mb D_{\hat \eta + v} ; L^{\phi_s}(\hat Q_0)) < 1$ and $\rho(\mb D_{\hat \eta + \hat v} ; L^{\phi_s}(Q_0)) = \rho(\mb D_{\hat \eta + \hat v} ; L^{\phi_s}(\hat Q_0)) < 1$. Hence, $(\mb I - \mb D_{\hat \eta + v})$ and $(\mb I - \mb D_{\hat \eta + \hat v}) $ are continuously invertible on $L^{\phi_s}$ for all $s \geq 1$.

We now establish inequalities (\ref{e:ineq-WW-44}) and (\ref{e:ineq-WW-55}). First deduce that $\hat{\mb T} v \in L^{\phi_1}$. For the lower bound:
\begin{align} 
 \hat v(x) - v(x) 
 & = \hat{\mb T} \hat v (x) - \hat{\mb T} v (x) + \hat{\mb T} v (x) - v (x) \notag \\
 & = \beta \log \mb E_{\hat \ell + v} \left[ \left. e^{\hat v(X_{t+1}) - v(X_{t+1})} \right| X_t = x \right] + \hat{\mb T} v (x) - v (x)  \notag \\
 & \geq \mb D_{\hat \ell + v} (\hat v - v) (x) + \hat{\mb T} v (x) - v (x) \notag \\
 & =  \mb D_{\hat \eta + v} (\hat v - v) (x) + \hat{\mb T} v (x) - v (x)  \label{e:ineq-WW-10}
\end{align}
($Q_0$-a.e.).  Applying $(\mb I - \mb D_{\hat \eta + v})^{-1} : L^{\phi_1} \to L^{\phi_1}$ to both sides of (\ref{e:ineq-WW-10}), using monotonicity of $(\mb I - \mb D_{\hat \eta + v})^{-1}$, and rearranging, we obtain:
\begin{equation} \label{e:ineq-WW-1}
 \hat v - v \geq (\mb I-\mb D_{\hat \eta + v})^{-1} (\hat{\mb T} v - v) \,.
\end{equation}
For the upper bound:
\begin{align} 
 v(x) - \hat v(x) 
 & = v(x) - \hat{\mb T} v(x) + \hat{\mb T} v(x) - \hat{\mb T} \hat v(x) \notag \\
 & = v(x) - \hat{\mb T} v(x) + \beta \log \mb E_{\hat \ell + \hat v} \left[ \left. e^{v(X_{t+1}) - \hat v(X_{t+1})} \right| X_t = x \right] \notag \\
 & \geq v(x) - \hat{\mb T} v(x) + \mb D_{\hat \ell + \hat v}(v - \hat v)(x) \notag \\
 & = v(x) - \hat{\mb T} v(x) + \mb D_{\hat \eta + \hat v}(v - \hat v)(x)  \label{e:ineq-WW-20} 
\end{align}
($Q_0$-a.e.). Applying $(\mb I - \mb D_{\hat \eta + \hat v})^{-1}$ to both sides of (\ref{e:ineq-WW-20}), using monotonicity, and rearranging, we obtain:
\begin{equation} \label{e:ineq-WW-2}
 \hat v - v \leq (\mb I - \mb D_{\hat \eta + \hat v})^{-1} (\hat{\mb T} v - v)\,.
\end{equation}
Combining (\ref{e:ineq-WW-1}) and (\ref{e:ineq-WW-2}):
\begin{align} \label{e:ineq-WW-33}
 (\mb I - \mb D_{\hat \eta + \hat v})^{-1} (\hat{\mb T} v - v) 
 \geq \hat v - v 
 \geq (\mb I-\mb D_{\hat \eta + v})^{-1} (\hat{\mb T} v - v)\,.
\end{align}
Using the fact that $\|\cdot\|_{\phi_1}$ is a lattice norm, it follows from (\ref{e:ineq-WW-33}) that:
\[
 \| \hat v - v\|_{\phi_1}  \leq  \|(\mb I-\mb D_{\hat \eta + v})^{-1} (\hat{\mb T} v - v)\|_{\phi_1} + \|(\mb I - \mb D_{\hat \eta + \hat v})^{-1} (\hat{\mb T} v - v) \|_{\phi_1} 
\]
as required. 
Under condition AM(a) we may in fact deduce $\hat{\mb T} v \in L^{\phi_r}$, in which case the preceding inequality can be restated to hold in $\|\cdot\|_{\phi_r}$.
\end{proof}

Recall that $\kappa_{\hat \eta}(x) = \log \mb E^Q[e^{\hat \eta(X_t,X_{t+1})}|X_t = x]$.

\begin{lemma}\label{lem-kappa-lipschitz}
Let $\| \hat \eta\|_{\phi_s} \leq 1$ for some $s \geq 1$. Then: $\|\kappa_{\hat \eta}\|_{\phi_s} \leq \| \hat \eta\|_{\phi_s}$.
\end{lemma}

\begin{proof}[Proof of Lemma \ref{lem-kappa-lipschitz}]
By convexity of $x \mapsto e^{\frac{1}{c^s}|\log x|^s}$ for $c \leq 1$ and Jensen's inequality:
\begin{align*}
 \mb E^{Q_0}[\exp(|\kappa_{\hat \eta}(X_t)/\|\hat \eta\|_{\phi_s}|^s)]
 & = \mb E^{Q_0} \left[\exp \left(\left|\frac{1}{\|\hat \eta\|_{\phi_s}}\log \mb E^Q \left[ \left. e^{\hat \eta(X_t,X_{t+1})}\right|X_t\right]\right|^s\right)\right] \\ 
 & \leq \mb E^{Q_0} \left[ \mb E^Q \left[ \left. \exp \left(\left|\frac{1}{\|\hat \eta\|_{\phi_s}}\log  e^{\hat \eta(X_t,X_{t+1})}\right|^s\right) \right|X_t\right] \right] \\ 
 & = \mb E^{Q_0 \otimes Q} \left[  \exp \left(\left|\frac{\hat \eta(X_t,X_{t+1})}{\|\hat \eta\|_{\phi_s}}\right|^s\right) \right] = 2\,.
\end{align*}
The result follows by Lemma \ref{lem-pollard}.
\end{proof}

\begin{lemma}\label{lem-v-bdd}
Let Assumption U hold. Then:
\[
 \|v\|_{\phi_r} \leq \frac{\alpha \beta}{1-\beta} \| u\|_{\phi_r} + \beta \Big(\Big(\mb E^{Q_0 \otimes Q}\Big[ e^{|\frac{\alpha}{1-\beta} u(X_t,X_{t+1})|^r}\Big] - 1\Big) \vee 1\Big)\,. 
\]
\end{lemma}

\begin{proof}[Proof of Lemma \ref{lem-v-bdd}]
From the proof of Theorem \ref{t-id-W}, we have $\ul v \leq v \leq \ol v$ where
\begin{align*}
 \|\ol v \|_{\phi_r} & \leq \beta \Big(\Big(\mb E^{Q_0 \otimes Q} \Big[ e^{|\frac{\alpha}{1-\beta} u(X_t,X_{t+1})|^r}\Big] - 1\Big) \vee 1\Big) \\
 \|\ul v \|_{\phi_r} & = \| (\mb I - \beta \mb E^Q)^{-1}( \beta \mb E^Q(\alpha u)) \|_{\phi_r} \leq \frac{\alpha \beta}{1-\beta} \| u\|_{\phi_r} \,.
\end{align*}
The result follows because $\|\cdot\|_{\phi_r}$ is a lattice norm.
\end{proof}

\begin{lemma}\label{lem-T-lipschitz}
Let Assumptions U and AM2 hold and let $\| \hat \eta \|_{\phi_1} \leq 1$. Then: there exists a constant $C$ which depends only on $\alpha$, $\beta$, $\|u\|_{\phi_r}$, $r$, and $M$ under AM2(a) and only on $\alpha$, $\beta$, $\|u\|_{\phi_r}$, $r$, $M$, and $p$ under AM2(b) such that $\|\hat{\mb T}v - v \|_{\phi_1} \leq C\| \hat \eta\|_{\phi_1}$.
\end{lemma}

\begin{proof}[Proof of Lemma \ref{lem-T-lipschitz}]
By convexity of $\mb T$ and the fact that $\hat{\mb T}v = \mb T(\hat \eta + v) - \beta \kappa_{\hat \eta}$, we have:
\begin{equation}\label{e-T-cvx-bd}
  \mb D_{\hat \eta + v} \hat \eta - \beta \kappa_{\hat \eta} \geq \hat{\mb T}v - v \geq  \mb D_v \hat \eta - \beta \kappa_{\hat \eta}
\end{equation}
and so $\|\hat{\mb T}v - v\|_{\phi_1} \leq ( \|\mb D_{\hat \eta + v} \hat \eta\|_{\phi_1} + \|\mb D_v \hat \eta \|_{\phi_1} + 2 \beta \|\hat \eta\|_{\phi_1})$ by the lattice property and Lemma \ref{lem-kappa-lipschitz}. By similar arguments to the proof of Lemma \ref{lem-D-bdd}:
\[
 \| \mb D_v \hat \eta \|_{\phi_1} \leq 2\beta  ( ( 2^\frac{1}{2} \|m_v\|_2 - 1 ) \vee 1 ) \| \hat \eta \|_{\phi_1} \,.
\]
It follows by Lemmas \ref{lem:com} and \ref{lem-v-bdd} that $\|m_v\|_2$ can be bounded by a term depending only on $\beta$, $\alpha$, $\|u\|_{\phi_r}$, $r$, and $M$. The term $\| \mb D_{\hat \eta + v} \hat \eta \|_{\phi_1}$ is controlled similarly under AM2(a), noting $\mb D_{\hat \eta + v} = \mb D_{\hat \ell + v}$, and is therefore bounded by a term depending only on $\beta$, $\alpha$, $\|u\|_{\phi_r}$, $r$, and $M$.

Under AM2(b), we first bound $\| \mb D_{\hat \eta + v} \hat \eta \|_{L^{\phi_1}(\hat Q_0)}$ then use equivalence of the $L^{\phi_1}(Q_0)$ and $L^{\phi_1}(\hat Q_0)$ norms (cf. Lemma \ref{lem:equiv:0}). By similar arguments to the proof of Lemma \ref{lem-D-bdd}:
\[
 \| \mb D_{\hat \eta + v} \hat \eta \|_{L^{\phi_1}(\hat Q_0)} = \| \mb D_{\hat \ell + v} \hat \eta \|_{L^{\phi_1}(\hat Q_0)} \leq 2\beta  ( ( 2^\frac{1}{2} \|\hat m_v\|_{L^2(\hat Q_0 \otimes \hat Q)} - 1 ) \vee 1 ) \| \hat \eta \|_{L^{\phi_1}_2(\hat Q_0 \otimes \hat Q)}
\]
where $\hat m_v(X_t,X_{t+1}) = \frac{e^{v(X_{t+1}) + \alpha u(X_t,X_{t+1})}}{\mb E^{\hat Q}[{e^{v(X_{t+1}) + \alpha u(X_t,X_{t+1})}}|X_t]}$. 
By Lemma \ref{lem:com}, $\|\hat m_v\|_{L^2(\hat Q_0 \otimes \hat Q)}$ is bounded by a term depending only on $\beta$, $\alpha$, $\|v\|_{L^{\phi_r}(\hat Q_0)}$, $\|u\|_{L^{\phi_r}_2(\hat Q_0 \otimes \hat Q)}$, and $r$ and hence (by Lemmas \ref{lem:embed:0} and \ref{lem-v-bdd}) by a bound depending only on $\beta$, $\alpha$, $\|u\|_{\phi_r}$, $r$, $p$, and $M$. Similarly, $\| \hat \eta \|_{L^{\phi_1}_2(\hat Q_0 \otimes \hat Q)} \leq C' \| \hat \eta \|_{\phi_1}$ for $C'$ depending only on $p$ and $M$ by Lemma \ref{lem:embed:0}.
\end{proof}

\begin{proof}[Proof of Lemma \ref{l-consistent}]
In view of (\ref{e:ineq-WW-55}) and Lemma \ref{lem-T-lipschitz}, to prove $\| \hat v - v \|_{\phi_1} \leq C \| \hat \eta\|_{\phi_1}$ and $\| \hat v - v \|_{\phi_1} \leq C \| \hat{\mb T} v - v\|_{\phi_1}$ it suffices to show that there are finite positive constants $C_1$, $C_2$ depending on primitives as described, such that:
\begin{align*} 
 \|(\mb I-\mb D_{\hat \eta + v})^{-1} \|_{L^{\phi_1}} & \leq C_1 \,, &
 \|(\mb I-\mb D_{\hat \eta + \hat v})^{-1} \|_{L^{\phi_1}} & \leq C_2\,. 
\end{align*}
We control $C_1$ and $C_2$ by different arguments under AM2(a) and AM2(b). Under AM2(a), recall that $\mb D_{\hat \eta + v} = \mb D_{\hat \ell + v}$ and $\mb D_{\hat \eta + \hat v} = \mb D_{\hat \ell + \hat v}$. We may deduce from Lemma \ref{lem:msubexp} that there are constants $A_1,A_2 \in (0,\infty)$ and $a_1,a_2 \in (0,1-\beta)$ such that $\| m_{\hat \ell + v}^{\otimes n}\|_p \leq A_1 e^{(\beta+a_1)^{-n}}$ and $\| m_{\hat \ell + \hat v}^{\otimes n}\|_p \leq A_2 e^{(\beta+a_2)^{-n}}$ for each $n \geq 1$, and that $A_1,a_1,A_2,a_2$ depend only on $\beta$, $r$, and either $\|v + \hat \ell + \alpha u \|_{\phi_r}$ (for $A_1,a_1$)  or $\|\hat v + \hat \ell + \alpha u \|_{\phi_r}$ (for $A_2,a_2$). Therefore, by Lemma \ref{lem:c-bdd2} there are constants $C_1$ and $C_2$ depending only on $\beta$, $r$ and either $\|v + \hat \ell + \alpha u \|_{\phi_r}$ (for $C_1$)  or $\|\hat v + \hat \ell + \alpha u \|_{\phi_r}$ (for $C_2$) that satisfy the above bounds. The constants $C_1$ and $C_2$ may be further increased to depend only on $\beta$, $\alpha$, $\|u\|_{\phi_r}$, $r$, $M$ and either $\| v\|_{\phi_r}$ (for $C_1$) or $\| \hat v\|_{\phi_r}$ (for $C_2$). By Lemma \ref{lem-v-bdd}, $C_1$ can be further increased to depend only on $\beta$, $\alpha$, $\|u\|_{\phi_r}$, $r$ and $\mb E^Q[ e^{|\frac{\alpha}{1-\beta} u(X_t,X_{t+1})|^r}]$. Moreover, by Lemma \ref{lem-vhat-exist-bdd}: 
\[
 \|\hat v\|_{\phi_r} \leq \frac{\beta}{1-\beta} ( \| \hat \ell\|_{\phi_r} + \alpha \| u\|_{\phi_r}) + \beta \Big(\Big(\mb E^{Q_0 \otimes Q}\Big[ e^{|\frac{1}{1-\beta} (\hat \ell(X_t,X_{t+1}) + \alpha u(X_t,X_{t+1}))|^r}\Big] - 1\Big) \vee 1\Big)\,.
\]
It follows that $C_2$ may be further increased to depend only on $\beta$, $\alpha$, $\|u\|_{\phi_r}$, $r$ and $M$.

Under AM2(b), we first bound $\|(\mb I - \mb D_{\hat \eta + v})^{-1} \|_{L^{\phi_1}(\hat Q_0)}$ and  $\|(\mb I - \mb D_{\hat \eta + \hat v})^{-1} \|_{L^{\phi_1}(\hat Q_0)}$ then use equivalence of the norms (cf. Lemma \ref{lem:equiv:0}) to translate these to bounds under $\|\cdot\|_{L^{\phi_1}}$. We may deduce similarly from Lemmas \ref{lem:c-bdd2}, \ref{lem:msubexp}, and \ref{lem:com} that the inequalities
\begin{align*} 
 \|(\mb I-\mb D_{\hat \eta + v})^{-1} \|_{L^{\phi_1}(\hat Q_0)} & \leq C_1 \,, &
 \|(\mb I-\mb D_{\hat \eta + \hat v})^{-1} \|_{L^{\phi_1}(\hat Q_0)} & \leq C_2 
\end{align*}
hold for positive constants $C_1$ and $C_2$ depending only on $\beta$, $\alpha$, $\|u\|_{L^{\phi_r}_2(\hat Q_0)}$, $r$, and either $\|v \|_{L^{\phi_r}(\hat Q_0)}$ (for $C_1$) or $\|\hat v \|_{L^{\phi_r}(\hat Q_0)}$ (for $C_2$). It follows by Lemmas \ref{lem:embed:0}, \ref{lem-vhat-exist-bdd} and \ref{lem-v-bdd} that $C_1$ and $C_2$ can be increased so as to depend only on $\beta$, $\alpha$, $\|u\|_{\phi_r}$, $r$, $p$, $M$ and $\mb E^{Q_0 \otimes Q}[ e^{q^2|\frac{\alpha}{1-\beta} u(X_t,X_{t+1})|^r}]$.

To prove the final inequality $\| \hat v - v \|_{\phi_1} \geq C^{-1} \| \hat \eta\|_{\phi_1}$, first use (\ref{e:ineq-WW-44}) to deduce
\[
 (\mb I-\mb D_{\hat \eta + v})(\hat v - v )
 \geq (\hat{\mb T} v - v)
 \geq (\mb I - \mb D_{\hat \eta + \hat v}) (\hat v - v )
\]
from which it follows that $\| \hat v - v\| \leq (2 + \|\mb D_{\hat \eta + \hat v}\|_{L^{\phi_1}} + \|\mb D_{\hat \eta + v}\|_{L^{\phi_1}}) \| \hat v - v\|_{\phi_1}$. The terms $\|\mb D_{\hat \eta + \hat v}\|_{L^{\phi_1}}$ and $\|\mb D_{\hat \eta + v}\|_{L^{\phi_1}}$ may be controlled as in the proof of Lemma \ref{lem-T-lipschitz}.
\end{proof}

\begin{proof}[Proof of Lemma \ref{l-linear}]
It follows from equations (\ref{e:ineq-WW-44}) and (\ref{e-T-cvx-bd}) that:
\[
 (\mb I - \mb D_{\hat \eta + \hat v})^{-1} (\mb D_{\hat \eta + v} \hat \eta - \beta \kappa_{\hat \eta}) 
 \geq \hat v - v 
 \geq (\mb I-\mb D_{\hat \eta + v})^{-1} ( \mb D_v \hat \eta - \beta \kappa_{\hat \eta}) \,,
\]
where $(\mb I - \mb D_{\hat \eta + \hat v})^{-1}(\beta \kappa_{\hat \eta})$ and $(\mb I-\mb D_{\hat \eta + v})^{-1}(\beta \kappa_{\hat \eta})$ are $O(\|\kappa_{\hat \eta}\|_{\phi_1})$ by the proof of Lemma \ref{l-consistent}. But $\| \kappa_{\hat \eta}\|_{\phi_1} = o(\|\eta\|_{\phi_1})$ by Fr\'echet differentiability of $\kappa_{\eta}$ at $\eta =0$ where the derivative is the conditional expectation operator $h \mapsto \mb E^Q[h(X_t,X_{t+1})|X_t=x]$ and $\mb E^Q[\hat \eta(X_t,X_{t+1})|X_t=x] = 0$ by construction. Therefore:
\[
 (\mb I - \mb D_{\hat \eta + \hat v})^{-1} (\mb D_{\hat \eta + v} \hat \eta) + o(\|\hat \eta\|_{\phi_1}) 
 \geq \hat v - v 
 \geq (\mb I-\mb D_{\hat \eta + v})^{-1} ( \mb D_v \hat \eta ) + o(\|\hat \eta\|_{\phi_1}) \,.
\]
By the second resolvent equation:
\begin{align*}
 \|(\mb I - \mb D_{\hat \eta + \hat v})^{-1} - (\mb I-\mb D_v)^{-1}  \|_{L^{\phi_1}} 
 & \leq \| (\mb I - \mb D_{\hat \eta + \hat v})^{-1}\|_{L^{\phi_1}} \| \mb D_{\hat \eta + \hat v} - \mb D_v\|_{L^{\phi_1}} \| (\mb I-\mb D_v)^{-1} \|_{L^{\phi_1}} \,, \\
 \| (\mb I-\mb D_{\hat \eta + v})^{-1} - (\mb I-\mb D_v)^{-1} \|_{L^{\phi_1}} 
 & \leq \| (\mb I-\mb D_{\hat \eta + v})^{-1} \|_{L^{\phi_1}} \| \mb D_{\hat \eta + v} - \mb D_v \|_{L^{\phi_1}} \| (\mb I-\mb D_v)^{-1} \|_{L^{\phi_1}} \,.
\end{align*}
Note $\| \mb D_{\hat \eta + \hat v} - \mb D_v\|_{L^{\phi_1}} \to 0$ and $\| \mb D_{\hat \eta + v} - \mb D_v\|_{L^{\phi_1}} \to 0$ as $\|\hat \eta\|_{\phi_1} \to 0$ by Lemma \ref{l-consistent} and the assumed continuity condition on $\mb D_v$ on a neighborhood of $v$. Therefore, the inequalities 
\begin{align*}
 \| \mb D_{\hat \eta + \hat v} - \mb D_v\|_{L^{\phi_1}} \| (\mb I-\mb D_v)^{-1} \|_{L^{\phi_1}} & \leq \frac{1}{2} \,, &
 \| \mb D_{\hat \eta + v} - \mb D_v\|_{L^{\phi_1}} \| (\mb I-\mb D_v)^{-1} \|_{L^{\phi_1}} & \leq \frac{1}{2}
\end{align*}
 hold for all $\hat \eta$ sufficiently small, in which case:
\begin{align*}
 \|(\mb I - \mb D_{\hat \eta + \hat v})^{-1} - (\mb I-\mb D_v)^{-1}  \|_{L^{\phi_1}} 
 &\leq 2 \| (\mb I-\mb D_v)^{-1} \|_{L^{\phi_1}}^2 \| \mb D_{\hat \eta + \hat v} - \mb D_v\|_{L^{\phi_1}} = o(1) \,, \\
 \|(\mb I - \mb D_{\hat \eta + v})^{-1} - (\mb I-\mb D_v)^{-1} \|_{L^{\phi_1}} 
 &\leq 2 \| (\mb I-\mb D_v)^{-1} \|_{L^{\phi_1}}^2 \| \mb D_{\hat \eta + v} - \mb D_v\|_{L^{\phi_1}} = o(1)
\end{align*}
 as $\|\hat \eta \| \to 0$. Therefore:
\[
 (\mb I - \mb D_v)^{-1} (\mb D_v \hat \eta) + o(\|\hat \eta\|_{\phi_1}) 
 \geq \hat v - v 
 \geq (\mb I-\mb D_v)^{-1} ( \mb D_v \hat \eta ) + o(\|\hat \eta\|_{\phi_1}) \,.
\]
and the result follows because $\|\cdot\|_{\phi_1}$ is a lattice norm. 
\end{proof}

\begin{proof}[Proof of Proposition \ref{p-freq}]
Existence and uniqueness of $\hat v$ wpa1 follows by Lemma \ref{l-perturb}. Then $\|\hat v - v\|_{\phi_1} = o_p(1)$ by Lemma \ref{l-consistent}. Moreover, the proof of Lemma \ref{l-consistent} shows
\begin{align*} 
 \|(\mb I-\mb D_{\hat \eta + v})^{-1} \|_{L^{\phi_1}} & \leq C_1 \,, &
 \|(\mb I-\mb D_{\hat \eta + \hat v})^{-1} \|_{L^{\phi_1}} & \leq C_2\,. 
\end{align*}
where $C_1$ and $C_2$ depend only on $\alpha$, $\beta$, $\|u\|_{\phi_r}$, $r$, and $M$ under AM2(a) and only on $\alpha$, $\beta$, $\|u\|_{\phi_r}$, $r$, $M$, $p$, and $\mb E^{Q_0 \otimes Q}[ \exp(q^2|\frac{\alpha}{1-\beta} u(X_t,X_{t+1})|^r)]$ under AM2(b) where $q$ is the dual index of $p$. By Lemma \ref{l-perturb-full} we therefore have $\| \hat v - v\|_{\phi_1} = O_p(1) \times \| \hat{\mb T} v - v\|_{\phi_1}$ and therefore $\|\hat v - v\|_{\phi_1} = O_p(a_n)$.

The result for $\| \frac{m_{\hat v}}{m_v} - 1\|_{p}$ now follows by a first-order Taylor-series expansion of the exponential function and continuity of the embedding $E^{\phi_1} \hookrightarrow L^p$ for each $1 \leq p < \infty$. The result for $\|m_{\hat v} - m_v\|_p$ also follows from $\|m_{\hat v} - m_v\|_p \leq \|m_v\|_{2p} \| \frac{m_{\hat v}}{m_v} - 1\|_{2p}$ and Lemma \ref{lem-lp-m}.
\end{proof}

\begin{proof}[Proof of Proposition \ref{p-bayes}]
For each $\hat Q \in \mc A_n$, $\mb T(\hat Q)$ has a unique fixed point $v(\hat Q) \in E^{\phi_r}$ with $\sup_{\hat Q \in \mc A_n} \|v(\hat Q) - v(Q)\|_{\phi_1} \leq C a_n$ for some finite positive constant $C$ (cf. Lemma \ref{l-consistent}). As $\Pi_n(\mc A_n) = 1 + o_p(1)$, we therefore have:
\[
 \Pi_n ( \{ \hat Q : \| v(\hat Q) - v(Q)\|_{\phi_1} > C_n a_n \} ) \leq \Pi_n ( \{ \hat Q : \underbrace{\| v(\hat Q) - v(Q)\|_{\phi_1}}_{\leq C a_n} > C_n a_n \} \cap \mc A_n ) + o_p(1) = o_p(1)\,.
\]
The contraction rates for $m_v$ follow similarly, as in the proof of Proposition \ref{p-freq}.
\end{proof}

\begin{proof}[Proof of Lemma \ref{lem-moe-reg}]
Note that  for any $Q \in \mc Q_K$, we have $Q \lll \hat Q \lll Q$, $X$ is stationary under $\hat Q$, and $Q_0 \ll \hat Q_0 \ll Q_0$. For the remainder of the proof, it's enough to show that Assumption AM2(b) holds for each $\hat Q \in \mc Q_K$ when $\mc Q_K$ is restricted as described in the Lemma. 

Let $\lambda_{\min}(\cdot)$ and $\lambda_{\max}(\cdot)$ denote smallest and largest eigenvalues. Let $c_\mu \in (0,\infty)$, $c_A \in (0,1)$, $\ul c_\lambda,\ol c_\lambda, \ul c_\lambda^{(2)},\ol c_\lambda^{(2)} \in (0,\infty)$ such that $\| \mu_k\| \leq c_\mu$, $|\lambda_{\max}(A_k)| \leq c_A$, $\ul c_{\lambda} \leq \lambda_{\min} (\Omega_k)  \leq \lambda_{\max}(\Omega_k) \leq \ol c_{\lambda}$ and $\ul c_{\lambda}^{(2)} \leq \lambda_{\min} (\Omega_k^{(2)})  \leq \lambda_{\max}(\Omega_k^{(2)}) \leq \ol c_{\lambda}^{(2)}$ holds for each $1 \leq k \leq K$ and each $\hat Q \in \mc Q_K$.

Let $q_0$ and $q_{01}$ denote the densities of $Q_0$ and $Q_0 \otimes Q$. Similarly, let $\hat q_0(x_t) = \sum_{k=1}^K w_k \, \phi( x_t ; \mu_k , \Omega_k )$ and $\hat q_{01}(x_t,x_{t+1}) = \sum_{k=1}^K w_k \, \phi( (x_t',x_{t+1}')' ; \mu_k^{(2)} , \Omega_k^{(2)} ) $ denote the densities of $\hat Q_0$ and $\hat Q_0 \otimes \hat Q$.

Consider $\hat q_0$. As $\ul c_\lambda I \leq \Omega_k \leq \ol c_\lambda I$ (where the inequalities should be understood in the sense of positive-definite matrices) and $\frac{1}{2} \|x\|^2 - c_\mu^2 \leq \|x - \mu_k\|^2 \leq 2 \|x\|^2 + 2 c_\mu^2$, we may deduce:
\[
 -\frac{d}{2} \log (2 \pi \ol c_\lambda) - \frac{c_\mu^2}{\ul c_\lambda} - \frac{1}{\ul c_\lambda} \|x_t\|^2 
 \leq \log \hat q_0(x_t) 
 \leq -\frac{d}{2} \log (2 \pi \ul c_\lambda) + \frac{c_\mu^2}{2 \ol c_\lambda} - \frac{1}{4 \ol c_\lambda} \|x_t\|^2 \,.
\]
Let $\ul c, \ol c, \ul s, \ol s \in (0,\infty)$ be such that $\ul c \exp(-\frac{1}{2\ul s^2} \|x\|^2 ) \leq q_0(x) \leq \ol c \exp(-\frac{1}{2\ol s^2} \|x\|^2)$. For any $\hat Q \in \mc Q_K$, we then have:
\begin{align*}
 \mb E^{Q_0}[ \hat \Delta(X_t)^{p_1} ] & \leq 
 \ul c^{1-p_1} \bigg( \frac{ \exp(c_\mu^2/(2 \ol c_\lambda))}{ (2 \pi \ul c_\lambda)^{d/2}}  \bigg)^{p_1} \int e^{ - \left( \frac{p_1}{4 \ol c_\lambda} - \frac{p_1-1}{2 \ul s^2} \right)  \|x\|^2} \, \mr d x
\end{align*} 
which is finite provided $\frac{\ul s^2}{2 \ol c_\lambda} > \frac{p_1-1}{p_1}$. Choosing $p_1$ so that this inequality holds, we then obtain  $\mb E^{Q_0}[\hat \Delta(X_t)^{p_1}] \leq M_1$ for some $M_1 < \infty$. Similarly, for any $\hat Q \in \mc Q_K$:
\begin{align*}
 \mb E^{Q_0}[ \hat \Delta(X_t)^{1-p_2} ] & \leq 
 \ol c^{p_2} \bigg( \frac{ \exp(-c_\mu^2/\ol c_\lambda)}{ (2 \pi \ol c_\lambda)^{d/2}}  \bigg)^{1-p_2} \int e^{ - \left( \frac{p_2}{2 \ol s^2} - \frac{p_2-1}{\ul c_\lambda} \right)  \|x\|^2} \, \mr d x
\end{align*} 
which is finite provided $\frac{\ul c_\lambda}{2\ol s^2} > \frac{p_2-1}{p_2}$. Choosing $p_2$ so that this inequality holds, we then obtain  $\mb E^{Q_0}[\hat \Delta(X_t)^{1-p_2}] \leq M_2$ for some $M_2 < \infty$.

Now consider $\hat q_{01}$. As $\ul c_\lambda^{(2)} I \leq \Omega_k^{(2)} \leq \ol c_\lambda^{(2)} I$ and $\frac{1}{2} \|x^{(2)}\|^2 - 2c_\mu^2 \leq \|x^{(2)} - \mu_k^{(2)}\|^2 \leq 2 \|x^{(2)}\|^2 + 4 c_\mu^2$, we may deduce by similar arguments that:
\[
 -\frac{d}{2} \log (2 \pi \ol c_\lambda^{(2)}) - \frac{2c_\mu^2}{\ul c_\lambda^{(2)}} - \frac{1}{\ul c_\lambda^{(2)}} \|x^{(2)}\|^2 
 \leq \log \hat q_{01}(x^{(2)}) 
 \leq -\frac{d}{2} \log (2 \pi \ul c_\lambda^{(2)}) + \frac{c_\mu^2}{ \ol c_\lambda^{(2)}} - \frac{1}{4 \ol c_\lambda^{(2)}} \|x^{(2)}\|^2 \,.
\]
Choose $\ul c^{(2)}, \ol c^{(2)}, \ul s^{(2)}, \ol s^{(2)} \in (0,\infty)$ such that 
\[
 \ul c^{(2)} \exp\left(-\frac{1}{2(\ul s^{(2)})^2} \|x^{(2)}\|^2 \right) \leq q_0(x^{(2)}) \leq \ol c^{(2)} \exp\left(-\frac{1}{2(\ol s^{(2)})^2} \|x^{(2)}\|^2 \right)\,.
\] 
For any $\hat Q \in \mc Q_K$, we then have:
\begin{align*}
 \mb E^{Q_0 \otimes Q}[ \hat \Delta_2(X_t,X_{t+1})^{p_3} ] & \leq (\ul c^{(2)})^{1-p_3} \bigg( \frac{ \exp(c_\mu^2/\ol c_\lambda^{(2)})}{ (2 \pi \ul c_\lambda^{(2)})^{d}}  \bigg)^{p_3} \int e^{ - \left( \frac{p_3}{4 \ol c_\lambda^{(2)}} - \frac{p_3-1}{2 (\ul s^{(2)})^2} \right)  \|x^{(2)}\|^2} \, \mr d x^{(2)} \,,
\end{align*} 
which is finite provided $(\ul s^{(2)})^2/(2 \ol c_\lambda^{(2)}) > \frac{p_3-1}{p_3}$. Choosing $p_3$ so that this inequality holds, we then obtain  $\mb E^{Q_0}[\hat \Delta_2(X_t,X_{t+1})^{p_3}] \leq M_3$ for some $M_3 < \infty$. 

To complete the proof, it remains to show that $\hat \ell \in L^{\phi_1}_2$. As $\hat \ell = \log(\hat q_{01}/q_{01}) + \log (q_0/\hat q_0)$, it suffices to show that $\log(\hat q_{01}/q_{01}) \in L^{\phi_1}_2$ and $\log (\hat q_0/q_0) \in L^{\phi_1}$. Here we have:
\begin{align*}
 \ul a_0(x_t) & := 
 -\frac{d}{2} \log (2 \pi \ol c_\lambda) - \frac{c_\mu^2}{\ol c_\lambda}  - \log \ol c  + \left( \frac{1}{2\ol s^2} - \frac{1}{\ul c_\lambda} \right)\|x_t\|^2 \\
 & \leq \log \frac{\hat q_0(x_t)}{q_0(x_t)} \\
 & \leq -\frac{d}{2} \log (2 \pi \ul c_\lambda) + \frac{c_\mu^2}{2 \ol c_\lambda} - \log \ul c + \left( \frac{1}{2\ul s^2}- \frac{1}{4 \ol c_\lambda} \right) \|x_t\|^2
 =: \ol a_0(x_t) \,.
\end{align*}
As $Q_0$ has Gaussian-like tails, $x \mapsto \|x\|^2 \in L^{\phi_1}$ and so $\ul a_0,\ol a_0 \in L^{\phi_1}$. As $\|\cdot\|_{\phi_1}$ is a lattice norm, it follows that $\log(\hat q_0/q_0) \in L^{\phi_1}$ with $\|\log(\hat q_0/q_0)\|_{\phi_1} \leq \|\ul a_0\|_{\phi_1} + \|\ol a_0\|_{\phi_1}$ for each $\hat Q \in \mc Q$. An identical argument shows  $\log(\hat q_{01}/q_{01}) \in L^{\phi_1}_2$ and delivers a uniform bound on its norm.
\end{proof}

\subsection{Proofs for Appendix \ref{ax:id:gen}}

\begin{proof}[Proof of Proposition \ref{p:exun}]
Existence: Consider the sequence $\ol v_n = \mb T^n \ol v$. This is a monotone sequence with $\ul v \leq \ldots \leq \ol v_{n+1} \leq \ol v_n \leq \ldots \leq \ol v$ with $\ul v,\ol v \in \mc E$. The sequence is therefore bounded in $\mc E$ and hence in $L^1 = L^1(\mu)$. It follows by Beppo Levi's Theorem \cite[Theorem I.7.1]{Malliavin} that there exists $v \in L^1$ such that $\lim_{n \to \infty} \ol v_n = v$ (almost everywhere) and $\lim_{n \to \infty} \| \ol v_n - v\|_{L^1(\mu)}$. 

To strengthen convergence in $\|\cdot\|_{L^1(\mu)}$ to convergence in $\|\cdot\|_\psi$, first observe that $\ul v \leq v \leq \ol v$ and hence $v \in \mc E$. To establish a contradiction, suppose that $\limsup_{n \to \infty} \| \ol v_n - v\|_\psi \geq 2\varepsilon$ for some $\varepsilon > 0$. Then:
\begin{equation}\label{e-idpf-norm}
 \limsup_{n \to \infty} \int \psi(|\ol v_n - v|/\varepsilon) \, \mr d \mu \geq 1\,.
\end{equation}
Note that $\{f_n\}_{n \in \mb N}$ with $f_n = \psi(|\ol v_n - v|/\varepsilon)$ is a monotone sequence of non-negative functions with $\limsup_{n \to \infty}f_n = 0$ (almost everywhere). Moreover, each $f_n \leq \psi((|\bar v| + |\ul v| + |v|)/\varepsilon)$ where $\int \psi((|\bar v| + |\ul v| + |v|)/\varepsilon) \, \mr d \mu < \infty$ for each $\varepsilon > 0$ because $\bar v$, $\ul v$ and $v$ all belong to $\mc E$. Therefore, by reverse Fatou:
\[
 \limsup_{n \to \infty} \int \psi(|\ol v_n - v|/\varepsilon) \, \mr d \mu \leq  \int \limsup_{n \to \infty} \psi(|\ol v_n - v|/\varepsilon) \, \mr d \mu = 0
\]
contradicting (\ref{e-idpf-norm}). Therefore $\|\ol v_n - v\| \to 0$. Finally:
\[
 \| \mb Tv - v\|_\psi \leq \|\mb Tv - \mb T \ol v_n\|_\psi + \| \mb T \ol v_n - v\|_\psi = \|\mb T v - \mb T \ol v_n\|_\psi + \| \ol v_{n+1} - v\|_\psi \to 0 
\]
by continuity of $\mb T$, hence $\mb T v = v$.

Uniqueness: To establish a contradiction, suppose that $\mb T$ has two distinct fixed points in $\mc E$, say $v$ and $v'$. By order-convexity of $\mb T$:
\[
 v = \mb T v \geq \mb T v' + \mb D_{v}(v - v') = v' + \mb D_{v}(v - v')
\]
which implies that 
\begin{equation} \label{e-idpf-ineq}
 (\mb I - \mb D_{v})(v - v') \geq 0 \,.
\end{equation}
As $\rho(\mb D_v;\mc E) < 1$, we have $(\mb I - \mb D_v)^{-1} = \sum_{i=0}^\infty (\mb D_v)^i$. The operator $\mb D_v$ is monotone and so $(\mb I - \mb D_v)^{-1}$ is also monotone. Applying $(\mb I - \mb D_v)^{-1}$ to both sides of equation (\ref{e-idpf-ineq}) yields  $v-v' \geq 0$.
A parallel argument yields $v' - v \geq 0$. Therefore, $v = v'$, a contradiction.
\end{proof}

\subsection{Proofs for Appendix \ref{ax:ddc}}

Define the operator $\mb S$ by $\mb S f = \mb T f - \log D - \gamma_{\mr{EM}}$. Thus,
\[
 \mb S f(x) = \log\left( \frac{1}{D} \sum_{d=1}^D e^{u_d(x) + \beta \mb E^M \left[ \left. f(X_{t+1}) \right| X_t=x,D_t = d \right]} \right) \,.
\]
It suffices to derive the existence and uniqueness results for $\mb S$ rather than $\mb T$ as their fixed points differ only by translation by a constant.

The operator $\mb S$ satisfies a subgradient inequality (cf. (\ref{e:subgrad})) with subgradient $\mb D_v$ given by
\begin{align} \label{e:D-ddc}
 \mb D_v f(x) = \beta \sum_{d=1}^D w_{d,v}(x) \mb E^M \left[ \left. f(X_{t+1}) \right| X_t=x,D_t = d \right] = \beta \mb E^{W_v} (\mb E^M f)(x) \,,
\end{align}
which is clearly monotone, and where
\begin{align} \label{e:weight-ddc}
 w_{d,v}(x) = \frac{ e^{u_d(x) + \beta \mb E^M \left[ \left. v(X_{t+1}) \right| X_t=x,D_t = d \right]} }{ \sum_{d'=1}^D e^{u_{d'}(x) + \beta \mb E^M \left[ \left. v(X_{t+1}) \right| X_t=x,D_t = d' \right]}  }
\end{align}
denotes the conditional choice probability of the agent choosing action $d$ in state $x$ if the agent's value function were $v$ and for $h : \mc X \times \{1,\ldots,D\} \to \mb R$ we define
\[
 \mb E^{W_v} h(x) = \sum_{d=1}^d w_{d,v}(x) h(x,d)\,.
\]
The operator $\mb S$ is monotone, which follows from monotonicity of conditional expectations and the exponential and logarithmic functions. Convexity of $\mb S$ also follows by H\"older's inequality (relative to the normalized discrete measure on $\{1,\ldots,D\}$).

\begin{lemma} \label{lem-T-prop-ddc}
$\mb S$ is a continuous, monotone and convex operator on $E^{\phi_s}$ for each $1 \leq s \leq r$.
\end{lemma}

\begin{proof}[Proof of Lemma \ref{lem-T-prop-ddc}]
 First, take any $f \in E^{\phi_s}$ and any $c \in (0,1]$. We have:
\begin{align*}
 \mb E^{Q_0}\left[ \exp(|\mb S f(X_t)/c|^s)\right]
 & = \mb E^{Q_0}\left[ \exp\left( \left| \frac{1}{c} \log\left( \frac{1}{D} \sum_{d=1}^D e^{u_d(X_t) + \beta \mb E^M \left[ \left. f(X_{t+1}) \right| X_t,D_t = d \right]} \right) \right|^s \right)\right] \\
 & \leq \mb E^{Q_0}\left[ \frac{1}{D} \sum_{d=1}^D e^{ \left|   \mb E^M \left[ \left. \frac{1}{c} \left( u_d(X_t) + \beta f(X_{t+1}) \right) \right| X_t,D_t = d \right]  \right|^s }\right] \\
 & \leq \mb E^{Q_0}\left[ \frac{1}{D} \sum_{d=1}^D  \mb E^M \left[ \left. e^{ \left| \frac{1}{c} \left( u_d(X_t) + \beta f(X_{t+1}) \right)  \right|^s } \right| X_t,D_t = d \right]  \right] \\
 & \leq \mb E^{Q_0}\left[ \frac{1}{D} \sum_{d=1}^D  \mb E^M \left[ \left. e^{ 2^{s-1}  \left| \frac{1}{c} u_d(X_t) \right|^s + 2^{s-1} \left| \frac{1}{c} \beta f(X_{t+1}) \right|^s } \right| X_t,D_t = d \right]  \right] \\
 & \leq \mb E^{Q_0}\left[ e^{ 2^{s-1} \sum_{d=1}^D \left|  \frac{1}{c} u_d(X_t) \right|^s } \frac{1}{D} \sum_{d=1}^D  \mb E^M \left[ \left. e^{ 2^{s-1}  \left|  \frac{1}{c} \beta f(X_{t+1}) \right|^s } \right| X_t,D_t = d \right]  \right] \\
 & = \mb E^{Q_0 \otimes Q}\left[ e^{ 2^{s-1} \sum_{d=1}^D  \left| \frac{1}{c} u_d(X_t) \right|^s + 2^{s-1}  \left|  \frac{1}{c} \beta f(X_{t+1}) \right|^s  }\right] 
\end{align*}
where the first and second inequalities are by Jensen's inequality and convexity of $x \mapsto e^{\frac{1}{c^s}|\log x|^s}$, the third is because $(a+b)^p \leq 2^{(p-1) \vee 0}(a^p + b^p)$ for $a,b \geq 0$, and the fourth is by the triangle inequality. The right-hand side is finite because $f \in E^{\phi_s}$ and $u_d \in E^{\phi_r}$ for each $1 \leq d \leq D$.

To verify continuity, take $f \in E^{\phi_s}$ and $g \in E^{\phi_s}$ with $\|g\|_{\phi_s} \leq 1$ and let $c = \|g\|_{\phi_s}$. Then: 
\[
 \mb S(f+g)(x) - \mb S f(x) = \log \bigg( \sum_{d=1}^D w_{d,f}(x) e^{\beta \mb E^M \left[ \left. g(X_{t+1}) \right| X_t=x,D_t = d \right]} \bigg) 
\]
where $w_{d,f}(x)$ is defined in equation (\ref{e:weight-ddc}). Therefore:
\begin{align*}
 \mb E^{Q_0}\left[ e^{|(\mb S (f+g)(X_t)- \mb S f(X_t))/(\beta c)|^s} \right] 
 & = \mb E^{Q_0}\left[ \exp\left( \left| \frac{1}{\beta c} \log \left( \sum_{d=1}^D w_{d,f}(X_t) e^{\beta \mb E^M \left[ \left. g(X_{t+1}) \right| X_t,D_t = d \right]} \right) \right|^s \right)\right] \\
 & \leq \mb E^{Q_0}\left[ \sum_{d=1}^D w_{d,f}(X_t) e^{ \left| \frac{1}{c} \mb E^M \left[ \left. g(X_{t+1}) \right| X_t,D_t = d \right] \right|^s } \right] \\
 & \leq \mb E^{Q_0}\left[ \sum_{d=1}^D e^{ \left| \frac{1}{c} \mb E^M \left[ \left. g(X_{t+1}) \right| X_t,D_t = d \right] \right|^s } \right] \\
 & \leq \mb E^{Q_0}\left[ \sum_{d=1}^D \mb E^M \left[ \left. e^{ \left| \frac{1}{c}  g(X_{t+1}) \right|^s } \right| X_t,D_t = d \right]  \right] \\
 & = D \times \mb E^{Q_0}\left[  e^{ \left| \frac{1}{c}  g(X_t) \right|^s }  \right] 
\end{align*}
where the first and third inequalities are by Jensen's inequality, the second is because $0 \leq w_{d,f}(x) \leq 1$, and the final line is by stationarity. It follows by $c = \|g\|_{\phi_s}$ and Lemma \ref{lem-pollard} that $\|\mb S(f + g) - \mb Sf\|_{\phi_s} \leq (2D -1) \beta \|g\|_{\phi_s}$, verifying continuity. 
\end{proof}

\begin{lemma} \label{lem-D-bdd-ddc}
Fix any $v \in L^{\phi_1}$. Then:  for all $s \geq 1$, the operator $\mb D_v$ in equation (\ref{e:D-ddc}) is a continuous linear operator on $E^{\phi_s}$ with $\rho(\mb D_v;E^{\phi_s}) < 1$.
\end{lemma}

\begin{proof}[Proof of Lemma \ref{lem-D-bdd-ddc}]
For any $f \geq 0$:
\begin{align}
  \mb E^{W_v} \mb E^M f(x) & = \sum_{d=1}^D w_{d,v}(x) \mb E^M[ f(X_{t+1})|X_t = x,D_t = d] \notag \\
  & \leq D \times \frac{1}{D} \sum_{d=1}^D \mb E^M[ f(X_{t+1})|X_t = x,D_t = d]  = D \times \mb E^Q f(x) \,.\label{e:Dbd-ddc}
\end{align}
Now by Jensen's inequality and (\ref{e:Dbd-ddc}), for any $f \in E^{\phi_s}$ and $c \in (0,1]$, we have:
\begin{align*}
 \mb E^{Q_0} \left[ \exp(|\mb D_v f(X_t)/(\beta c)|^s) \right] 
 & \leq D \times \mb E^{Q_0 \otimes Q} \left[  \exp(|f(X_{t+1})/c|^s) \right] \\
 & = D \times \mb E^{Q_0} \left[  \exp(|f(X_t)/c|^s) \right]
\end{align*}
which is finite because $f \in E^{\phi_s}$. Therefore, $\mb D_v : E^{\phi_s} \to E^{\phi_s}$. Taking $c = \|f\|_{\phi_s}$ in the above equation and applying Lemma \ref{lem-pollard} we deduce that $\|\mb D_v\|_{E^{\phi_s}} \leq \beta(2D-1)$.

Now take any $f \in E^{\phi_s}$ with $\|f\|_{\phi_s} = 1$ and choose $c \in (\beta,1)$. We have:
\begin{align*}
 \mb E^{Q_0} \left[ \exp(|(\mb D_v)^n f(X_t)/c^n|^s) \right] & = \mb E^{Q_0} \left[ \exp((\beta/c)^{sn} |( \mb E^{W_v} \mb E^M )^n f(X_t)|^s) \right] \\
 & \leq \mb E^{Q_0} \left[ \exp( |( \mb E^{W_v} \mb E^M )^n f(X_t)|^s) \right]^{(\beta/c)^{sn}} \\
 & \leq \mb E^{Q_0} \left[ ( \mb E^{W_v} \mb E^M )^n\exp( | f(X_t)|^s) \right]^{(\beta/c)^{sn}} \\
 & \leq \left( D^n \mb E^{Q_0 \otimes Q^n} \left[ \exp( | f(X_{t+n})|^s) \right] \right)^{(\beta/c)^{sn}} \\
 & \leq \left( 2 D^n \right)^{(\beta/c)^{sn}} 
\end{align*}
by two applications of Jensen's inequality, inequality (\ref{e:Dbd-ddc}), and the fact that $\|f\|_{\phi_s} = 1$. Therefore:
\[
 \| \mb D_v^n \|_{E^{\phi_s}} \leq \left( \left( \left( 2 D^n \right)^{(\beta/c)^{sn}}  -1  \right) \vee 1 \right) c^n \,.
\]
Finally, as $0 < \beta < c < 1$, we have $0 < \beta/c < 1$ and so
\[
 \rho(\mb D_v ; E^{\phi_s} ) = \lim_{n \to \infty} \left( \| \mb D_v^n \|_{E^{\phi_s}} \right)^{1/n} \leq c \times \lim_{n \to \infty}  \left( \left( \left( 2 D^n \right)^{(\beta/c)^{sn}}  -1  \right) \vee 1 \right)^{1/n} = c < 1
\]
as required.
\end{proof}

\begin{proof}[Proof of Theorem \ref{t-id-W-ddc}]
We prove the result by applying Proposition \ref{p:exun}. Lemma \ref{lem-T-prop-ddc} establishes continuity and monotonicity of $\mb S$. Let $U(x) = \frac{1}{D} \sum_{d=1}^D \exp(\frac{u_d(x)}{1-\beta})$ and define 
\[
 \bar v(x) = (1-\beta) \sum_{n=0}^\infty \beta^n \log \left( (\mb E^Q)^n U (x) \right)  \,.
\]
Using the condition $u_d \in E^{\phi_r}$ for each $d$, we may deduce that $\bar v \in E^{\phi_r}$ also. To see that $\mb T \bar v \leq \bar v$, first note that by H\"older's inequality and Jensen's inequality:
\begin{align*}
 \mb S f(x) & \leq \log\left( \left( \frac{1}{D} \sum_{d=1}^D e^{\frac{u_d(x)}{1-\beta}} \right)^{1-\beta} \left( \frac{1}{D} \sum_{d=1}^D e^{ \mb E^M \left[ \left. f(X_{t+1}) \right| X_t = x,D_t = d \right]} \right)^\beta \right) \\
 & = (1-\beta) \log U(x) + \beta \log \left( \frac{1}{D} \sum_{d=1}^D e^{ \mb E^M \left[ \left. f(X_{t+1}) \right| X_t = x,D_t = d \right]} \right)  \\
 & \leq (1-\beta) \log U(x) + \beta \log \left( \mb E^Q \left[ \left. e^{  f(X_{t+1}) } \right| X_t = x \right] \right) \,.
\end{align*}
Substituting in the above expression for $\bar v$ and using a version of H\"older's inequality for infinite products (see, e.g., \cite{Karakostas}), we obtain:
\begin{align*}
 \mb S \bar v (x) & \leq (1-\beta) \log U(x) + \beta \log \left( \mb E^Q \left[ \left. e^{ (1-\beta) \sum_{n=0}^\infty \beta^n \log \left( (\mb E^Q)^n U (X_{t+1}) \right)  } \right| X_t = x \right] \right) \\
 & = (1-\beta) \log U(x) + \beta \log \left( \mb E^Q \left[ \left. \prod_{n=0}^\infty  \left( (\mb E^Q)^n U (X_{t+1}) \right)^{(1-\beta)\beta^n} \right| X_t = x \right] \right)   \\
 & \leq (1-\beta) \log U(x) + \beta \log \left(  \prod_{n=0}^\infty  \mb E^Q \left[ \left.  (\mb E^Q)^n U (X_{t+1})  \right| X_t = x \right]^{(1-\beta)\beta^n} \right)  \\
 & =  (1-\beta) \log U(x) + (1-\beta) \sum_{n=0}^\infty \beta^{n+1} \log \left( (\mb E^Q)^{n+1} U(x) \right) = \bar v (x)
\end{align*}
as required. To see that the sequence $\mb S^n \bar v$ is bounded from below, observe that:
\begin{align*}
 \mb S \bar v(x) & \geq \frac{1}{D} \sum_{d=1}^D \left( u_d(x) + \beta \mb E^M \left[ \left. \bar v(X_{t+1}) \right| X_t=x,D_t = d \right] \right)  =: \ul u(x) + \beta \mb E^Q \bar v(x) 
\end{align*}
where $\ul u = \frac{1}{D} \sum_{d=1}^D  u_d \in E^{\phi_r}$. It follows by induction that $\mb S^n \bar v \geq \sum_{s=0}^{n-1} (\beta \mb E^Q)^s \ul u + (\beta \mb E^Q)^n \bar v$. As $\rho(\mb E^Q;E^{\phi_r}) = 1$, we may deduce: $\liminf_{n \to \infty} \mb S^n \bar v \geq (\mb I - \beta \mb E^Q)^{-1} \ul u \in E^{\phi_r}$. Applying part (i) of Proposition \ref{p:exun} establishes existence of a fixed point $v \in E^{\phi_r}$.

For uniqueness, Lemma \ref{lem-D-bdd-ddc} shows that $\mb D_v$ at any $v \in E^{\phi_s}$ is a continuous, monotone linear operator with $\rho(\mb D_v;E^{\phi_s}) < 1$. Uniqueness now follows by Proposition \ref{p:exun}(ii).
\end{proof}

\singlespacing

\putbib

\end{bibunit}

\end{document}